%% file: FinalNatPhysArXiv.tex
\documentclass[a4paper,twocolumn,superscriptaddress,notitlepage]{revtex4-2}
\usepackage{bbm}
\usepackage{amsmath}
\usepackage{amssymb}
\usepackage{amsthm}
\usepackage{amsfonts}
\usepackage{amsmath}
\usepackage{graphicx}
\usepackage{xcolor}
\usepackage{wrapfig}
\usepackage{csquotes}
\usepackage{braket}
\usepackage{setspace}
\usepackage{verbatim}
\usepackage[colorlinks=true,linkcolor=teal,citecolor=teal,urlcolor=teal]{hyperref}
\usepackage{mathtools}
\usepackage{booktabs}
\usepackage{braket}
\usepackage{enumitem}
\usepackage{pgfplots}

\usepackage{pdfpages}
\usepackage{pgffor}
\makeatletter
\AtBeginDocument{\let\LS@rot\@undefined}
\makeatother

\usepackage{mleftright}
\usepackage{times}
\usepackage{newtxtext}

\input{commands.tex}

\DeclarePairedDelimiterX{\infdivx}[2]{(}{)}{%
  #1\;\delimsize\|\;#2%
}

\newcommand{\mC}{\mathcal{C}}

\newcommand{\Par}{\ensuremath{\mathsf{PARITIES}}}
\newcommand{\poly}{\ensuremath{\operatorname{poly}}}

\definecolor{dgreen}{RGB}{0,100,0}

\begin{document}	
\title{Exponentially tighter bounds on limitations of quantum error mitigation}

\author{Yihui Quek}
\address{Dahlem Center for Complex Quantum Systems, Freie Universit{\"a}t Berlin, 14195 Berlin, Germany}
\affiliation{School of Engineering and Applied Sciences, Harvard University, Cambridge, MA 02318, USA}

\author{Daniel Stilck Fran\c{c}a}
\address{Department of Mathematical Sciences, University of Copenhagen, 2100 K{\o}benhavn, Denmark}
\address{Univ Lyon, Inria, ENS Lyon, UCBL, LIP, F-69342, Lyon Cedex 07, France.}
\address{Dahlem Center for Complex Quantum Systems, Freie Universit{\"a}t Berlin, 14195 Berlin, Germany}

\author{Sumeet Khatri}
\address{Dahlem Center for Complex Quantum Systems, Freie Universit{\"a}t Berlin, 14195 Berlin, Germany}

\author{Johannes Jakob Meyer}
\address{Dahlem Center for Complex Quantum Systems, Freie Universit{\"a}t Berlin, 14195 Berlin, Germany}

\author{Jens Eisert}
\address{Dahlem Center for Complex Quantum Systems, Freie Universit{\"a}t Berlin, 14195 Berlin, Germany}

\address{Helmholtz-Zentrum Berlin f{\"u}r Materialien und Energie, 14109 Berlin, Germany}

	\date{\today}

	\maketitle

{\bf Quantum error mitigation has been proposed as a means to combat unwanted and unavoidable errors in near-term quantum computing without the heavy resource overheads required by fault-tolerant schemes. Recently, error mitigation has been successfully applied to reduce noise in near-term applications. In this work, however, we identify strong limitations to the degree to which quantum noise can be effectively ‘undone’ for larger system sizes. Our framework rigorously captures large classes of error mitigation schemes in use today. By relating error mitigation to a statistical inference problem, we show that even at shallow circuit depths comparable to the current experiments, a superpolynomial number of samples is needed in the worst case to estimate the expectation values of noiseless observables, the principal task of error mitigation. Notably, our construction implies that scrambling due to noise can kick in at exponentially smaller depths than previously thought. They also impact other near-term applications, constraining kernel estimation in quantum machine learning, causing an earlier emergence of noise-induced barren plateaus in variational quantum algorithms and ruling out exponential quantum speed-ups in estimating expectation values in the presence of noise or preparing the ground state of a Hamiltonian.}
%While improvements in quantum hardware will push noise levels down, if error mitigation is used, ultimately this can only lead to an exponential time algorithm with a better exponent when compared with classical algorithms, putting up a strong obstruction to the hope for exponential quantum speedups in this setting.
%Our results imply that any hope of efficiently scaling error mitigation up to larger system sizes at $\textrm{poly log log(\textit{n})}$ depth.
%They also impact other near-term applications, constraining kernel estimation in quantum machine learning and causes an earlier emergence of noise-induced barren plateaus. %This is no detail: 
%What is at stake is no less than the question of to what extent there is scope for quantum computing before the advent of fault-tolerant quantum computers.}

%\JJM{Above a proposal for a shortened abstract.}

%\JJM{Nature Communications guidelines: up to 5000 words for the main text (currently ~8500) and 3000 for the methods, up to 10 figures and 70 references. This means we have to cut nearly half of the text!}
%\je{Nature Physics asks for 3000 words, in fact.}

%\subsection*{Introduction}

Quantum computers promise to efficiently solve some computational tasks out of reach even of classical supercomputers. As early as in the 
1980s~\cite{Feynman-1986},
%,QuantumTuring}, 
it was suspected that quantum devices may have computational capabilities that go substantially beyond those of classical computers. 
%\DSF{"This part can be squeezed"}
% True, let us see whether we need it
Shor's algorithm, presented in the mid 1990s, confirmed this suspicion by presenting an efficient quantum algorithm
for factoring, for which no efficient classical algorithm is known~\cite{Shor-1994}. 
Since then, the quantum computer has been a hugely inspiring theoretical idea. 
%It also served as a guiding principle in devising actual quantum devices. 
Soon, it became clear that unwanted interactions with the environment and hence the concomitant decoherence would be the major threat against
realizing quantum computers as actual physical devices. 
Early fears that decoherence could not be overcome in principle, fortunately,
were proven wrong. 
The field of quantum error correction presented ideas that show that 
%even if one cannot read out
%logical quantum information along the way without necessarily perturbing the very same quantum information, 
one can still correct for %errors, in fact 
arbitrary unknown errors~\cite{PhysRevA.52.R2493}. 
This key insight triggered a development that led to the blueprint of
what is called the fault-tolerant quantum computer~\cite{QEC2,Roads}, a (so far still fictitious) device that allows for arbitrary local errors and
can still maintain an arbitrarily long and complex quantum computation.
That said, known schemes for fault tolerance come along with demanding, possibly prohibitive overheads 
%In the most popular of fault tolerant schemes, one would 
%think of realizing surface codes, while computation would be performed by
%so-called lattice surgery and magic state distillation. 
%This means that 
%Logical qubits are encoded in a vastly larger number of physical qubits
~\cite{Roads}.
For the quantum devices that have been % experimentally developed
realized 
in recent years, %---extremely impressive devices after all the time of 
%quantum computers being primarily 
%objects of theoretical thoughts---
such prescriptions still seem 
%by far 
out of scope.

For this reason, \emph{quantum error mitigation} has gained traction recently~\cite{PhysRevX.7.021050,PhysRevLett.119.180509,PhysRevX.8.031027,MitigationReview} as a possible near-term surrogate for quantum error correction. The idea is to correct the effect of quantum noise on a near-term computation's result via classical post-processing on measurement outcomes, without mid-circuit measurements and adaptive gates, as is done in error correction. This minimizes the overheads in physical hardware.

While a compelling prospect, we ask: {To what extent can we really achieve such classical correction of quantum noise {\em post-hoc}?}

In this {work}, we argue that the current generation of error mitigation schemes must be 
{seriously reconsidered} in order to reach this goal. Empirically, some of them seem to come at a severe quantum resource penalty: At least one specific protocol, 
\emph{zero-noise extrapolation}, requires a number of samples scaling exponentially in the number of gates in the light-cone of the observable of interest, with the exponent depending on the noise levels~\cite{PhysRevX.8.031027}. A similar exponential scaling besets the technique of \emph{probabilistic error cancellation} under a sparse noise model~\cite{IBMexponential}. Our results contribute to a theoretical understanding of the conditions under which this happens.

We identify striking obstructions to error mitigation: even to mitigate noisy circuits slightly beyond constant depth requires a super-polynomial number of uses of those circuits in the worst case. These obstructions come from turning the lens of statistical learning theory onto the problem of error mitigation.
%they require exponential in 
%, using the tools of
%statistical learning theory
%we find they have fundamental limitations % of quantum error mitigation
%by proving lower bounds on their %its 
%worst-case complexity.
%Since this work is aimed at providing 
%rigorous bounds, we have to be 
%precise what we exactly mean
%by a quantum error mitigation protocol.
%We basically discuss two types of schemes.
%In the first type, 
We formulate this problem rigorously as one where the mitigation algorithm works with a classical description
of a noiseless circuit and specimens of the noisy circuit's output state, on which it can perform arbitrary measurements. 

We then distinguish two tasks: In the first (\emph{weak error mitigation}),
the goal is to output a collection of expectation values of the noiseless circuit's
output state. %We refer to this task as % family of
%algorithms as
%\emph{expectation-value-error 
%mitigation} or
%\emph{weak error mitigation}. 
This approach is taken to mitigate errors in variational
quantum circuits~\cite{PhysRevX.7.021050,McClean_2016},
an important family of circuits for near-term quantum devices.
In other situations (\emph{strong error mitigation}), the goal is to output a sample from the clean output state when measured in the computational basis. 
%We refer to this task as % family of methods as
%\emph{sample-error mitigation}
%or \emph{strong error mitigation}. 
This approach is taken to mitigate errors in algorithms for hard combinatorial optimization problems, a class that includes the famous \emph{quantum approximate
optimization algorithm} (QAOA)~\cite{QAOA}. Our framework and results encompass many  error-mitigation protocols used in practice, including  \emph{virtual distillation}~\cite{huggins2021virtual,koczor2021exponential}, \emph{Clifford data regression} (CDR)~\cite{czarnik2021error}, \emph{zero-noise extrapolation} (ZNE)~\cite{PhysRevLett.119.180509} and \emph{probabilistic error cancellation} (PEC) \cite{PhysRevLett.119.180509}. %PhysRevX.8.031027}.

%\ynote{KIV: For both families of methods, we prove
%rigorous limitations.} While the language and framework used
%is that of mathematical physics, to come up with precise
%statements, our formalism also captures practically minded algorithms as they are actually 
%being used in today's laboratories around the world.
%Related work

Novel error mitigation schemes are being intensely developed, even as we write~\cite{PhysRevX.7.021050,PhysRevLett.119.180509,PhysRevX.8.031027}. Unsurprisingly, given the high expectations for such techniques, their limitations have also been studied previously.
%and so has the study of the limitations of these techniques. 
{In particular, we build on the key work
\cite{TEMG21}
%,TTG22} 
that has first identified limitations to quantum error mitigation by studying how noise affects the distinguishability of quantum states.} There as well as in Refs.~\cite{TTG22, Fisher22} concurrent to our work, the authors
take an information-theoretic approach to study the sample complexity of weak error mitigation under depolarizing noise {(more generally, Pauli channels and thermal noise applied to the identity circuit) and show that it scales exponentially, but only in the depth of the circuit. This is not a limitation when depth is $O(\log n)$, as is the case in practice.} 
%Similar to us, their bounds arise from considering how noise affects the distinguishability of quantum states.

{In fact, the quantum community has known for some time that quantum states being manipulated by noisy quantum circuits undergoing depolarizing noise exponentially
quickly (in depth) converge to maximally mixed states \cite{PRXQuantum.3.040329,MullerHermes,FG20}. This behavior
suggests that quantum advantage is lost once circuits exceed logarithmic depth (if no error correction is used). 
Strikingly, in contrast to this natural expectation, our limitations kick in for circuits that come much closer to
constant-depth. This is because we introduce a novel dependence of the sample cost on the \emph{width} of the circuits ($n$). We are also, to our knowledge, the first work to analyze error mitigation under non-unital noise, a highly physically relevant class of noise that includes $T_1$ decay, the primary source of noise in superconducting qubit architectures. In addition to the mathematical tools we originate for analyzing noisy computation, we make a strong conceptual point: error mitigation schemes in practice, empirically found to have $\exp(n)$ sample complexity, are pretty much as good as they can get -- unless they are crafted to address circuits with a special structure that eludes our bounds.} 

While quantum error mitigation has seen practical success on small noisy quantum devices~\cite{MitigationReview}, our results put hard limits on how much we can expect error mitigation to scale with the size of such devices. We have demonstrated {\em worst-case} circuits that must be run exponentially-many times for error mitigation on them to work. Lower bounds shine light on upper bounds: our work suggests that future sample-efficient error mitigation schemes must dodge the limitations we have identified. How then can we design circuits to be more resilient to noise, and how far into the feasible regime can we push the line between error mitigation and error correction?
%Could other models of quantum computation, say dissipative preparation~\cite{verstraete2009quantum}, offer good and more robust alternatives to the circuit model in the NISQ era? 
%Our work also shines some preliminary light on the little-studied domain of strong error-mitigation protocols, which output samples (as opposed to expectation values) from the noiseless circuit, and invites proposals of how to do so efficiently.

%\JJM{I propose to keep the above introduction, but to just relentlessly condense it. Also I would put the related works here and shorten it.}
%\je{Indeed, here is where we cite off past work. We have no space in the main text for a long discussion of previous work. We can keep it in the appendix. But let us not shorten the introduction too much, this is important for Nature Physics.}

\subsection*{Introduction to the technique}
To establish the framework, we start off by defining what constitutes `error mitigation' in the rest of this work. In the literature, the term `error mitigation' has been used to describe protocols that are appended after a noisy quantum algorithm. Such protocols reduce the unwanted effect of noise on a quantum circuit by measuring it and classically post-processing the results (see Fig.~\ref{fig:error_mitigation_general}). %At its core, an error mitigation algorithm converts noisy quantum states to a classical representation of the noiseless quantum state -- 
They then output either samples from the would-be noiseless circuit or its expectation values on observables of interest, depending on the purpose of the original quantum algorithm. More concretely, we formulate the problem of error mitigation as follows.

%The appropriate representation varies: error mitigation is often appended after a quantum algorithm, to revert its outputs to their noiseless version, so the goal of error mitigation should be output the desired output of the quantum `outer loop' into which it is inserted (see Fig.~\ref{fig:error_mitigation_general}) 
%for a schematic depiction.

%Let $\mathcal{C}$ and $\mathcal{C}'$ denote the noiseless circuit and noisy circuit we wish to run, respectively, and let $\sigma$ and $\sigma'$ denote their output states. We formulate the problem of error mitigation of the noisy circuit as follows.

\begin{problem}[Error mitigation (informal)\label{problem:EM_intro}]
Upon input of:
\begin{enumerate}[itemsep=0mm,topsep=0mm,parsep=0mm]
    \item[(i)] a classical description of a noiseless circuit $\mC$ and a finite set $\mathcal{M} = \{O_i\}$ of observables;
    \item[(ii)] copies of the output state $\sigma'$ of the noisy circuit $\mC'$, and the ability to perform collective measurements,
\end{enumerate}
output either:
\begin{enumerate}[itemsep=0mm,topsep=0mm,parsep=0mm]
    \item[(i)] estimates of the expectation values $\Tr(O_i \sigma)$ for each $O_i \in \mathcal{M}$ (weak error mitigation), where $\sigma$ is the output state of the noiseless quantum circuit $\mC$; or
    \item[(ii)] samples from a probability distribution close to the distribution of $\sigma$ when measured in the computational basis (strong error mitigation).
\end{enumerate}
The number $m$ of copies of $\sigma'$ needed is the \emph{sample complexity} of (weak or strong) error mitigation.
\label{prob:EM_intro}
\end{problem}
%\begin{problem}[Error mitigation, informal definition]\label{prob:EM_intro}
%Given as input:
%\begin{itemize}
%    \item Classical descriptions of a noiseless circuit $\mC$, some noise map that transforms $\mC$ to $\mC'$ and a set $\mathcal{M} = \{O_i\}$ of observables.
%    \item Copies of $\sigma'$, and the ability to measure those copies.
%\end{itemize}
%Output either:
%\begin{itemize}
%    \item Estimates of the expectation values $\Tr(O_i \sigma)$ for each $O_i \in \mathcal{M}.$ We call this {\em weak} error mitigation; denote the number of copies of $\sigma'$ needed as $m_{\text EM, \text{weak} }$.
%    \item Samples from a probability distribution which is approximately that associated with measuring $\sigma$ in the computational basis, i.e., $x \sim \Tr(\ketbra{x}{x}\sigma)$. We call this {\em strong} error mitigation; denote the number of copies of $\sigma'$ needed as $m_{\text EM, \text{strong} }$.
%\end{itemize}
%\end{problem} 
{The bounds we prove apply even to error mitigation protocols that are given an {\em exact} description of the noise model affecting the circuit. Naturally, this also means our bounds apply to those protocols that know nothing about the noise and must learn it on-the-fly, as they are a subset of the protocols mentioned in the previous sentence.}
We refer to the Supplemental Material for rigorous definitions of error mitigation algorithms (see Definitions~\ref{def:Weak_EMalg} and~\ref{def:Strong_EMalg} therein) and for an explanation of how these definitions connect to well-known protocols such as zero-noise extrapolation and probabilistic error cancellation (see Section~I %~\ref{suppmat:protocols} 
therein). These protocols, for instance, run $\mC$ with varying levels of noise, apply simple quantum post-processing to $\mC$, or modify the circuit they run based on the output of intermediate measurements. {This level of abstraction of error mitigation protocols builds on those of Refs.~\cite{TEMG21,TTG22,Fisher22}}.
% Corrected this.

%our analysis that bounds the sample complexity of Problem \ref{prob:EM_intro}, can be extended even to ostensibly more complicated error mitigation protocols (such as zero-noise extrapolation and probabilistic error cancellation)  

% We also consider error mitigation protocols with other types of outputs. An error mitigation algorithm is often appended after a quantum algorithm, to revert its outputs to their noiseless version. Hence the desired output of error mitigation depends on the desired output of the quantum `outer loop' into which it is inserted. In general, an error mitigation algorithm can have two types of outputs:
% \begin{itemize}
%     \item Expectation values of a set $\mathcal{M}$ of observables, i.e. the values $\Tr\left[O_i \sigma \right]$ for $O_i \in \mathcal{M}$. We call this {\em weak} error mitigation. This is made rigorous in Definition \ref{def:Weak_EMalg}.
%     \item Samples from a probability distribution which is approximately that associated with measuring $\sigma$ in the computational basis, i.e. $x \sim \Tr[\ketbra{x}{x}\sigma]$. We call this {\em strong} error mitigation. This is made rigorous (for two notions of the word `approximately') in Definitions \ref{def:Strong_EMalg},\ref{def:Robust_EMalg}.
% \end{itemize}

In this work, we reduce a {\em statistical inference} problem on noisy quantum states to error mitigation on the circuits that produced them. 
%---where `statistical inference' is an umbrella term for problems that deal with identifying underlying parameters of an underlying distribution based on (possibly corrupted) samples. 
%In this view, the underlying distribution in error mitigation comes from the noiseless state, and the samples are those from measuring the noisy state.
What motivates this perspective is the observation that a `good' error mitigation algorithm should act as an effective denoiser, allowing one to distinguish one state from another, even if the states can only be accessed by measuring their noisy versions. Seen from this angle, error mitigation solves the following problem:
%In fact, consider
%us to consider 

\begin{problem}[Noisy state discrimination (informal)]
Upon input of:
\begin{enumerate}[itemsep=0mm,parsep=0mm,topsep=0mm]
    \item[(i)] classical descriptions of a set $S = \{\rho_0, \rho_1, \ldots,  \rho_N\}$ of $n$-qubit quantum states, a noiseless circuit $\mC$; and
    \item[(ii)] $m$ copies of a state $\C'(\rho_i)$ (where $\C'$ is the noisy version of $\C$), with $i \in \{0,\ldots,  N\}$ unknown, and the ability to perform collective measurements,
\end{enumerate}
output $\hat{i} \in \{0,\ldots , N\}$ such that $\hat{i} = i$ (`success') with high probability.
\label{prob:state_intro}
\end{problem}

%A central observation in this work is that noisy state discrimination for a particular choice of $S$ and $\mC$ reduces to weak error mitigation. The reduction can be sketched as follows:
%
%\begin{proof}[Proof Sketch] The bound on the probability of success originates from Ref.~\cite{Tsybakov_2009} (Corollary 2.6). We refer the reader to Lemma \ref{lem:fano} for an extended statement. 
%
%We now prove the reduction, i.e., Inequality~\eqref{eq:reduction_intro}, by exhibiting the aforementioned set of instances. 
Problems~\ref{prob:EM_intro} and \ref{prob:state_intro} are intimately related. To see this, consider, in Problem~\ref{prob:state_intro}, discriminating between the maximally mixed state $\rho_N \coloneqq \mathbb{I}/2^n$ and the states $\rho_x \coloneqq \ketbra{x}{x}$ for $x\in \01^n$. This problem can be solved by performing weak error mitigation with the observables $\mathcal{M} = \{\mathcal{C}(Z_i)\}_{i\in [n]}$. This is because a weak error mitigation algorithm should output the estimates $\Tr(\mathcal{C}(Z_i)\mathcal{C}(\rho_x)) = \Tr(Z_i \rho_x)$. Now consider two cases: (I) if the unknown state had been $\mC'(\rho_x)$ for $x\in \01^n$, then $\Tr(\rho_x Z_i) = 2x_i-1$; (II) if instead the unknown state had been $\mC'(\rho_N)$, then $\Tr(\rho_N Z_i) = 0 \quad \forall i\in [n]$. The output of Problem \ref{prob:EM_intro} thus completely identifies the label of the state, $i$, thus solving Problem \ref{prob:state_intro}. The upshot is as follows: if {\em at least} $m$ copies of the noisy state must be used for successful discrimination in the above setting (solving Problem \ref{prob:state_intro}), then {\em at least} the same number of copies are needed for successful weak error mitigation (solving Problem \ref{prob:EM_intro}). We use the information-theoretic generalized Fano method for multiple hypothesis testing~\cite{Tsybakov_2009} to provide a lower bound on $m$. The crucial quantity to study, it turns out, is a quantum relative entropy
\begin{equation}\label{eq:relent_intro}
D(\mC'(\rho_x) \Vert \mathbb{I}/2^n).
\end{equation}
The inverse rate of decay of this quantity bounds the sample complexity of error mitigation. We use tools for precise control of quantum relative entropies of unitary 2-designs under unital and non-unital quantum noise in order to construct circuits $\mC'$ that yield our strong bounds on error mitigation.

\subsection*{Results}

With our information-theoretic perspective, we are able to establish the following fundamental limits on a broad swath of error mitigation protocols.

%\textbf{How many noisy copies of $\sigma'$ are required for mitigating local depolarizing noise?} 
\textbf{What is the sample complexity of error mitigation for local depolarizing noise?} Intuitively, one expects the sample complexity of error mitigation to scale in the size of the noiseless circuit $\mC$ and the amount of noise affecting it. Our Theorem~\ref{thm:2_intro} confirms this intuition starkly: it shows that the dependence of the resource requirement on these parameters is {\em exponentially} higher than previously known. Concretely,  Theorem~\ref{thm:2_intro} hinges on showing that the distance of a noisy circuit's output state to the maximally-mixed states goes below constant order at $\poly \log \log n$ depth, while the onset of this effect at the exponentially-larger $\log n$ depth pointed out in previous works, was already detrimental to many near-term applications~\cite{TEMG21,TTG22,Fisher22,Wasserstein_variational_22, LosAlamos_EM,LosAlamos_Kernels,FG20,Wasserstein_variational_22,Wang2021,Abhinav_21}. 

\begin{theorem}[Number of samples for mitigating depolarizing noise scales exponentially in the number of qubits and depth]\label{thm:2_intro}
% %\begin{theorem}[Mitigating depolarizing noise requires exponential-in-$n, D$ samples]\label{thm:2_intro}
% For $D> \Omega(\log^2(n))$, weak error mitigation on an $n$-qubit, $D$-layer quantum circuit affected by local depolarizing noise of parameter $p\in(0,1)$, requires at least $p^{-\Omega(n \,D)}$ copies of the circuit's output state.
Let $\mA$ be a weak error mitigation algorithm that mitigates the errors in an $n$-qubit, $D$-layer quantum circuit $\mC$ affected by local depolarizing noise $\mN$ of parameter $p$. For some parameter $s>0$ and depths $D \geq \Omega(\log^2(n/s))$, $\mA$ requires as input at least $s^{-1}p^{-\Omega(n \,D/s)}$ copies of the output state $\sigma'$ of the noisy quantum circuit $\mC'$.
\end{theorem}

We refer to Theorem 4
%~\ref{thm:overall_lightcones} 
in the Supplemental Material for the full, formal statement. In other words, setting $s=\mathcal{O}(1)$ tells us that error mitigation of even a circuit of poly-logarithmic depth demands exponentially-many (in qubit number) samples, in the worst case.  By picking, e.g., the parameter $s=\Omega(n/\log^2(n))$, we see that even at the low depth $D= \textrm{poly log log}(n)$, error mitigation already requires {\em super}-polynomially-many samples. Figure \ref{fig:error_mitigation_englament} illustrates the intuition behind our construction of circuits that saturate this bound: they are rapidly entangling, shifting weight onto high-Hamming weight Paulis, and this makes them particularly sensitive to noise. 

Note, however, that these circuits, which leverage a construction by Ref.~\cite{Cleve16}, require all-to-all connectivity to implement, so that they mix rapidly (i.e., in short depth). In 
Sec.~II.E.1
%\ref{sec:geo_local} 
of the Supplemental Material, we show analogous results for geometrically local circuits on a $d$-dimensional lattice. In that setting, we show that exponentially many samples are required for error mitigation at depths of $\mathcal{O}(n^{\frac{1}{d}}\textrm{poly log}(n))$ and already at depths $\tilde{\mathcal{O}}(\log^{\frac{2}{d}}(n)\textrm{poly log log}(n))$ it requires super-polynomially many samples. These effects kick in at depths $\tilde{\mathcal{O}}(n^{1/d})$, which is the minimal depth required to ensure that all qubits have a light-cone proportional to the system size. These constructions thus showcase a refined version of the basic intuition: rather than the size of the circuit, it is more precisely the number of gates in the light-cone of observables, that determines the difficulty of error mitigation.
\textbf{What is the sample complexity of error mitigation for non-unital noise?}
Our result in Theorem~\ref{thm:2_intro} relies critically on the structure of depolarizing noise---namely, the fact that it is a Pauli channel, and thus it is unital. Moving beyond this toy model, we are also able to show sample complexity bounds for circuits affected by non-unital noise, which is notably trickier to analyze. Here, the essence of our conclusion remains the same: whenever a family of circuits is highly entangling, which we model through the assumption that it forms a unitary 2-design, mitigating local non-unital noise typically requires a number of samples that is exponential both in the number of channel applications and number of qubits. 

%\begin{theorem}[Mitigating non-unital noise also requires exponential-in-$n, D$ samples]\label{thm:3_intro}
\begin{theorem}[Number of samples for mitigating non-unital noise scales exponentially in the number of qubits and depth]\label{thm:3_intro}
Weak error mitigation on an $n$-qubit, $D$-layer noisy quantum circuit of the form $\bigcirc_{t=1}^D \mathcal{N}_t\circ \mathcal{U}_t$, where the $\mathcal{U}_t$ are drawn independently from a unitary 2-design and $\mN_t=\mN^{\otimes n}$ is an $n$-fold tensor product of a qubit non-unital noise channel $\mathcal{N}$, requires as input at least $c^{-\Omega(n \,D)}$ of the circuit's output state, for some constant $c$ that depends on the noise $\mathcal{N}$.
%Let $\mA$ be a weak error mitigation algorithm that mitigates the errors in an $n$-qubit, $D$-layer quantum circuit $\mC$ affected by non-unital product noise $\mN=\otimes_{i=1}^n \mN_i$, such that the noisy circuit takes the form $T_D=\bigcirc_{t=1}^D \mathcal{N}\circ \mathcal{U}_t$, where $\mathcal{U}_t$ are drawn independently from a $2$-design. Then, in expectation, $\mA$ requires as input at least $c^{-\Omega(n \,D)}$ copies of the noisy circuit output $\sigma'_{\mC,\mN}$ for some constant $c$ that depends on the quantum channel $\mathcal{N}$.
\end{theorem}

We refer to Theorem~7
%\ref{thm:3} 
of the Supplemental Material for the full, formal statement. To prove it, we compute the expected overlap of two quantum states that are output by $D$ alternating layers of a unitary sampled from a unitary $2$-design followed by a non-unital noise channel. While this model of only applying local noise after a (global) unitary is simplified, to the best of our knowledge this is the first result to study error mitigation under non-unital noise beyond the setting of the variational algorithms studied in Ref.~\cite{Wasserstein_variational_22}. Our observation here that global unitaries make for fast decay to the maximally-mixed state under non-unital noise further strengthens the evidence for a connection between the entanglement generated by a circuit and the rate at which noise spreads.
%---one that already emerged in our analysis of depolarizing noise.
%: as long as the underlying noisy circuit is highly entangling and errors are not corrected, the cost of error mitigation is exponential in the depth and the number of qubits. 
% Admittedly, the model of only applying local noise after a (global) unitary is artificial; however, we believe that this result showcases that the conclusions we achieved under depolarizing are expected to transfer to the non-unital case: as long as the underlying noisy circuit is highly entangling, the cost of error mitigation is exponential in the depth and the number of qubits.

We prove Theorems~\ref{thm:2_intro} and \ref{thm:3_intro} in the setting of weak error mitigation, and we assume that the error mitigation algorithm does not use knowledge of the input state to the circuit. While this is the setting most often used in practice, it does not cover all proposals in the literature. 
We now ask if our results can be extended to close variants of this setting that are also practically relevant.
\begin{figure}
     \centering
 %    \begin{subfigure}[b]{0.5\textwidth}
         \centering

\includegraphics[width=1\columnwidth]{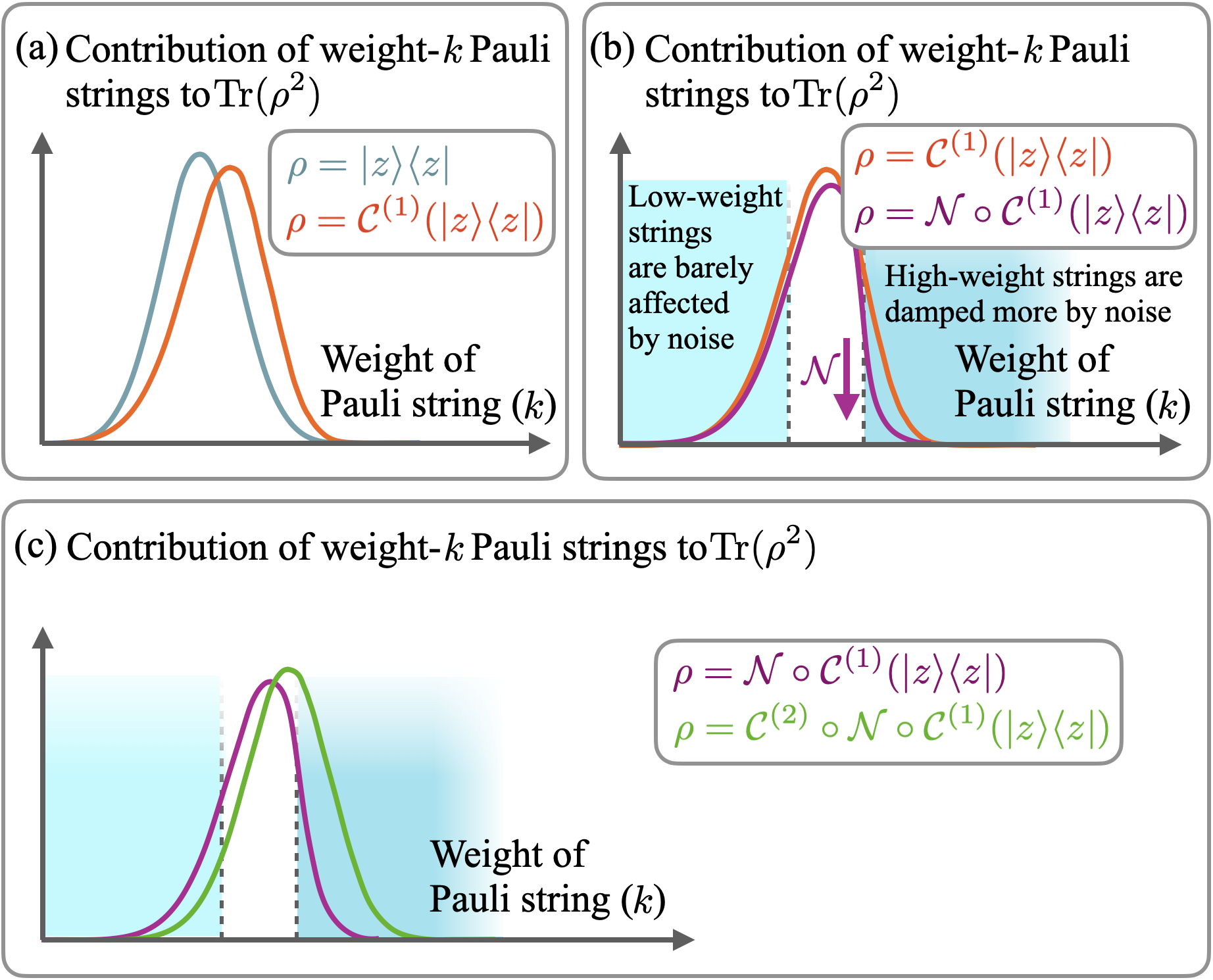}
    \caption{Intuition for our circuit construction: the higher the weight of a Pauli string, the more sensitive it is to Pauli noise, as showcased in Eq.~\eqref{eq:noisesensitivity} in the Supplementary Material.
    Whereas for product pure states there are Pauli strings with constant-order expectation values for all weights, most quantum states only have correlators of high weight. Thus, most random states are significantly more sensitive to noise than product quantum states. (a) The effect of applying a noiseless Pauli mixing circuit to a computational
    basis state is to shift the binomial of contributions to a weighted binomial (compare Eq.~\eqref{eq:purity_compstate} to Eq.~\eqref{eq:Panel1b}). (b) The effect of applying one subsequent layer of depolarizing noise on the output of the aforementioned circuit (compare Eq.~\eqref{eq:Panel1b} to Eq.~\eqref{eq:Panel2}). (c) The effect of applying yet another Pauli mixing layer to the state output by the aforementioned circuit (as captured in Eq.~\eqref{eq:Panel1b}).}
\label{fig:error_mitigation_englament}
  %   \end{subfigure}
 %    \hfill
 \end{figure}
\vspace{3mm}

\textbf{How does fixing the input state change the picture?}
The strong limitations we have proven so far are based on the assumption that the error mitigation algorithm does not make any use of information about the input state to the noisy circuit---hence that the error mitigation algorithm is \emph{input-state agnostic}. Does lifting this 
restriction---thereby including \emph{input-state aware} error-mitigation protocols---change the picture significantly? A first observation is that now, we can no longer rule out the possibility of successful error mitigation with sub-exponential worst-case sample complexity, simply because classical simulation with no resort to the noisy quantum device at all is a valid error mitigation algorithm with zero sample complexity---albeit possibly being computationally hard. Our no-go result in this case thus takes a different perspective: we show that an error mitigation algorithm that makes use of a noisy quantum device must use exponentially many copies of its noisy output state to produce an output meaningfully different from some fully-classical procedure which does not invoke the quantum device. 

\begin{theorem}[Resource cost of successful error mitigation]
A successful weak error mitigation algorithm $\calA$ must use $m = c^{\calO(n D)}$ copies of the noisy output state $\sigma'$ in the worst case, or there exists an equivalent algorithm $\calA'$ with purely classical inputs such that the output of $\calA'$ is indistinguishable from the output of $\calA$.
\end{theorem}

For the full, formal statement, see Theorem~2
%\ref{thm:minimum_number_of_samples_input_aware} 
in the Supplemental Material. In this way, we show that our constructions apply both to the input state-agnostic and input state-aware setting. %with, however, different implications.

\vspace{3mm}
\textbf{Does strong error mitigation imply weak error mitigation, or vice versa?} 
%We also study the relationship between the two types of error mitigation algorithms we have defined (weak and strong). We ask if the output of one is sufficient to obtain the output of the other. 
%One direction of this question is easy to answer---classically, samples of a probability distribution suffice to estimate expectation values of bounded functions. The quantum generalization of this is that 
To our knowledge, our work is the first to define and study strong error mitigation. The relevance of this notion is that, in some cases, weak error mitigation does not quite achieve the algorithmic goal at hand, for instance, solving hard combinatorial optimization problems on a noisy quantum device. In that case, one is not only necessarily interested in the optimal value of the cost function (obtained through weak error mitigation), but also in an assignment that achieves that value (obtained through strong error mitigation). While it is not hard to see that for local observables, strong error mitigation implies weak error mitigation, we also ask the converse question: could there be an algorithm that takes in the output of weak error mitigation (expectation values) and ``bootstraps'' them into the output of strong error mitigation (samples)? We give a partial negative answer that rules out certain types of protocols.

\begin{theorem}[Exponentially-many observables (in the same eigenbasis) are needed to output samples from the same basis]\label{thm:WeaktoStrong_intro}
Suppose we have a weak error mitigation algorithm that, given a set $\mathcal{M}=\{O_i\}_i$ of observables with the same eigenbasis, outputs their expectation value estimates $\{\hat{o}_i\}$. In general, at least $\exp(n)$-many distinct $\hat{o}_i$s must be queried by any algorithm that takes as input the $\{\hat{o}_i\}$ and outputs samples from their eigenbasis.
\end{theorem}

We refer readers to Theorem~\ref{thm:WeaktoStrong} in the Supplemental Material for the full, formal statement. The proof of this statement takes the perspective of expectation value estimates as {\em statistical queries}~\cite{Reyzin}, and leverages existing results on optimality guarantees of hypothesis selection~\cite{Yatracos85} and on the statistical query hardness of learning parity functions~\cite{BFJKMR94}.

\subsection*{Consequences of our results}

\subsubsection{The interplay of entanglement and noisy computation}
At the heart of our technical results is \emph{entanglement}. The circuits we construct are built of highly-entangling gates that allow them to scramble quickly to the maximally-mixed state, showcasing the intuition that such circuits are exponentially more sensitive to noise than are general circuits. This opens up an intriguing research direction: how does the amount of entanglement generated by a quantum circuit relate to its noise sensitivity? 
%
%\subsubsection{Extending previous distinguishability bounds for unital noise to non-unital noise} In addition, all of the above references can only show exponential convergence under the assumption that the quantum circuit is affected by unital noise. In our work, we go one step further and establish results for a toy model of circuits affected by non-unital noise interspersed by global 2-designs in Section~\ref{sec:beyondunital}. We show that even when the noise is non-unital, we are still able to obtain a decay in distinguishability between different inputs that is exponential in both the depth and the number of qubits. Such an extension is important because many physically relevant noise models, such as amplitude damping, are not unital. Our toy model illustrates that as long as the circuit is entangling enough the decay in distinguishability will be exponential in number of qubits.
%
%
%Additionally, in order to prove Theorem \ref{thm:2_intro}, we do not just construct a single circuit, but a random {\em ensemble} of circuits. We are able to prove that, on expectation, their relative entropy from the maximally-mixed state displays the stated exponential decay. 
%
That studying entanglement could be key to defeating noise is illustrated by the contrast between our results and those of Ref.~\cite{Abhinav_21}. The latter showed that {\em on expectation}, the total variation distance between a noisy random quantum circuit output distribution $\mathcal{D}$ and the uniform distribution decays exponentially in depth only, that is
\begin{align}\label{equ:abhinav}
\mathbb{E}_{\mathcal B}[\lVert \mathcal{D} - U \rVert_{TV}] = \Omega(\exp(-CD)).
\end{align}
This result applies to the noisy circuit ensemble $\mathcal B$, created by applying uniformly random $2$-qubit Clifford gates at each step, while our circuit ensemble is more structured. Viewing the quantity on the left-hand-side of Eq.~\eqref{equ:abhinav} as analogous to our Eq.~\eqref{eq:relent_intro}, however, we see that our circuit ensemble also mixes faster than theirs (as the relative entropy of our ensemble decays exponentially fast in both $n$ and $D$). We believe this difference arises precisely from the fact that $\mathcal{B}$ is not as entangling as our ensemble. 
%and 
%\footnote{
{(We substantiate this with the following oversimplified summary of their proof: a constant fraction of Clifford gates acting on $2$ qubits are a product of $1$ qubit gates. Thus, at each step there is a constant probability $q>0$ that the gate we pick to act on qubit $1$ is product. As the gates at different layers are independent, at depth $D$ with probability at least $q^D$ the first qubit of the system will be in a product state with the rest. And in that case the effect of the noise on that qubit cannot depend on the system size $n$.)}. 
%\JJM{Why do we need to put the above in here? I don't see why we would need to re-explain the proof of another paper. We could just say that our results don't contradict them because they don't have circuits with dynamics and then argue a bit about how dynamics (i.e. making more entanglement along the way) is the crucial thing.}
%\DSF{I think that the idea was to explain why a standard random circuit construction does not work}
%\je{Good, but indeed, this may be a bit too much about Abhinav's paper here? Can we shift this to the methods section part where we discuss other settings? The trouble is, we do not discuss the Abhinav paper in the methods section as is.} 

\subsubsection{Loss of quantum advantage at log log \textit{n} depth}

%In the course of proving Theorem \ref{thm:2_intro}, we prove the existence of very rapidly mixing circuits: circuits whose output states converge exponentially fast, in both number of qubits ($n$) and circuit depth ($D$), to the maximally-mixed state. 
The proof technique of Theorem \ref{thm:2_intro}, to study the convergence of a noisy circuit's output state to the maximally-mixed state, has appeared in various forms before: fast convergence is disastrous for error mitigation \cite{TEMG21,TTG22,Fisher22,Wasserstein_variational_22, LosAlamos_EM}, kernel estimation \cite{LosAlamos_Kernels} in quantum machine learning, as well as the depth at which quantum advantage from sampling noisy random quantum circuits \cite{Abhinav_21} or variational quantum algorithms \cite{FG20,Wasserstein_variational_22,Wang2021} can obtained.
%this metric has been central to many results in the literature (see, e.g., Refs.~\cite{LosAlamos_EM, LosAlamos_Kernels, TEMG21, TTG22, Wasserstein_variational_22, Abhinav_21, FG20}), where 
However, these results share the common feature that they are only able to show an exponential-in-{\em depth} contraction of the circuit's output state towards the maximally-mixed state and only show that the distance to the maximally mixed state goes below constant order at logarithmic depth. Furthermore, one can show that such results are tight for trivial circuits consisting solely of the identity gate.

We point out that the aforementioned applications pertain to \emph{noisy intermediate-scale quantum} (NISQ) processors. Such processors run shallow quantum circuits, which means that their depth scales like $D=O(\log(n))$, so that the rate of convergence in the aforementioned results is inverse polynomial in $n$. In this depth regime, our circuits converge exponentially fast in $n$. The onset of the attendant effects is therefore exponentially earlier in our circuits: at {\em $\log \log (n)$} depth, increasing the number of qubits already brings about a super-polynomial drop in distinguishability from the maximally-mixed state. 

In addition, our results imply that there will be no exponential quantum speedups for estimating expectation values in the presence of noise. This is because there are classical algorithms that exhibit the same exponential scaling in complexity~\cite{PhysRevLett.126.210502,PhysRevA.99.062337} exhibited by our results. As such, both classical and quantum algorithms will have exponential complexity for this task; quantum can at best improve the exponent.

%Our results show that, typically, the number of samples from a noisy device required to estimate the expectation value of the clean circuit will scale exponentially in the number of gates in the light-cone of the observable and therefore rules out the possibility of exponential quantum advantages for this task on noisy devices, as there are classical algorithms that exhibit the same scaling in complexity~\cite{PhysRevLett.126.210502,PhysRevA.99.062337,Bravyi_2021}. As such, one can only hope for better exponents in the worst-case complexity.

\subsubsection{No noisy circuits for ground state preparation}
Our work also bears on a task originating from the quantum sphere: preparation of the ground state of a Hamiltonian.
When optimizing the energy of a Hamiltonian, the outputs of noisy quantum circuits concentrate strongly~\cite{Wasserstein_variational_22,FG20} which causes any possibility of quantum advantage to be lost at a depth $D$ for any circuit when the error probability satisfies $p=\Omega(D^{-1})$, an argument that follows a similar route as our work. Our contributions show that this generic bound can be loose in the worst case, as we show that the 2-R\'{e}nyi entropy can decrease exponentially faster. Thus, in the worst case and at sufficiently high depth $D$ polylogarithmic in $n$, unless the error probability is $p = \calO((nD)^{-1})$, the outputs of noisy circuits will concentrate around strings with energies that do not offer a quantum advantage. Unfortunately, it has been shown that for many important classes of local Hamiltonians~\cite{Wasserstein_variational_22,PhysRevLett.125.260505,whole_graph,8104078,nlts_conjecture},
any circuit that aims at preparing the ground state has to have at least logarithmic depth. For highly entangled ground states of non-local Hamiltonians, we expect even larger lower bounds on the depth at which geometrically local circuits can prepare such ground states, as entanglement has to be built between distant sites. The fast buildup of  entanglement, however, is also a key reason that our circuits display a system-size dependent decay to the maximally mixed state.
Together these statements imply that noisy variational ground state preparation is mostly a lose-lose proposition: Either the depth is insufficient to reach a good approximation of the ground state or it is so high that noise takes over.

\subsection*{Outlook}\label{sec:outlook}

We have established a general rigorous framework that encapsulates large classes of schemes for quantum error mitigation that are already being used on today's noisy quantum processors. For these schemes, we have identified severe information-theoretic limitations that are exponentially tighter than what was previously known {\cite{TEMG21,TTG22,Fisher22}, due to the additional dependence we have now identified of the sampling cost on the width of the
circuit}. 

{While our bounds are still worst-case bounds, they are worst-case within a class of quantum circuits that can be as shallow as $\poly \log \log(n)$ -- substantially closer to practice than the $d = \Omega(n)$ of previous bounds. To close this gap even further, one should aim to lift the requirement of global connectivity in our circuits. This also hints at a prescription for practitioners: choosing suitably local quantum circuits may allow for an improved performance of error mitigation schemes (although potentially coming at the price of a reduced expressivity). It would be intriguing to examine if this is why some error mitigation protocols run in the recent past have been found to fare rather well in practice.}

We reiterate that our results do not rule out quantum error mitigation to tackle suitably small noise levels in existing quantum architectures, but should be viewed as a clarion call for deepened understanding: what aspects of our {\em worst-case} construction of circuits carries over into a typical use case? {We believe the spread of entanglement has to play a role. Broadly speaking, we see two regimes emerge once we go to the Heisenberg picture: one where the evolution does not spread out the observables, or generate much entanglement. But in this regime, we expect certain simulation techniques like tensor networks to perform well too. The other extreme is the one where evolution does generate a lot of entanglement, and this is the regime our circuits operate in. Then error mitigation is likely to require exponential resources, as we have found.
What we have observed, then, is a tension at the heart of near-term quantum computing: while it is known that entanglement needs to spread extensively over the architecture so that one can hope for a quantum advantage, we see here (as well as in Ref.~\cite{CiracNoise}) that this spreading of entanglement can be an adversary: it also assists in noise spreading rapidly.}

%On a positive and practical note, our results provide actionable advice to circuit design to render noise less impactful: raise the possibility of tailoring the error mitigation protocol according to the architecture and type of gates in the circuit being run. 
%\ynote{MORE PHYSICS}
%While by no means is this work aimed at suggesting that quantum error mitigation is no powerful tool to deal with physical noise levels in near-term architectures. 
%At the same time, this work identifies substantial limitations that arise when one aims at scaling up such ideas to near-term quantum devices. 
%In fact, this work can be seen as a strong indication that 

%The development of quantum computers---as with any new technology---is a game of cat and mouse.

We strongly believe that the way forward for dealing with errors in quantum computation will involve new schemes that are intermediate between mere quantum error mitigation and resource expensive fault tolerant quantum computing. In the medium term, some ``parsimonious" form of quantum error correction involving limited amounts of quantum redundancy will presumably be necessary. How much is enough? Clarifying this point will be an important step in bringing quantum computers closer to reality.

\subsection*{Acknowledgments}
%{{We thank the reviewers for constructive comments.}
{This work has been supported by the BMBF (RealistiQ, MUNIQC-Atoms, PhoQuant), the BMWK (PlanQK, EniQmA), the DFG (CRC 183), the Einstein Foundation (Einstein Research Unit on Quantum Devices), the Quant-ERA (HQCC), {the European Research Council (DebuQC)}, and the Alexander von Humboldt Foundation. The research is also part of the Munich Quantum Valley (K-8), which is supported by the Bavarian state government with funds from the Hightech Agenda Bayern Plus. DSF acknowledges financial support from the VILLUM FONDEN via the QMATH Centre of Excellence (Grant no.~10059), the QuantERA ERA-NET Cofund in Quantum Technologies implemented within the European Union's Horizon 2020 Program (QuantAlgo) via the Innovation Fund Denmark and from the European Research Council (grant agreement no.~81876). DSF acknowledges that this work benefited from a government grant managed by the Agence Nationale de la Recherche under the Plan France 2030 with the reference ANR-22-PETQ-0007.}

\subsection*{Author contributions statement}
The following describes the different contributions of all authors of this work, using roles defined by the CRediT (Contributor Roles Taxonomy) project: \textbf{Y.Q.} Conceptualization, formal analysis, methodology, writing (original draft); \textbf{D.S.F.} Conceptualization, formal analysis, methodology, writing (original draft); \textbf{S.K.} Conceptualization, methodology, writing (review and editing); \textbf{J.J.M.} Conceptualization, methodology, writing (review and editing); \textbf{J.E.}: Conceptualization, methodology, writing (review and editing).

\newpage

\begin{figure*}
    \centering
\includegraphics[width=0.9\textwidth]{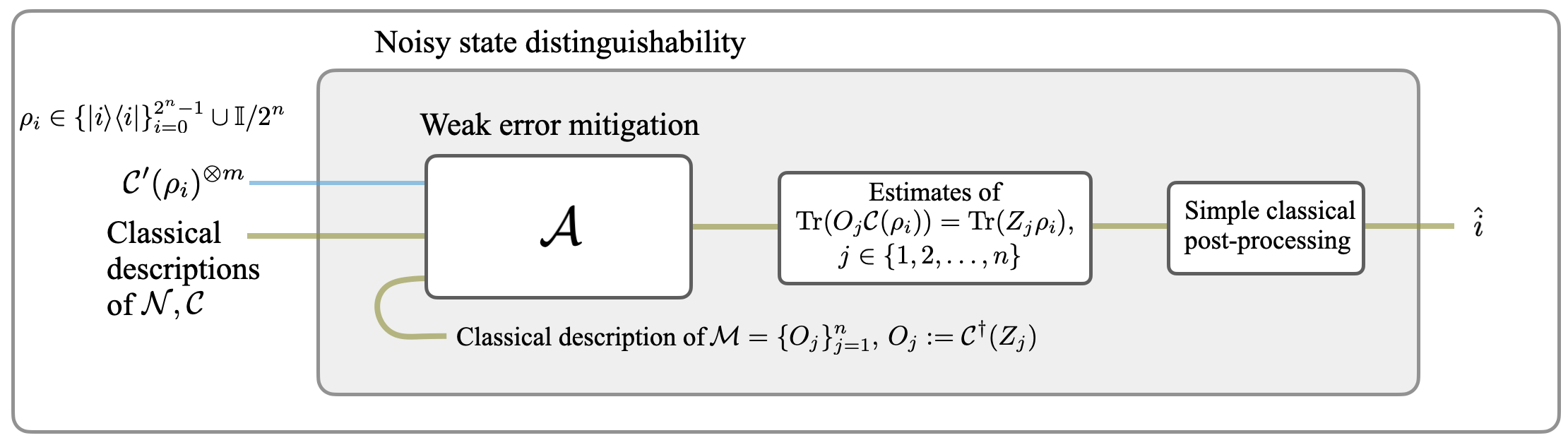}
    \caption{To lower-bound the sample complexity of weak error mitigation, we show that it can be used as a subroutine to solve a constructed problem of distinguishing states under noise.}
    \label{fig:EMoutline}
\end{figure*}

\begin{figure}
     \centering
 %    \begin{subfigure}[b]{0.5\textwidth}
         \centering
\includegraphics[width=1\columnwidth]{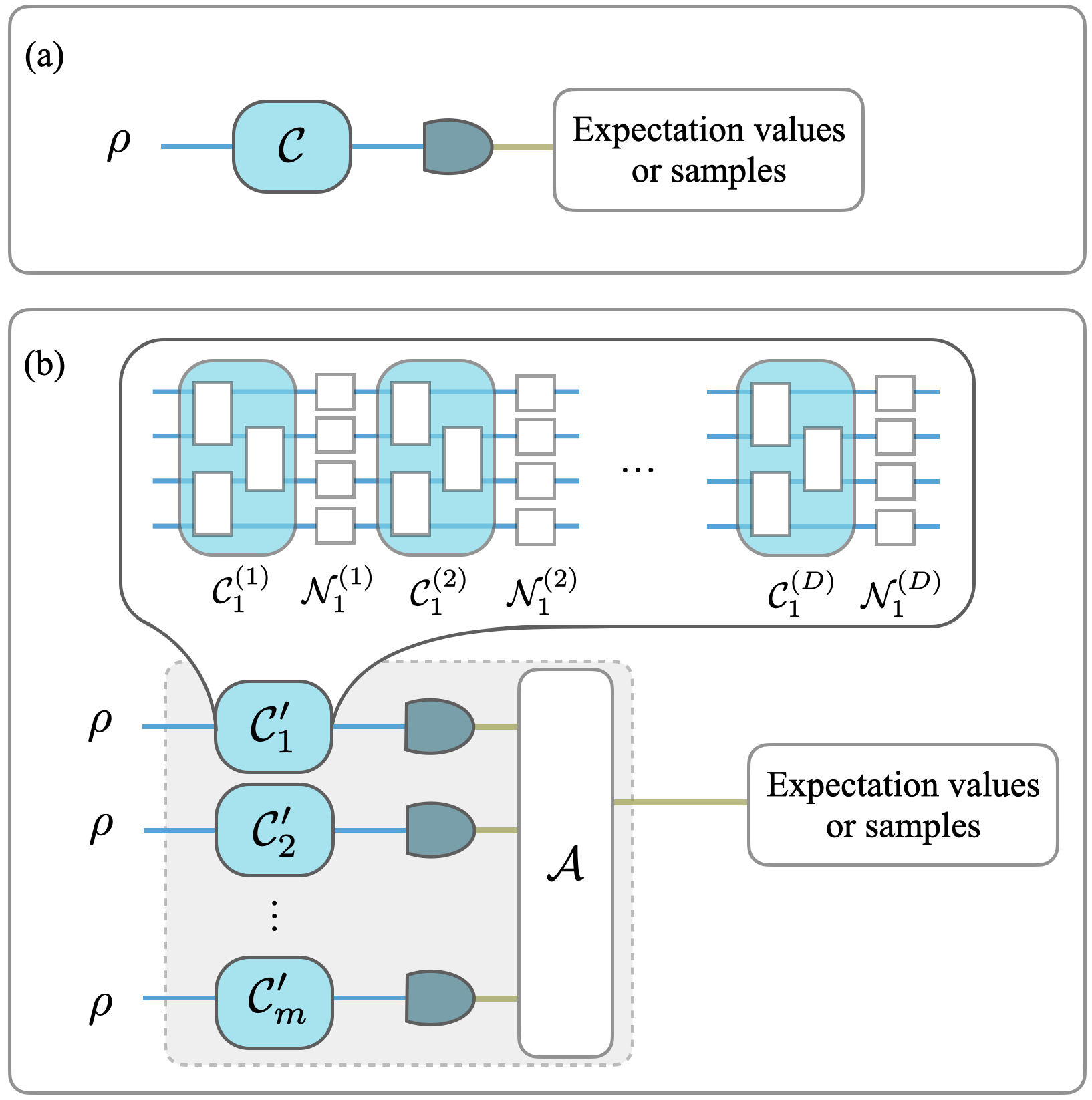}

         \caption{(a) Idealization of near-term quantum algorithms without quantum noise. Most such algorithms work by running an $n$-qubit quantum circuit $\mC$ on an input quantum state $\rho$, measuring the output state, and then returning either samples from the resulting probability distribution or expectation values of specified observables. (b) The model of error mitigation used in this work, {building on the framework established in Ref.~\cite{TEMG21}}. The quantum channel $\mC'_i$ represents the $i^{\text{th}}$ run of $\mC'$, the noisy version of $\mC$. We model the noise acting on $\mC$ by interleaving its layers with layers of a given noise channel. %Our model encompasses those error mitigation protocols that modify $\mC$ on each iteration. 
         In Section~I %~\ref{suppmat:protocols} 
         of the Supplemental Material, we show that this model applies to practical error mitigation protocols such as \emph{virtual distillation}~\cite{huggins2021virtual,koczor2021exponential}, \emph{Clifford data regression} (CDR)~\cite{czarnik2021error}, \emph{zero-noise extrapolation} (ZNE)~\cite{PhysRevLett.119.180509} and \emph{probabilistic error cancellation} (PEC) \cite{PhysRevLett.119.180509,endo2018practical}. In this work, we study how $m$, the number of noisy circuit runs, scales with $n$ and $D$ to reliably recover the expectation values. }\label{fig:error_mitigation_general}
  %   \end{subfigure}
 %    \hfill
 \end{figure}

\newpage

{\small 
\section*{Methods}

In this Methods section, we present substantial details of the arguments of the main text. In Section~\ref{sec:preliminaries}, we provide some background, defining relative entropies, elements of the Pauli group, and unitary 2-designs that we need for our proofs. Then, in Section~\ref{sec:defs}, we introduce our error mitigation setting as well as two variants of error mitigation that we call weak and strong respectively. In Supplementary Material Section I,
%\ref{suppmat:protocols}, 
we further argue that this setting encompasses many error-mitigation protocols used in practice, including  \emph{virtual distillation}~\cite{huggins2021virtual,koczor2021exponential}, \emph{Clifford data regression} (CDR)~\cite{czarnik2021error}, \emph{zero-noise extrapolation} (ZNE)~\cite{PhysRevLett.119.180509} and \emph{probabilistic error cancellation} (PEC) \cite{PhysRevLett.119.180509, PhysRevX.8.031027}. While most theoretical analyses of error mitigation have focused exclusively on the weak error mitigation setting, here we argue with reference to the practical protocols that the strong setting is equally relevant, and we make a first attempt at relating the two settings. We also show that many existing error mitigation protocols fit into the model that we will work with.

\section{Preliminaries: Relative entropy and derived quantities}\label{sec:preliminaries}

At the heart of many of our technical arguments are
notions of relative entropies---measures of distance between two quantum states---and how they contract under noise. %
We first define distance measures on classical probability distributions. Let $R,S$ be two probability measures on the same support $\mathcal{X}$. It will suffice for our purposes to let $\mathcal{X}$ be a finite set. 
%\SK{We are already using the uppercase $P$ and $Q$ for Pauli operators, and the lowercase $p$ is being used to denote the depolarizing probability. So I think we need a different notation for probability distributions.}\ynote{resolved}
\begin{itemize}[leftmargin=*]
    \item \textbf{Total variation (TV) distance.}
    This is defined as
    \begin{equation}\label{eq:TV}
        d_{TV}(R,S)\coloneqq \sup_{A \subseteq \mathcal{X}} |R(A)-S(A)| = \frac{1}{2} \sum_{x\in \mathcal{X}} |R(x)-S(x)|.
    \end{equation}
    \item \textbf{Kullback-Leibler (KL) divergence, or classical relative entropy.}
    The classical relative entropy is defined as
   \begin{equation}
    D(R || S) \coloneqq \sum_{x \in \mathcal{X}} R(x) \log \frac{R(x)}{S(x)} = \mathbb{E}_{x\sim R} \left[\log \frac{R(x)}{S(x)}\right],
    \end{equation}
    where throughout the manuscript we take $\log$ to be base $2$.
\end{itemize}
%\JJM{I would use \texttt{\textbackslash begin\{itemize\}[leftmargin=*]} here to remove the indent, it is weird if the equation starts further left than the item. Same for the other itemizes later.}

Now we introduce distance measures on quantum states. The primary such measure we will consider is quantum relative entropy, which can be understood as a quantum generalization of classical {\em KL divergence}. For this reason, we will use the same notation, $D(\cdot \| \cdot)$, for both quantum relative entropy and classical relative entropy.

Let $\rho, \sigma$ be two quantum states in $D(\mathcal{H}_n)$ (though in general these quantities are defined with $\sigma$ any positive semi-definite operator). We will use the following divergences:

\begin{itemize}[leftmargin=*]
\item \textbf{Relative entropy.}
If $\text{supp} (\rho) \subseteq  
\text{supp} (\sigma)$, 
we define the (quantum) relative entropy
as
\begin{equation}
    D(\rho \|\sigma)\coloneqq
    \Tr(\rho \log\rho)-
    \Tr(\rho \log\sigma).
\end{equation}
\item \textbf{Petz-R\'{e}nyi relative entropy.}
For a parameter $\alpha \in(0,1) \cup(1, \infty)$,
\begin{equation}
D_{\alpha}(\rho \| \sigma)\coloneqq\frac{1}{\alpha-1} \log \operatorname{Tr}\left[\rho^{\alpha} \sigma^{1-\alpha}\right] .
\end{equation}
For $\alpha\in(1,\infty)$, this definition holds in the case $\text{supp}(\rho)\subseteq\text{supp}(\sigma)$. In the limit $\alpha\to 1$, the Petz-R\'{e}nyi relative entropy reduces to the quantum relative entropy, i.e., $\lim_{\alpha\to 1}D_{\alpha}(\rho\Vert\sigma)=D(\rho\Vert\sigma)$.

\item \textbf{Max-relative entropy.} For $\text{supp}(\rho)\subseteq\text{supp}(\sigma)$,
\begin{equation}
    D_{\max }(\rho\Vert \sigma)\coloneqq\inf \left\{\gamma: \rho \leq 2^\gamma \sigma\right\}.
\end{equation}
It holds that $D(\rho\Vert\sigma)\leq D_{\max}(\rho\Vert\sigma)$~\cite{MosonyiHiai}.
\end{itemize}

It will often be illuminating to fix the second argument of the relative entropy to be the maximally mixed state, while putting the state of interest in the first argument of the relative entropy. Relative entropies of this form may even be upper-bounded explicitly in terms of quantities relating to the state of interest, as we now show the following.

\begin{lemma}[Purity controls relative entropy to the maximally mixed state]\label{lem:relent}
For any $\rho$ on $n$ qubits,
\begin{equation}
    D\left(\rho \Big\|  \frac{\mathbb{I}}{2^n}\right) \leq n + \log(\Tr(\rho^2)).
\end{equation}
\end{lemma}
\begin{proof}
The statement follows from the fact that relative entropy is upper-bounded by 2-R\'{e}nyi relative entropy
\begin{equation}
D_2(\rho\|\sigma) \coloneqq \log \Tr\left[\rho^{2} \sigma^{-1}\right].
\end{equation} 
This can be seen from the fact that the Petz-R\'{e}nyi relative entropies for $\alpha \in (0,1) \cup (1,\infty)$ satisfy an ordering property \cite{MosonyiHiai}, i.e., 
\begin{equation}
\text{for }\alpha>\beta>0:\quad D_{\alpha}(\rho \big\| \sigma) \geq D_{\beta}(\rho \big\| \sigma),
\end{equation}
and
\begin{equation}
\lim_{\alpha \rightarrow 1} D_{\alpha}(\rho \big\| \sigma) = D(\rho \| \sigma).
\end{equation}
For any $\rho$ on $n$ qubits and $\sigma$ the maximally mixed state, we thus have
\begin{eqnarray}\label{eq:boundD}
D\left(\rho \Big\|  \frac{\mathbb{I}}{2^n}\right) &\leq& D_2\left(\rho \| \frac{\mathbb{I}}{2^n}\right) = \log \Tr\left(\rho^2 \left(\frac{\mathbb{I}}{2^n}\right)^{-1}\right) \\
&=&  n + \log(\Tr(\rho^2))\nonumber
\end{eqnarray}
as stated.
\end{proof}

%
%\ynote{We may also be able to use the following slightly better statement which saves $\log(n)$ depth if we can deal with non-Clifford gates for the noise. Wait for an update!}
%
%Alternatively, with $O(n \log n \log \log n)$ non-Clifford gates that compose to a Clifford unitary and $\tOrd{n}$ auxiliary qubits,  in $O(\log(n))$ depth, it is possible to implement an {\em exact} unitary 2-design on $n$ qubits. 

\subsection{Pauli operators and Pauli channels}

We denote the single-qubit Pauli operators by $\mathbb{I},X,Y,Z$. Furthermore, for $a=(a_1,a_2,\dotsc,a_n)\in\{0,1\}^n$, we let
\begin{align}
    Z^{a}&\coloneqq Z^{a_1}\otimes Z^{a_2}\otimes\dotsb\otimes Z^{a_n},\\
    X^{a}&\coloneqq X^{a_1}\otimes X^{a_2}\otimes\dotsb\otimes X^{a_n},
\end{align}
where
\begin{align}
    Z^{a_i}&\coloneqq \left\{\begin{array}{l l} \mathbb{I} & a_i=0, \\ Z & a_i=1, \end{array}\right\},\\
    X^{a_i}&\coloneqq \left\{\begin{array}{l l} \mathbb{I} & a_i=0, \\ X & a_i=1. \end{array}\right\}.
\end{align}
We then define the Pauli group.

\begin{defn}[Pauli group and Pauli weight]
Let $n\in\{1,2,\dotsc\}$. The Pauli group $\mathcal{P}_{n}$, by definition, consists of all operators of the form $i^{k} X^{a} Z^{b}$, where $k \in$ $\{0,1,2,3\}$ and $a, b \in\{0,1\}^{n}$. Let $\mathcal{Q}_{n}\coloneqq\mathcal{P}_{n} /\{\pm 1, \pm i\}$ be the quotient group that results from disregarding global phases in $\mathcal{P}_{n}$. 
%Finally let $\mathcal{Q}_{n}^z$ be the group of $n$-qubit Paulis of the form $Z^a$, $a\in\{0,1\}^n$.
%
For every Pauli operator $P \in \Q_n$, we denote by $w(P)$ the {\em weight} of $P$, which is the number of qubits on which $P$ acts non-trivially. 
\end{defn}

We remind readers of the following basic fact about Pauli operators: For $n$-qubit Paulis $P,Q \in \Q_n$, 
\begin{equation}\label{eq:trpaulipdt}
\Tr(P\cdot Q) = \begin{cases} 
      2^n & P=Q,\\
      0 & \text{otherwise.}
   \end{cases}
   \end{equation}
A great deal of our analysis is devoted to error mitigation on circuits affected by depolarizing noise. Depolarizing noise is an example of a Pauli channel, which are channels that act on Hilbert space operators as 

\beq
\mathcal{P}(\cdot) = \sum_{P\in \mathcal{Q}_n} q_P P (\cdot) P
\eeqc
where $q_P$ is a probability distribution over $\mathcal{Q}_n$.

\subsubsection{Depolarizing noise}
%\ynote{Sumeet please change all $1-q$ to $p$ in this section.}
%\item \textbf{Clifford circuits with Pauli noise}
%Let $\C = \C_l \ldots \circ \C_2 \circ \C_1$ be an $l$-layer Clifford circuit and $\tilde{\C} = \C_1 \circ P_1 \circ \C_2, \ldots \circ \C_l \circ P_l$ be the same Clifford circuit interspersed with Paulis $P_1,\ldots P_l \in \Q_n$. Then for any two Pauli strings $Q, Q' \in \Q_n$, 
%\beq\label{eq:Pauliflip}
%\Tr(\C(Q) Q') = \pm \Tr(\tilde{\C}(Q)Q')
%\eeqp

\begin{defn}[Depolarizing channels]  For $M \in L(\mathbb{C}^d)$, the $d$-dimensional depolarizing channel $\mathcal{D}_p^{(d)}$, for $d\geq 2$, is defined as
\begin{equation}
    \mathcal{D}_p^{(d)}(M)\coloneqq pM+(1-p)\Tr[M]\frac{\mathbb{I}}{d},\quad p\in[-1/(d^2-1),1].
\end{equation}
When the superscript for the dimension is omitted, we implicitly refer to the single-qubit depolarizing channel (with $d=2$):
\begin{equation}
\mathcal{D}_p(M):=\mathcal{D}_p^{(2)}(M)=pM+(1-p)\Tr[M]\frac{\mathbb{I}}{2},\quad p\in[-1/3,1].
\end{equation}
\end{defn}
We note that the single-qubit depolarizing channel has an
alternate representation in terms of the Pauli operators
\begin{equation}\label{eq-single_qubit_depolar_Pauli}
\mathcal{D}_p(\rho) = q_X X\rho X + q_Y Y \rho Y + q_Z Z\rho Z + q_{\mathbb{I}} \rho,
\end{equation}
where 
\begin{equation}
q_X=q_Y=q_Z=\frac{1-p}{4} 
\end{equation}
and 
\begin{equation}
q_I=\frac{1+3p}{4}.
\end{equation}
In the context of depolarizing noise acting within an $n$-qubit circuit, a \textit{global depolarizing channel} is simply a $2^n$-dimensional depolarizing channel. 
Using \eqref{eq:trpaulipdt}, we find that the {\em global} depolarizing channel on $n$ qubits can be written as
\begin{equation}
    \mathcal{D}_p^{(2^n)}(P)=pP+(1-p)\delta_{P,\mathbb{I}}\mathbb{I}\quad\forall~P\in\mathcal{P}_n.
\end{equation}
Alternatively, one could also model the noise within a circuit as an $n$-fold local (single-qubit) depolarizing channel $\mathcal{D}_p^{\otimes n}$. In particular, the single-qubit depolarizing channel has the property that
\begin{equation}\label{eq:single_qubit_depolar_Pauli}
\mathcal{D}_p(X)=pX,\quad \mathcal{D}_p(Y)=pY,\quad \mathcal{D}_p(Z)=pZ,\quad\text{and}\quad\mathcal{D}_p(\mathbb{I})=\mathbb{I}.
\end{equation}

The following Lemma shows that depolarizing noise is particularly amenable to analysis in the Pauli basis:

\begin{lemma}[Action of depolarizing noise on a Pauli string depends on its weight\label{lem:depolarizing}]
For all $p\in[-1/3,1]$ and $P\in\mathcal{P}_n$,
\begin{equation}\label{eq:noisesensitivity}
\mathcal{D}_p^{\otimes n}(P) = p^{w(P)} P.
\end{equation}
\end{lemma}
\begin{proof} This follows immediately from \eqref{eq-single_qubit_depolar_Pauli} and the definition of $\mathcal{D}_p^{\otimes n}$. \end{proof}

\subsection{Unitary 2-designs and Clifford unitaries}

    A \textit{unitary $t$-design}, for $t\in\{1,2,\dotsc\}$, is a finite ensemble $\{(1/K,U_k)\}_{k=1}^K$ of unitaries such that~\cite{DCEL09}
    \begin{equation}
        \frac{1}{K}\sum_{k=1}^K U^{\otimes t}\otimes (U^{\dagger})^{\otimes t}=\int_U U^{\otimes t}\otimes (U^{\dagger})^{\otimes t}~\text{d}U,
    \end{equation}
    where the integral on the right-hand side is with respect to the Haar measure on the unitary group. The $n$-qubit Clifford group $\mathcal{C}_n$ forms a unitary 2-design~\cite[Theorem~1]{DCEL09}, where by definition the Clifford group is the normalizer of the Pauli group $\mathcal{P}_n$~\cite{Gottesman_thesis}, i.e.\ the unitaries that map elements of the Pauli group to elements of the Pauli group under conjugation.
    
When the unitaries in a given ensemble are all Cliffords, the ensemble additionally possesses the following property \cite{Cleve16} that will be crucial to us:

\begin{defn}[Pauli mixing\label{def:Paulimix}]
Consider an ensemble $\mathcal{E}=\left\{p_i, U_i\right\}_{i=1}^k$ where $U_i \in \mathcal{C}_n$. $\mathcal{E}$ is Pauli mixing, if for all $P \in \mathcal{Q}_n$ such that $P \neq I$, the distribution $\left\{p_i, \pi_{U_i}(P)\right\}$ is uniform over $\mathcal{Q}_n \backslash\{I\}$, where $\pi_{U_i}$ is the permutation of $\mathcal{Q}_n \backslash\{I\}$ induced by conjugating $P$ by $U_i$. 
\end{defn}
A Clifford 2-design (i.e., a unitary 2-design whose elements are Cliffords) if and only if it is {\em Pauli mixing}~\cite{Wat18_book}.

\subsection{Decay of purities}

Consider that the $n$-fold Paulis are an orthogonal basis for $\mathcal{H}_2^{\otimes n}$, and recall that the purity upper-bounds the relative entropy to the maximally-mixed state (see Lemma \ref{lem:relent}). A conceptual cornerstone of our proof construction is to consider how the purity of states decays after each successive layer in the circuit, by looking at the distribution over contributions of weight-$k$ Pauli strings to this quantity. More rigorously, for state of interest $\rho = \sum_{P \in \calP_n} c_P P$ expanded in the Pauli basis, the purity may also be expanded as
\begin{equation}
\Tr(\rho^2) = \sum_{P,P' \in \calP_n} c_P c_{P'} \Tr(PP') = \sum_{P \in \calP_n} c_P^2 2^n.
\end{equation}
Here, the second equality is because the product of two non-identical Paulis is another Pauli, which is traceless. We will look at the contribution
\begin{equation}
    C_k := \sum_{\substack{P \in \calP_n\\w(P) = k}} c_P^2 2^n,
\end{equation}
and see how it is distributed over $k$, for different choices of $\rho$. These $
\rho$s correspond roughly to a state progressing through different layers of our circuit. 

\begin{itemize}[leftmargin=*]
\item \textbf{Computational basis state:}
For $a\in\{0,1\}^n$, let
\begin{equation}
    \rho = \ketbra{a}{a} = \frac{1}{2^n} \sum_{b\in\{0,1\}^n} (-1)^{a\cdot b} Z^{b},
\end{equation}
where $a\cdot b=a_1b_1+a_2b_2+\dotsb+a_nb_n\text{ (mod 2)}$. Then 
\beq\label{eq:purity_compstate}
\Tr(\rho^2) = \frac{1}{2^n} \sum_{b \in\{0,1\}^n} 1
\eeqc
where we have invoked orthogonality of Pauli operators Eq.~\eqref{eq:trpaulipdt}. Eq.~\eqref{eq:purity_compstate} shows that the only Pauli strings contributing to the purity are those containing $Z$s, without $X$s and $Y$s. This means that $C_k$ is simply proportional to the total number of weight-$k$ $n$-bit-strings, which is 
${{n \choose k}}/({2^n-1})$. 
%\JJM{Is this argument important later? To me showing the purity of a pure state is one is a bit superfluous.}\ynote{Added a note to the effect that we want to look at what is being summed over}

\item \textbf{Pure product state:}
A very similar story to the above holds in this case. We may decompose an $n$-qubit product state as
\beq
\rho = \rho_1 \otimes \ldots \otimes \rho_n = \sum_k \sum_{P \in \mathcal{P}_n: w(P)=k} c_{P} P
\eeqc
where for $P$ such that $w(P)=k$, $c_{P} = \prod_{i=1}^k \Tr(P_i \rho_i)$ is a $k$-fold product. In particular, suppose that for each $\rho_i$ the distribution over Paulis is bounded away from uniform in the following sense: there exists some Pauli $Q\in \mathcal{P}\setminus \mathbbm{I}$ such that the coefficient of $Q$ in $\rho_i$ is at least $1-\eps$, then for product states, the contribution of weight-$k$ Pauli strings to the purity is approximately also ${{n \choose k}}/({2^n-1})$.
\item \textbf{State acted on by a noiseless Clifford 2-design}
We consider the {\em expected} purity after the two design acts. We see that for any initial state 
\begin{align}
\rho = \frac{1}{2^n} \mathbb{I} + \sum_{P \in \mathcal{P}_n\setminus \mathbbm{1}} c_P P, 
\end{align}
this is
\begin{align}
{\displaystyle\operatornamewithlimits{\mathbb{E}}_{\mathcal{C}\sim\mathcal{E}}}[\Tr(\calC(\rho)^2)]&= {\displaystyle\operatornamewithlimits{\mathbb{E}}_{\mathcal{C}\sim\mathcal{E}}}
\left[\Tr\left( \sum_{P\in \calP_n} c_P \calC(P)\right)^2 \right] \\
&= {\displaystyle\operatornamewithlimits{\mathbb{E}}_{\mathcal{C}\sim\mathcal{E}}} \left[\Tr\left( \sum_{P\in \calP_n} c_P^2  \calC(P) \calC(P)\right)\right]\nonumber\\
&= \Tr \left[ \left( \mathbbm{1}\frac{1}{4^n} + \sum_{P\in \calP_n \setminus \mathbbm{1}} c_P^2{\displaystyle\operatornamewithlimits{\mathbb{E}}_{\mathcal{C}\sim\mathcal{E}}} [\calC(P) \calC(P)]\right)\right] 
\nonumber
\end{align}
and hence
\begin{align}
{\displaystyle\operatornamewithlimits{\mathbb{E}}_{\mathcal{C}\sim\mathcal{E}}}[\Tr(\calC(\rho)^2)]&= \frac{1}{2^n} + \sum_{P\in \calP_n \setminus \mathbbm{1}} c_P^2 \Tr\left(\sum_{Q\in \calP_n \setminus \mathbbm{1}} \frac{1}{4^n-1} Q\cdot Q\right) \nonumber\\
&= \frac{1}{2^n} + \sum_{k=1}^n \sum_{\substack{Q \in \mathcal{P}_n\setminus \mathbbm{1} :\\ w(Q)=k}}  \frac{1}{4^n-1} \left(\sum_{P\in \calP_n \setminus \mathbbm{1}} c_P^2 \right) 2^n  .\label{eq:Panel1b}
\end{align}
Here, the second-last equality follows from the Pauli mixing property of Clifford 2-designs, which says that any non-identity $P$ gets mapped by such an ensemble to a uniform distribution over non-identity Paulis. This means that, in contrast to the above two scenarios, the contribution of weight-$k$ Pauli strings to the purity $\Tr(\mathbb{E}_{\mathcal{C}\sim\mathcal{E}}[\mathcal{C}(\rho)]^2)$ is proportional to the number of weight-$k$ Paulis there are, which is ${{n \choose k}3^k}/({4^n-1})$: {\em higher-weight Paulis contribute more}.
\item \textbf{State acted on by a noisy Clifford 2-design:}
\begin{align}
&{\displaystyle\operatornamewithlimits{\mathbb{E}}_{\mathcal{C}\sim\mathcal{E}}}[\Tr(\Phi_p \circ \mathcal{C}(\rho)^2)] \\ \nonumber
&= \mathbb{E}_{\mathcal{C}\sim\mathcal{E}}\left[\Tr\left(\left(\frac{1}{2^n}
\mathbb{I}+ \frac{1}{4^n-1} \sum_{P \in \mathcal{P}_n\setminus \mathbb{I}} c_P p^{w(\mC(P))} \mC(P) \right)^2\right)\right]\\ \nonumber
&= \mathbb{E}_{\mathcal{C}\sim\mathcal{E}}\left[ \Tr\left(\frac{1}{4^n}\mathbb{I} + \left(\frac{1}{4^n-1}\right)^2  \sum_{P \in \mathcal{P}_n\setminus \mathbbm{1}} c_P^2 p^{2w(\mC(P))} I\right) \right]\\
&= \frac{1}{2^n} + \left(\frac{1}{4^n-1}\right)^2 2^n \sum_{P \in \mathcal{P}_n\setminus \mathbbm{1}} c_P^2 \mathbb{E}_{\mathcal{C}\sim\mathcal{E}}[p^{2w(\mC(P))}] \\ \nonumber
&\approx \frac{1}{2^n} \left(1 + \frac{1}{4^n-1}\sum_{P \in \mathcal{P}_n\setminus \mathbbm{1}} c_P^2 \mathbb{E}_{\mathcal{C}\sim\mathcal{E}}[p^{2w(\mC(P))}]\right) . \nonumber
\end{align}
%\JJM{What do you think of using $\displaystyle\operatornamewithlimits{\mathbb{E}}_{\mathcal{C}\sim\mathcal{E}}$ instead of $\displaystyle\mathbb{E}_{\mathcal{C}\sim\mathcal{E}}$ to make this a bit more compact? Would also apply in the rest of the document.]}
%\JJM{Maybe explain how we get the second equality (using the traceless property)} \ynote{Explained in the preamble to this subsubsection}
By the Pauli mixing property, we find
\begin{align}
\mathbb{E}_{\mathcal{C}\sim\mathcal{E}}[p^{2w(\mC(P))}] &= \sum_{i=1}^{4^n-1} \frac{1}{4^n-1} p^{2w(Q_i)}\\
\nonumber
&= \sum_{k=1}^n \sum_{\substack{Q \in \mathcal{P}_n\setminus \mathbbm{1}:\\ w(Q)=k}} \frac{1}{4^n-1} p^{2k},
\end{align}
implying that 
\begin{align}
&\mathbb{E}_{\mathcal{C}\sim\mathcal{E}}[\Tr(\Phi_p \circ \mathcal{C}(\rho)^2)]\nonumber
\\ &\approx
\frac{1}{2^n} \left(1 + \frac{1}{4^n-1}\sum_{P \in \mathcal{P}_n\setminus \mathbbm{1}} c_P^2  \sum_{k=1}^n \sum_{\substack{Q \in \mathcal{P}_n\setminus \mathbbm{1}:\\ w(Q)=k}} \frac{1}{4^n-1} p^{2k}  \right) . \label{eq:Panel2}
\end{align}
In particular, we see that all weight-$k$ Pauli strings have exactly equal weight in the Pauli basis expansion of the state after the two-design, but now this weight is damped by a factor exponential in the Hamming weight $k$. That is, the contribution of all weight-$k$ Paulis is proportional to ${{n \choose k}3^k}/({4^n-1}) p^{2k}$.
\end{itemize}

\section{Error mitigation: The lay of the land}\label{sec:defs}

%Error mitigation is a procedure that takes copies of a noisy circuit's output and computes a representation of the noise{\em less} circuit's output. To begin the discussion, we will introduce some terminology to do with noisy and noiseless circuits.
%\JJM{[Proposal]
The aim of an error mitigation procedure is to produce a representation of the output of a noiseless quantum circuit given access to the actual noisy quantum device. Before we formally define different error mitigation tasks, we outline the model for noisy quantum circuits that is used throughout this work.
%{\em Noise models} --- 
We consider quantum circuits of depth $D$ acting on $n$ qubits. In the quantum channel picture, such circuits take the form
\begin{align}
\mathcal{C}=\mC^{(D)}\circ\dotsb\circ\mC^{(2)}\circ\mC^{(1)}
\end{align}
where each $\mC^{(1)},\dots,\mC^{(D)}$ denotes a layer of unitary quantum gates. We will then assume that in the noisy version, a quantum channel will act after each unitary we implement.
That is, for such a circuit, we take its noisy version to be 
\begin{equation}
\Phi_{\mC,\mathcal{N}^{(D)},\cdots \mathcal{N}^{(1)}}=\mathcal{N}^{(D)}\circ\mC^{(D)}\circ\dotsb\circ\mathcal{N}^{(2)}\circ\mC^{(2)}\circ\mathcal{N}^{(1)} \circ\mC^{(1)}, 
\end{equation}
where $\mathcal{N}^{(1)},\mathcal{N}^{(2)},\dotsc,\mathcal{N}^{(D)}$ are quantum channels that describe the noise.

For simplicity, we will often consider the case in which every layer has an identical noise channel acting on it; in that case we will denote the noisy version of the circuit as $\Phi_{\mC,\mathcal{N}}$, or in cases when the noise channel $\mathcal{N}$ can be described by a parameter $p$, $\Phi_{\mC,p}$. However, our results extend to the case where the noise is not uniform, or not unital. For pedagogical reasons, our results are initially derived for local depolarizing noise and then extended to other noise models.

\subsection{Error mitigation setting} 

%\JJM{TODO: Update definitions to include training data from other circuits (statistically uncorrelated); Can we update definitions such that we run poly many different circuits to include ZNE?; Motivate and add new definition of BOF (robust strong mitigation) which asserts that probability vectors in any sector can at most decrease polynomially, i.e. if we would like to sample from $p$ but instead sample from $q$ we need to have $p/q \leq O(\operatorname{poly}(n))$, the motivation being that solution states (ones that have $1/poly$ overlap) should be preserved but we don't know where they are; Show that this implies strong mitigation (probably more about finding out what scaling of $\epsilon$ this implies); Add as many examples for protocols}

The protocols listed in 
%Subsection \ref{suppmat:protocols} 
Section I
in the Supplemental Material are but a slice of the wealth of error mitigation techniques that have so far been proposed. In this section we aim to unify all of these techniques in introducing the model of error mitigation that we will be using in the rest of this work. As explained in Section I of the Supplemental Material,
%Subsection \ref{suppmat:protocols}, 
lower bounds against the model will introduce imply lower bounds for many of those protocols. We will take the input of an error mitigation protocol to be the following.

\begin{defn}[Resources for error mitigation]\label{def:error_mitigation_restricted}
An error mitigation algorithm predicts properties of the noiseless output state of a quantum circuit $\calC(\rho)$ given the following elements.
\begin{itemize}\setlength{\itemsep}{0pt}
    \item Classical descriptions of $\mC$ and the noise channel $\calN$ acting on $\mC$.
    \item (Optional) A classical description of the input state to the circuit, $\rho$. If the algorithm utilizes this classical description, we call the algorithm input-state {\em aware}. If the algorithm doesn't utilize this classical description, we call it input-state {\em agnostic}.
    \item The ability to perform collective quantum measurements on multiple copies of the noisy circuit output state $\Phi_{\calC,\calN}(\rho)$ (see the remarks below for the case of randomized families of circuits).
\end{itemize} 
\end{defn}
We now make several remarks about this definition. 
Firstly, our goal is to demarcate the information-theoretic limits of error mitigation. This we do by studying the {\em sample complexity} of error mitigation in our model---how many copies of the noisy output state are required to achieve the desired error mitigation output---rigorously quantifying how this number scales in the complexity of $\mC$ and the noise $\mN$. We note that sample complexity lower bounds computational complexity, and so lower bounds for sample complexity are also lower bounds for computational complexity.

Secondly, the practical scope of our model is broader than meets the eye. It encompasses even those protocols that run {\em multiple} different circuits with varying levels of noise, or take as additional input `training data' which are pairs of (experimental) noisy and (simulated) noiseless expectation values for different circuits. In this case, we can always identify a representative circuit/noise level out of the class of circuits that may be run; it is this circuit whose parameters will determine the complexity of error mitigation. The reader is referred to the discussion in 
Section I
%the Subsection \ref{suppmat:protocols} 
in the Supplemental Material for more details.

Thirdly, we start in Section II
%\ref{sec:conc} 
by proving lower bounds against input-state {\em agnostic} protocols, and then in 
Subsection II.B,
%Section \ref{sec:aware} 
we extend the results to input-state {\em aware} protocols. Input-state {\em agnostic} protocols are a natural starting point, because they cannot possibly work by simulating the circuit $\mathcal{C}$, no matter how shallow that circuit is---simply because they do not know the input to the circuit. In fact, many existing error mitigation protocols in the literature are already covered by the input state-agnostic model, because the classical description of the input state is not a parameter in the protocol. 

Fourthly, note that we have allowed for the ability to perform arbitrary collective measurements. This requires access to a quantum memory, a nontrivial quantum resource that might be out of reach. Because of this, most error mitigation protocols in the literature do not require this ability. However, our conclusions hold up {\em even} against algorithms that have such additional power.

Error mitigation is not an end in itself; typically, it is used as the last step in a pipeline to solve some (quantum) computational task. Depending on what kind of task that is, different outputs of the mitigation procedure are required.
Arguably, the most popular intended application for near-term quantum computers are \emph{variational quantum algorithms}~\cite{McClean_2016,cerezo2021variational,bharti2022noisy} where a quantum state is prepared through a parametrized quantum circuit and the parameters are iteratively adjusted to optimize a function $L(\langle O_1 \rangle, \langle O_2 \rangle, \dots)$ of expectation values of operators evaluated on said state. The archetypal variational quantum algorithm is the variational quantum eigensolver~\cite{McClean_2016,cerezo2021variational,bharti2022noisy} 
% Changed to save a reference.
where the function 
\beq
L(\ket{\psi}) = \sum_i \bra{\psi} O_i \ket{\psi} = \sum_i \langle O_i \rangle
\eeq
is the expectation value of a Hamiltonian $H = \sum_i O_i$. In this case, the ground state of $H$ yields the solution to the optimization problem, and so the optimized parametrized quantum circuit should ideally prepare a state close to the ground state of $H$. However, such circuits are noisy and the goal of error mitigation is to correct this. It stands to reason then that an error mitigation protocol should output estimates of the expectation 
values $\langle O_i \rangle$ on the state output by the noiseless version of these circuits. We can formally capture this task in the following 
definition of \emph{weak} error mitigation. In all of the definitions below, let $\ket{\psi} = \mC(\rho)$ be the state vector output by the noiseless circuit. 

\begin{defn}[Weak error mitigation (expectation 
value error mitigation)]\label{def:Weak_EMalg}
An $(\eps,\delta)$ weak error mitigation algorithm $\mA$ with resources as in Definition~\ref{def:error_mitigation_restricted} takes as input a classical description of a set of 
observables $\calM = \{O_1,\ldots,O_{\ell}\}$ satisfying $\|O_i\|\leq1$ and outputs estimates $\hat{o}_i$ of $o_i=\bra{\psi}O_i\ket{\psi}$ such that
\begin{align}\label{eq:weak}
    \bbP\big[ |\hat{o}_i - o_i| \leq \eps \text{ for all } 1\leq i\leq \ell\big] \geq 1 - \delta.
\end{align}
Here, the probability in Eq.~\eqref{eq:weak} is over the randomness of the error mitigation algorithm. This randomness could come from using classical random bits or from making measurements of the quantum states available as input. 
% eq:PAC and eq:weakoutput used to be here
\end{defn}
The task of computing expectation values is ubiquitous in near-term quantum computing. Most error mitigation algorithms in the literature address the weak error mitigation task, including all protocols listed in 
Section~I %~\ref{suppmat:protocols} 
of the Supplemental Material.
However, in some applications, knowing expectation values is not sufficient. In these cases, we would like to represent the strongest possible access, on par with access to the quantum computer, which is sampling of the circuit's output. This leads us to a definition of \emph{strong} error mitigation:

\begin{defn}[Strong error mitigation (sample error mitigation)]\label{def:Strong_EMalg}
An $(\eps,\delta)$ strong error mitigation algorithm $\mA$ with resources as in Definition~\ref{def:error_mitigation_restricted} outputs a bit-string $z$ sampled according to a distribution $z \sim \mu$ such that with probability $1-\delta$,
\begin{align}
d_{\mathrm{TV}}(\mu, D_{\ket{\psi}}) \leq \eps, \qquad \text{(additive error $\eps$)}
\end{align}
or alternatively
\begin{align}
\frac{D_{\ket{\psi}}(z)}{\mu(z)} \leq \kappa \text{ for all } z \in \01^n, \qquad \text{(multiplicative error $\kappa$)} 
\end{align}
where $D_{\ket{\psi}}$ is the distribution arising from a computational basis measurement of $\ket{\psi}$.
\end{defn}
As the strong error mitigation task is more difficult than the weak error mitigation task---as we will show below, strong error mitigation implies weak error mitigation---and weak error mitigation is usually sufficient for near-term applications, there are fewer methods available that achieve this, an example being virtual distillation~\cite{huggins2021virtual,koczor2021exponential}.
The two notions of error are related---$\kappa$ multiplicative error implies 
%\begin{equation}
$\eps = ({1-{1}/{\kappa}})^{1/2}
$
%\end{equation}
additive error. To see this, note that the multiplicative error requirement can be re-written as 
\beq
\frac{D_{\ket{\psi}}(z)}{\mu(z)} \leq \kappa \rightarrow D_{\infty}(D_{\ket{\psi}} \| \mu) \leq \log(\kappa)
\eeqp
By the monotonicity of relative entropies, we then have
\beq
D_{KL}(D_{\ket{\psi}} \| \mu) = D_{1}(D_{\ket{\psi}} \| \mu) \leq D_{\infty}(D_{\ket{\psi}} \| \mu) \leq \log(\kappa)
\eeqp
By the \emph{Bretagnolle-Huber inequality}, we can then relate this to total variation distance as
\begin{align}
    d_{TV}(D_{\ket{\psi}}, \mu) \leq \sqrt{1 - \exp(- D_{KL}(D_{\ket{\psi}}\|\mu))} \leq \sqrt{1 - \frac{1}{\kappa}}.
\end{align}
%This implies that, if a strong error mitigation algorithm returns samples from a distribution accurate up to multiplicative error $\kappa$, then the underlying distribution is also accurate up to additive error $\sqrt{1 - {1}/{\kappa}}$. Note that if we use $p_i = \delta_{i,1}$ and $q_1 = 1/\kappa$ and arbitrary everywhere else, we obtain $d_{TV}(p,q) = 1 - {1}/{\kappa}$, which means the above inequality is reasonably tight. 
On the other hand, it is easy to check that additive error does not imply multiplicative error for any setting of the parameters.

The additive-error requirement for strong error mitigation makes intuitive sense when sampling access to the noiseless quantum state is required and it is not important that the samples generated by the error-mitigated algorithm should come from any particular subset of the support. However, this is not the case for some of the near-term quantum algorithms that one might want to error mitigate. Consider for example variational quantum optimization algorithms like the quantum approximate optimization problem~\cite{QAOA}, where a diagonal Hamiltonian $\mathcal{H}$ encodes a hard combinatorial optimization problem. Here, computational basis states of low energy correspond to approximate solutions of the combinatorial optimization problem.  If the noiseless state $\ket{\psi}$ has an overlap of $O(1/\operatorname{poly}(n))$ with the low-energy subspace of $\mathcal{H}$, polynomially many samples from the noiseless distribution $D_{\ket{\psi}}$ are sufficient to solve the optimization problem with high probability. And then the same would hold for sampling from an error mitigated distribution with multiplicative error $\kappa=O(1/\operatorname{poly}(n))$: such distribution must also have an inverse polynomially small weight on the set of low-energy strings. Such important examples motivate our definition of the multiplicative error mitigation.

\subsection{Relationship between notions of error mitigation}

%\JJM{As a note, I would put this in the supplementary at the next iteration of the paper, maybe except for the parts that are tools (like the SQ hardness of parities). At the same time, I would put in something about the input-state aware case, especially as the main argument is also easy to state. In the same spirit, I would put in some methods of the beyond unital part, I think the relation between the expected overlap and the channel properties $\eta$ and $\nu$ would be nice. However, I would hold off with this until we get word from the editors, because they probably have their own idea how the methods could be altered.}
%\je{Again, I agree. No need to edit and re-edit now. We will hear from the editors and then  we can see.}

Having established that a comprehensive study of error mitigation must consider both weak and strong versions, we now ask: how are they related? We are particularly interested in understanding if the output of one kind of error mitigation can be used, in polynomial time, to compute the output of another kind of error mitigation. 
A first observation is that \textbf{strong error mitigation implies weak error mitigation with local observables.} Let $\mathcal{A}$ be a strong error mitigation algorithm. Definition \ref{def:Strong_EMalg} requires $\mathcal{A}$ to output a sample from the computational basis of the noiseless circuit's output state. However, if $O_i$ is a local observable (say a product Pauli observable) and we assume that the cost of strong error mitigation on $\mathcal{C}$ does not increase significantly by appending a layer of product unitaries to $\mathcal{C}$, then $\mathcal{A}$ can also output samples from the probability distribution associated with measuring $\mathcal{C}$ in an eigenbasis of $O_i$. After applying the strong error mitigation procedure to obtain enough clean samples from the eigenbasis of $O_i$, we can estimate the expectation value $\Tr(O_i \mathcal{C}(\rho))$ empirically to a desired precision, thereby achieving weak error mitigation.
One of the main questions asked in this work is whether we can hope for the other direction to hold. That is, whether weak error mitigation protocols can be used to also obtain samples from the noiseless circuit. The remainder of this subsection answers this question.

\subsubsection{Weak error mitigation implies strong error mitigation only with exponentially-many observables}
Let us now consider the problem of using the outputs of a weak error mitigation algorithm to obtain a strong error mitigation algorithm. That is, what is the minimum number of expectation values required to produce samples from the noiseless circuit? Here we allow arbitrarily complex post-processing steps because our focus is on sample complexity, or fundamental information-theoretic limits. Restricting our focus to bounded-time (i.e., realistic) computations would only further limit the set of protocols under consideration.

First note that in the limit of {\em exponentially} many error-mitigated expectation values, obtaining a (potentially inefficient) strong error mitigation algorithm is possible. 
Indeed, suppose we were able to perform weak error mitigation, outputting estimates of all $n$-qubit Paulis 
    \begin{equation}
    \Tr(\mathcal{C}(\rho) P_i) \qquad \forall P_i \in \mathcal{Q}_n
    \end{equation}
up to exponentially small precision. This allows us to perform full tomography on the noiseless output $\mathcal{C}(\rho)$---if we were allowed to query exponentially-many such expectation values (since there are exponentially many members of $\mathcal{Q}_n$). However, this procedure would clearly be inefficient and likely more costly than just simulating the circuit classically. 
Thus, the more interesting question is: \textbf{could there be an algorithm that only needs to see {\em polynomially}-many expectation values from the output of weak error-mitigation to obtain strong error mitigation?} 

We proceed to provide a partial negative answer to this question, showing that there cannot exist such an algorithm if the expectation values are all of observables {\em diagonal in the same eigenbasis}. To do so, we showcase an instance of weak error mitigation that provides an explicit counterexample to the conjecture.  
First, we will need to introduce the notion of a {\em statistical query} \cite{KearnsSQ}:

\begin{defn}[Statistical query] A statistical query is a pair $(q,\tau)$ with
\begin{itemize}
\item a function $q : \01^n \times \01  \rightarrow \{0,1\}$.
\item $\tau$: a tolerance parameter $\tau \geq 0$.
\end{itemize}
\end{defn}

Now we are ready to define a statistical query oracle.

\begin{defn}[Statistical query oracle\label{def:SQoracle}] Fix an unknown function $c:\01^n \rightarrow \{0,1\}$ and a probability distribution $D$ over the domain $\01^n$. The statistical query oracle, $SQ(q,\tau)$, when given a statistical query,
returns any value in the range as
\beq
\big[\mathbb{E}_{x\sim D}[q(x,c(x))] - \tau, \mathbb{E}_{x\sim D}[q(x,c(x))] + \tau \big] .
\eeq
\end{defn}

We will also need to use the observation that the problem of $\Par$ on $n$ bits (which we now define) can be solved with $\poly(n)$ samples but requires $\exp(n)$ statistical queries.

%Let us now introduce the class of problems $\Par$:
\begin{defn}[$\Par$]\label{def:par}
The class of $\Par$ is the set of functions $\{c_s: \01^n \rightarrow \01 \}_{s\in \01^n}$ where
\beq
c_s(x) = x\cdot s \quad \text{for }s\in \01^n
\eeqp
For every such $c_s$, we may define the 
associated distribution $P_{s}:\01^{n+1}\rightarrow [0,1]$ via
\begin{equation}
    P_{s}(x\Join y) = \begin{cases}
        2^{-n}\,,& \text{if}\quad y=x\cdot s,\\
        0 \,,&  \text{else}\,,
    \end{cases}
\end{equation}
where $x\in\{0,1\}^n$, $y\in\{0,1\}$, and the symbol $\Join$ means concatenation. That is, $P_s$ is the distribution which is supported uniformly on those bit-strings whose last bit is the parity of the subset of the first $n$ bits that are indexed by the `secret' $n$-bit string $s$. 
\end{defn}

There are three important facts about $\Par$ that we will use. 
\begin{enumerate}
    \item For every $s\in \01^n$, there exists an $n+1$-qubit quantum circuit $\mC_s$ that `encodes' the distribution $P_s$, in the following sense: when initialized on the all-$0$s state, the output state has support only on the computational basis states whose labels are in the support of the distribution $P_s$:
    \beq\label{def:Cs}
    \mC_s\ket{0}^{\otimes n+1} = \frac{1}{\sqrt{2^n}} \sum_{x\in \01^n} \ket{x \Join c_s(x)}
    \eeqc
    for the parity function $c_s(x) = x \cdot s.$ This has been proven in, for example, Ref.~\cite{ouramazingTgatepaper}.
    \item $\Par$ are not solvable with Clifford data regression sub-exponentially-many statistical queries \cite{BFJKMR94}:
    \begin{lemma}[SQ-hardness of $\Par$ \cite{BFJKMR94}\label{lem:SQparities}]
    Any learning algorithm that is restricted to make statistical queries of the form $(\chi, \tau )$, where $\tau \geq \tau_0$ for each query, and for all $c_s\in \Par$, and all $D$, is able to output a hypothesis $h: \01^n \rightarrow \01$ such that $\mathbb{P}_{x\sim D}[h(x)\neq c_s(x)] \leq \eps$ must make $\Omega(\tau_0^2 \cdot 2^n)$ queries. 
    \end{lemma}
    \item We can lower-bound the total variation distance between any two $\Par$ distributions as follows.
    
    \begin{lemma}[{Parity distributions}]\label{lem:paritiesdist}
    For any two $s, s' \in \01^k$, $d_{TV}(P_s, P_{s'}) \geq 1/2.$ 
    \end{lemma} 
    %In fact we can prove, using a different technique, that it is exactly $1/2$.\ynote{Insert the second proof if I have time}).
    \begin{proof}
    Le Cam's two point method says that the failure probability of any binary hypothesis test $\Psi:X \rightarrow \{P,Q\}$ on two distributions $P, Q:X \rightarrow [0,1]$ each chosen with probability $1/2$, is lower bounded as
    \beq\label{eq:lecam}
    \min_{\Psi} P_{\text{error}}(\Psi) \geq \frac{1-d_{TV}(P,Q)}{2}
    \eeqp
    For any two strings $s, s' \in \01^k$, let $P= P_s$ and $Q=P_{s'}$. We wish to bound $d_{TV}(P_s,P_{s'})$. To do so we will exhibit a hypothesis test on $P_s$ versus $P_{s'}$ that has error probability $1/4$. Plugging this into Eq.~\eqref{eq:lecam} yields the desired lower bound on $d_{TV}(P_s, P_{s'})$. 
    A simple test is then as follows: given a string $x\Join c$ drawn from either $P_s$ or $P_{s'}$, where $c$ is one bit, we compute the parity of $x$ with $s$ and $s'$. In Case 1, we have $x\cdot s=x\cdot s'$. Then we output $s$ or $s'$ uniformly at random. This succeeds with probability $1/2$. Otherwise (call this Case 2) $x\cdot s\not=x\cdot s'$. Then we output whichever string $s$ or $s'$ yields the right parity $c$. This succeeds with probability 1 in distinguishing the two distributions. So the probability of error of the test we have described is simply
    \beq
    P_{\text{error}} = P(\text{Case 1}) 1/2
    \eeqp
    Now let us compute $P(\text{Case 1}) = 1- P(\text{Case 2})$. As the $x$ part of the string is chosen uniformly from $\01^n$, computing $P(\text{Case 2})$ boils down to computing the number of $x$'s that either satisfy the equations $x\cdot s=0$ and $x\cdot s'=1 $ or $x\cdot s=1$ and $x\cdot s'=0$. In both cases, $x$ must satisfy two linearly independent relations. Thus, we conclude that the strings that satisfy $x\cdot s=0$ and $x\cdot s'=1$ form a vector space of dimension $n-2$, which contains $2^{n-2}$ strings. Thus, the number of strings that satisfy either $x\cdot s=0$ and $x\cdot s'=1 $ or $x\cdot s=1$ and $x\cdot s'=0$ is $2^{n-1}$ which is half of all the strings $x$. Thus, $P(\text{Case 2}) = P(\text{Case 1}) = 1/2$, and so the overall $P_{\text{error}} = 1/2*1/2 = 1/4$.
    \end{proof}
\end{enumerate}

In addition, we will need the notion of hypothesis selection and an algorithm that achieves a good approximation guarantee for it, due to Yatracos \cite{Yatracos85}:

\begin{theorem}[$3$-proper hypothesis selection]\label{thm:hs}
Fix a class of distributions $\mathcal{Q} = \{q_1, \dots , q_k\}$ and $\eps, \delta >0$. Given $O(\log |\mathcal{Q}|/\eps^2)$ samples from a target distribution $p$ (which may not be in $\mathcal{Q}$), there is an algorithm to output a distribution $q^{\ast} \in \mathcal{Q}$ satisfying 
\beq
d_{TV}(p, q^{\ast}) \leq 3 \, \min_{i\in[k]} {d_{TV}(p, q_i)} + \eps
\eeq
with probability at least $1-\delta$.
\end{theorem}

Here, the word {\em proper} refers to the fact that the output distribution is required to be in $\mathcal{Q}$, a feature that we will require for the argument we are about to make. %
The connection between expectation values (the output of weak error mitigation) and statistical queries is encapsulated in the following observations (where for $b\in \01^n$, $Z^{b}\coloneqq Z^{b_1}\otimes Z^{b_2}\otimes\dots\otimes Z^{b_n}$):

\begin{lemma}[Weak error mitigation on the circuits $\calC_s$ outputs statistical queries to the $\Par_s$ distribution]\label{lem:SQ}

Consider error mitigation on the $n+1$-qubit $\Par$ circuit $\mC_s$ (see Def.~\ref{def:par}). That is, let $\mathcal{M} = \{(\mathbb{I}+Z^b)/{2}\}_{b\in \01^{n+1}}$, and suppose we have a weak error mitigation algorithm $\mathcal{A}(\mathcal{C}_s,\mathcal{N},\mathcal{M})$ that, with probability $1-\delta$, outputs $\tau$-accurate estimates
\beq\label{eq:SQweak}
|\hat{o}_i - \Tr[\calC_s(\ketbra{0}{0}^{\otimes n+1}) O_b)]| < \tau \qquad \text{for all $O_b \in \mathcal{M}$}
\eeqp
The outputs $\hat{o}_i$ are valid responses to statistical queries with tolerance $\tau$ for the distribution $P_s$.
\end{lemma}

% \begin{fact}[Weak error mitigation outputs statistical queries]
% The outputs of weak error mitigation, the set of estimates of
% \beq\label{eq:expSQ}
% \{\Tr[\mathcal{C}(\ketbra{0}{0}) O_i]\}_i
% \eeq
% to error $\tau$, each take the form of a response to a (classical) statistical query of tolerance $\tau$, as long as $\lVert O_i \rVert$ is bounded for all $i$.
% \end{fact}
The observation that weak error mitigation outputs statistical queries has also been made in Ref.~\cite{AGY20} which introduces the notion of a {\em quantum statistical query} (QSQs) in the quantum PAC learning setting. 
We will come back to this after we state our main theorem. 
\begin{proof}
We will prove that each $\tau$-accurate expectation value output by $\mA$ that satisfies Eq.~\eqref{eq:SQweak} takes the form of a response of a statistical query oracle to a statistical query of tolerance $\tau$, by specifying what $c, q, D$ from Definition \ref{def:SQoracle} correspond to. Let us denote the circuit's clean output state vector as $\calC_s\ket{0}^{n+1} = \ket{\psi} = %\frac{1}{\sqrt{2^n}} 
2^{-n/2}
\sum_{x\in \01^n} \ket{x \Join c_s(x)}$. Then it is easy to see, from inspecting the basis vectors in the support of $\ket{\psi}$, that the desired expectation values can be rewritten as
\begin{align}
\Tr(\ketbra{\psi}{\psi} (\mathbb{I}+Z^b)/2) &= \mathbb{E}_{z\sim P_{s}}[q_b(z)] \nonumber 
\\
&= \mathbb{E}_{x\sim \text{Unif}\{0,1\}^n}[q_b(x,c_s(x))]
\end{align}
where $q_b: \01^{n+1} \rightarrow \01$ is defined as
\beq
q_b(z) = \frac{1+\langle z| Z^b | z \rangle}{2}
\eeqp
That is, $q=q_b, c= c_s$ and $D = \text{Unif}\{0,1\}^n$.
\end{proof}
The re-scaling of the observables $Z^b$ is cosmetic; it is merely to obtain a set of observables whose eigenvalues take values in $\01$ due to the way we have defined statistical queries in 
Def.~\ref{def:SQoracle}. In a similar fashion, while this lemma is written for the observables $Z^b$, one could write an analogous lemma for observables that are all diagonal in some other common eigenbasis, by simply applying the corresponding basis change at the end of the circuit $\mC_s$ and re-scaling the value of the tolerance parameter $\tau$ for the statistical query.

This observation will be crucial to the main theorem of this section (Theorem \ref{thm:WeaktoStrong}), which states that with sub-exponentially-many expectation values of observables which share the same eigenbasis, \text{weak} error mitigation cannot imply strong error mitigation in the worst-case. We now state this more formally:
%%%%%%%%%

\begin{theorem}[Exponentially-many expectation values of observables in the same eigenbasis are required to output samples from that basis]\label{thm:WeaktoStrong}
There is a class of $n$-qubit circuits, such that for every circuit $\mC$ in the class (denote the clean circuit output state vector as $\ket{\psi} = \mC(\ket{0})$) and the set of observables $\mathcal{M} = \{Z^b\}_{b\in \01^{n+1}}$ with $m = o(\tau^2 \cdot 2^n)$, the following holds:

No algorithm $\mathcal{B}$ exists that takes as input the output of weak error mitigation $\mA(\mC,\mN,\mathcal{M})$---the estimates $\{\hat{z}_b\}_{b \in \01^m}$ with $|\hat{z}_b - \bra{\psi}Z_b\ket{\psi}| < \tau$ for all $b$---and outputs $O(n)$-many samples from some distribution $p$ where
\beq
d_{TV}(p,D_{\ket{\psi}}) \leq 1/16
\eeqc
where $D_{\ket{\psi}}$ is the computational basis distribution on $\ket{\psi}$.
\end{theorem}

After we prove this, we will remark that not even the power to choose the observables {\em adaptively} will make it possible to transform a sub-exponential number of expectation values into linearly-many samples.

% \begin{theorem}[Superpolynomially-many observables from weak error mitigation are required for strong error mitigation]\label{thm:WeaktoStrong2}
% Suppose we have a weak error mitigation algorithm $\mathcal{A}(\mathcal{C},\mathcal{N},\mathcal{M})$ in the setting of Lemma \ref{lem:SQ}. That is, $\mathcal{A}$ outputs a set of estimates $\{\hat{o}_i\}_{i=1}^m$ where $m = O(poly(n))$, such that, with probability $1-\delta$, $|\hat{o}_i - \bra{\psi}O_i\ket{\psi}| < \eps$. 

% Then there does not exist an algorithm $\mathcal{B}$ that takes as input the output of $\mathcal{A}$ and outputs $O(n)$-many samples from some distribution $p'$ where
% \beq
% d_{TV}(p',D_{\ket{\psi}}) \leq \eps
% \eeqp
% Hence, $poly(n)$-many outputs of weak error mitigation do not imply strong error mitigation in the worst case. 
% \end{theorem}

\begin{proof}[Proof of Theorem \ref{thm:WeaktoStrong}]
Ref.~\cite{ouramazingTgatepaper} has shown that there is a class of $n$-qubit Clifford circuits whose output distributions are exactly the set of $\Par$ distributions. Suppose to the contrary that for some circuit $\mC$ in this class whose output distribution is some $P_s \in \Par$, there exists some set of observables $\mathcal{M}$ with $|\mathcal{M}|= o(\tau^2 \cdot 2^n)$ and some algorithm $\mathcal{B}$, which takes as input the estimates $\hat{z}_i$ of weak error mitigation up to error $\tau$ and outputs $O(n)$ samples from some distribution $p$ such that $d_{TV}(p,P_s) \leq 1/16$. 

We will show that we can use these samples to solve for the hidden string and thus solve the problem of $\Par$. However, this would imply a contradiction: Recall from Lemma \ref{lem:SQ} that each expectation value estimate is a statistical query to some distribution from the class of $\Par$, but as stated in Lemma \ref{lem:SQparities}, there is no \emph{statistical query} algorithm to solve $\Par$ with $o(\tau^2 \cdot 2^n)$ statistical queries.

In the rest of this proof, we will explain how to use the samples output by the presumed $\mathcal{B}$ to solve $\Par$. The key is simply to run the hypothesis selection algorithm of Ref.~\cite{Yatracos85} (Theorem \ref{thm:hs}), with the set of candidate hypotheses being the set of all Clifford distributions encoding $\Par$. This algorithm will take as input the $O(n)$ samples from $p$ and by the guarantees of Theorem \ref{thm:hs}, will output a distribution $P_{s'} \in \Par$ such that
\beq\label{eq:hs_guarantee}
d_{TV}(p,P_{s'}) \leq 3 \min_{\tilde{s} \in \01^n} d_{TV}(p, P_{\tilde{s}}) + 1/16
\eeq
Now, consider that by the guarantees of $\mathcal{B}$, $3 d_{TV}(p,P_{s}) + 1/16 = 1/4$ and so Eq.~\eqref{eq:hs_guarantee} yields that $d_{TV}(p,P_{s'}) < 1/4.$ But as stated in Lemma \ref{lem:paritiesdist}, for all $s, s'$,
\beq\label{eq:half}
 d_{TV}(P_{s}, P_{s'}) \geq 1/2
\eeqc
 and so by the triangle inequality, 
 \beq
 \arg \min_{\tilde{s} \in \01^n} d_{TV}(p, P_{\tilde{s}}) = s
 \eeq
i.e., hypothesis selection recovers the 
hidden string of the $\Par$ problem successfully. This implies the contradiction. 
\end{proof}

\textbf{Extensions of our proof.}  
We remark that our bound holds against even an {\em adaptive} choice of 
observables in error mitigation because the statistical query hardness of $\Par$ holds against an adaptive choice of statistical queries. 
However, we reiterate a limitation of our result: we can only rule out obtaining strong error mitigation from weak error mitigation with sub-exponentially-many {\em observables whose eigenbasis is the basis we desire to sample from}. It is natural to ask if we can lift this restriction. Here, we point out that the discussion in Ref.~\cite{francaGameQuantumAdvantage2020} shows that for `flat' circuit output distributions, and a specially-crafted set of observables with multiple different eigenbases, only polynomially-many expectation values are needed to obtain a sampler (hence we cannot lift the restriction in general). We could also take a different tack: our proof is based on the hardness of $\Par$ from {\em classical} statistical queries. 

As has been observed by Ref.~\cite{AGY20}, it is also possible to define {\em quantum} statistical queries for a given unknown distribution $D$, and unknown function $c$, and QSQs generalize classical statistical queries by allowing for observables not diagonal in the eigenbasis defined by $c$. Hence, if we knew of a problem that was hard for certain classes of {\em quantum} statistical queries (in the sense that exponentially-many QSQs are needed to solve them), our proof technique could then be applied to say something about the hardness of obtaining strong error mitigation from weak error mitigation without the restriction mentioned at the beginning of this paragraph. Unfortunately we do not know of any such problems.  

\section{Outline of the proof of Theorem~\ref{thm:2_intro}}

To prove Theorem~\ref{thm:2_intro}, and more generally Theorem~4
%\ref{thm:overall_lightcones} 
in the Supplemental Material, we engineer a family of quantum circuits $\mC$ that converge very quickly to the maximally-mixed state under noise, so that $D(\mC'(\ketbra{z}{z})\rVert \mathbb{I}/2^n) \leq n+ \log \Tr(\mC'(\ketbra{z}{z})^2) \leq p^{nD}$. Figure \ref{fig:error_mitigation_englament} illustrates the intuition behind our construction. A toy model of our construction is a circuit consisting of alternating noiseless 2-designs and depolarizing noise. The 2-designs are from an ensemble of Clifford circuits that is {\em Pauli mixing}, which means that it maps each Pauli to a uniformly random Pauli. Intuitively, such random circuits spread entanglement very fast, as such a uniform distribution puts significant weight on the set of higher Hamming-weight, i.e., more non-local, Paulis. 

But this renders such circuits more sensitive to noise. Our proof illustrates this phenomenon quantitatively by tracking the evolution of the purity $\Tr(\rho^2)$ of some state $\rho$ progressing through our circuit. Expanding the purity in terms of the Pauli basis and studying the distribution over contributions by Paulis (grouped by Hamming weight), we see that applying a Pauli-mixing circuit shifts the distribution toward higher-weight Paulis. Such Paulis, however, are damped more by the next layer of depolarizing noise, since depolarizing noise causes the coefficient of the Pauli expansion to decay exponentially in its Hamming weight. We then iterate the argument for every additional layer of a 2-design followed by noise. Formal definitions of the Pauli mixing property as well as a rigorous version of this argument, dealing also with noise {\em within} the 2-designs, are provided in the Supplementary Material. %\ynote{Sumeet please shift this to the Methods, as well as insert references in the preceding paragraph.}
}

\newpage
\pagebreak
\clearpage
\foreach \x in {1,...,\the\pdflastximagepages}
{
	\clearpage
	

\section*{Supplementary material}

# Supplemental Material: Exponentially tighter bounds on limitations of quantum error mitigation

This supplemental material provides substantial detail of the arguments presented in the main text and in the methods section.

- Section II presents our conceptual contribution. Here we prove the 'appetizer' Theorem 1 which studies the limitations of error mitigation under depolarizing noise. While this appetizer theorem implies a weaker scaling of error mitigation complexity as compared to our main theorems, we elect to start here as the proof already introduces the conceptual 'outer loop' that we will use in the rest of our theorems, namely the application of the generalized Fano's lemma to relate error mitigation to the noisy state discrimination problem.

- In Section III, we present our main theorem (Theorem 3), and fill out the technical 'inner loop' of our argument under depolarizing noise. The proof of Theorem 3 is our first technical contribution: the worst-case circuit ensemble that we construct, which has the property that the relative entropy of its average output state to the maximally-mixed state—and hence error mitigation's sample complexity—scales, in the worst case, exponentially in both $n$ and $D$. We conclude the section by further refining our results so that they hold for local circuits and even lower depths than before.

- Section IV contains our second technical contribution: we extend all our previous conclusions to error mitigation under a new class of noise models, non-unital noise, although under the toy model of local non-unital local noise interspersed by global $2-$designs. We arrive at the same conclusions in this toy model, the sample complexity for error mitigation is exponential.

## I. RELATION TO PRACTICAL ERROR MITIGATION PROTOCOLS AND RELATED WORK

In this section, we show that our mathematical framework describes many broad classes of error mitigation protocols used in practice, by showing that the same line of argument that we use in the main text can also be applied, with small adaptations, to these other protocols. We also briefly mention related attempts in the literature to rigorously analyze the sample complexity of error mitigation. We will list these classes, briefly explain how they work, and describe to what our extent our results apply to them. In all of the following, $\mathcal{C}'$ refers to a noisy version of the circuit $\mathcal{C}$, $\omega_{\mathcal{C},\mathcal{N}}$ refers to the output state of the circuit affected by noise, while $|\psi\rangle$ refers to the output state vector of the noiseless circuit.

The key interface between the specific workings of any class of algorithms, and our mathematical machinery, is a relative entropy quantity defined on various output probability distributions of those algorithms run on different inputs. The following Lemma, which we will use extensively in this section and hence refer to as the 'central lemma', captures this rigorously. The central lemma is proven by using exactly the same mapping from noisy state discrimination to error mitigation that we use to prove Theorem 1, and invoking the *generalized Fano method for multiple hypothesis testing* [1] to the former problem.

**Lemma 1** (Central lemma). *Fix a class of input state-agnostic weak error mitigation algorithms. Suppose there is a set of states $\{\rho_i\}_{i=0}^{N}$ and a circuit $\mathcal{C}$, such that if we let $P_i$ be the distribution over the outputs of any algorithm in that class, run on $m$ copies of the noisy circuit's output state $\mathcal{C}'(\rho_i)$, and that*

$$\frac{1}{N+1}\sum_{k=0}^{N} D(P_k \| P_N) \leq \log(N) \cdot m \cdot \beta, \tag{1}$$

*then there exists a problem instance (namely, specified by $\mathcal{C}', \{\rho_i\}_{i=0}^{N}$) on which no error mitigation algorithm in the class can succeed with less than $m = \Omega(\beta^{-1}(1-\delta))$ copies of the noisy circuit's output state.*

Our framework can thus handle any protocol, for which we can show an upper-bound on the quantity on the left-hand-side of Eq. (1). We now proceed to list out well-known classes of error mitigation protocols for which we can show this:

1. *Virtual distillation* [2, 3]: This protocol takes as input $k$ copies of the noisy circuit output state $\omega_{\mathcal{C},\mathcal{N}}$ and outputs an estimate of

$$\langle O \rangle_{\text{corrected}} = \frac{\text{Tr}(O\omega_{\mathcal{C},\mathcal{N}}^k)}{\text{Tr}(\omega_{\mathcal{C},\mathcal{N}}^k)}. \tag{2}$$

as its approximation of the ideal $\langle O \rangle = \text{Tr}(O|\psi\rangle\langle\psi|)$. The intention of exponentiating the noisy state $k$ times, is to amplify the dominant eigenvector in the eigen-decomposition of $\omega_{\mathcal{C},\mathcal{N}}$, as under some assumptions, this dominant eigenvector has a large overlap with the noiseless output state vector $|\psi\rangle\langle\psi|$. The protocol is known as 'virtual' distillation because the quantity in Eq. (2) can be estimated without ever preparing the state $\omega_{\mathcal{C},\mathcal{N}}^k/\text{Tr}(\omega_{\mathcal{C},\mathcal{N}}^k)$.

Concretely, the numerator and denominator of Eq. (2) are often estimated by applying a short quantum circuit to post-process the $k$ copies of $\omega_{\mathcal{C},\mathcal{N}}$, measuring it and then applying classical post-processing. For example, Ref. [2] shows that for certain observables O, it suffices to apply a single-layer of two-qubit gates to $k$ copies of $\omega_{\mathcal{C},\mathcal{N}}$, followed by a computational basis measurement. As remarked below the proof of Proposition 1, our proof goes through even when the error mitigation algorithm consists not only of the classical post-processing steps, but also manipulates the quantum states coherently before the measurement.

On the other hand, the procedure presented in Ref. [3] passes all $k$ quantum states through a short circuit, and the probability of measuring an auxiliary system in the '0' state is related to the value of the numerator and denominator of Eq. (2). For circuits affected by unital noise, Ref. [4] has shown that the probability of success of virtual distillation in this protocol decays exponentially both in the number of qubits and depth of the circuit. This work shows that the same holds for circuits that form a 2-design affected by non-unital noise. As we show in Section IV C, after we apply a 2-design followed by one layer of product noise, the purity of the system is already exponentially small in expectation. Since the success probability of virtual distillation is $\mathrm{Tr}(\omega_{\mathcal{C},\mathcal{N}}^k)$ which is upper-bounded by $\mathrm{Tr}(\omega_{\mathcal{C},\mathcal{N}}^2)$, the former will also be exponentially small when the latter is.

2. *Clifford data regression (CDR) and other learning-based schemes for error mitigation* [5–7]: Fix a circuit of interest $\mathcal{C}$, assumed to be not classically simulatable, which is to be error-mitigated. These protocols generally proceed by collecting 'training examples'

$$(X_i^{\mathrm{noisy}}, X_i^{\mathrm{exact}})_i,$$

which are expectation values of the observable on noisy and noiseless classically simulatable circuits related to $\mathcal{C}$. $X_i^{\mathrm{noisy}}$ are obtained via running the actual circuits; $X_i^{\mathrm{exact}}$ are obtained by classical simulation. One then makes the ansatz that $X_i^{\mathrm{noisy}}$ and $X_i^{\mathrm{exact}}$, for classes of observables or circuits, are related via some functional relation—Ref. [5] uses a linear ansatz—and fits the parameters of the ansatz using the training pairs. Then, finally, given a noisy expectation value $\mathrm{Tr}(O\omega_{\mathcal{C},\mathcal{N}})$ measured on the actual circuit, the learned function is applied to obtain the noiseless value $\mathrm{Tr}(O|\psi\rangle\langle\psi|)$.

The techniques in this work suffice to bound the number of runs of $\mathcal{C}'$ that must be used in the last step, and we sketch how to do so in the following paragraph. (This of course also lower bounds the total number of circuit runs, including those used to generate the training examples.) We note however, that this last step is what we term 'input state-aware' – that is, the learned function must be applied to circuits initialized with the same input state as was used in the 'training' phase. Nevertheless, the conclusion of the central lemma holds, albeit with a slight reinterpretation, for input-state aware protocols. We elaborate on this framework for dealing with the input state dependence in Section II B.

Let $\mathcal{C}_1, \ldots, \mathcal{C}_r$ be the $r$ distinct circuits that are run in the process of generating the training data, and suppose each circuit is run $g$ times and measured after each run. Let Train be the algorithm that generates the training data—that is, Train takes in

$$\omega_{\mathcal{C}_1,\mathcal{N}}^{\otimes g} \otimes \omega_{\mathcal{C}_2,\mathcal{N}}^{\otimes g} \otimes \ldots \omega_{\mathcal{C}_r,\mathcal{N}}^{\otimes g}, \tag{3}$$

(here we do not specify the input to the circuits used to generate the training data because it is irrelevant), makes a quantum measurement on $\omega_{\mathcal{C}_1,\mathcal{N}}^{\otimes g}$ to generate $X_i^{\mathrm{noisy}}$ and classically simulates it to generate $X_i^{\mathrm{exact}}$, for all $i$. Further, let $\mathcal{A}$ be the error mitigation algorithm that takes as input $m$ copies of $\Phi_{\mathcal{C},\mathcal{N}}(\rho_k)$, and the training data (which is the output of Train), and outputs the expectation values $\hat{O} \in \mathbb{R}^m$ where $\hat{O}[i]$ is the estimate for $\langle\psi|O_i|\psi\rangle$. We denote the probability distribution of this algorithm over outputs $\hat{O}$ when the circuit input is $|k\rangle\langle k|$ as $P_k$. Then

$$D(P_k \| P_N) = D\left(\mathcal{A}\left(\Phi_{\mathcal{C},\mathcal{N}}(\rho_k)^{\otimes m} \otimes \mathsf{Train}(\omega_{\mathcal{C}_1,\mathcal{N}}^{\otimes g} \otimes \omega_{\mathcal{C}_2,\mathcal{N}}^{\otimes g} \otimes \ldots \omega_{\mathcal{C}_r,\mathcal{N}}^{\otimes g})\right) \| \mathcal{A}\left(\frac{\mathbb{I}}{2^n}^{\otimes m} \otimes \mathsf{Train}(\omega_{\mathcal{C}_1,\mathcal{N}}^{\otimes g} \otimes \omega_{\mathcal{C}_2,\mathcal{N}}^{\otimes g} \otimes \ldots \omega_{\mathcal{C}_r,\mathcal{N}}^{\otimes g})\right)\right) \tag{4}$$

$$\leq D\left(\Phi_{\mathcal{C},\mathcal{N}}(\rho_k)^{\otimes m} \otimes \mathsf{Train}(\omega_{\mathcal{C}_1,\mathcal{N}}^{\otimes g} \otimes \omega_{\mathcal{C}_2,\mathcal{N}}^{\otimes g} \otimes \ldots \omega_{\mathcal{C}_r,\mathcal{N}}^{\otimes g}) \left\| \frac{\mathbb{I}}{2^n}^{\otimes m} \otimes \mathsf{Train}(\omega_{\mathcal{C}_1,\mathcal{N}}^{\otimes g} \otimes \omega_{\mathcal{C}_2,\mathcal{N}}^{\otimes g} \otimes \ldots \omega_{\mathcal{C}_r,\mathcal{N}}^{\otimes g})\right.\right) \tag{5}$$

$$\leq mD\left(\Phi_{\mathcal{C},\mathcal{N}}(\rho_k) \left\| \frac{\mathbb{I}}{2^n}\right.\right) + D\left(\mathsf{Train}(\omega_{\mathcal{C}_1,\mathcal{N}}^{\otimes g} \otimes \omega_{\mathcal{C}_2,\mathcal{N}}^{\otimes g} \otimes \ldots \omega_{\mathcal{C}_r,\mathcal{N}}^{\otimes g}) \left\| \mathsf{Train}(\omega_{\mathcal{C}_1,\mathcal{N}}^{\otimes g} \otimes \omega_{\mathcal{C}_2,\mathcal{N}}^{\otimes g} \otimes \ldots \omega_{\mathcal{C}_r,\mathcal{N}}^{\otimes g})\right.\right)$$

$$\leq mD\left(\Phi_{\mathcal{C},\mathcal{N}}(\rho_k) \left\| \frac{\mathbb{I}}{2^n}\right.\right), \tag{6}$$

where the first inequality follows by data processing and the second follows from the tensorization of relative entropy. We conclude by plugging Eq. (6) into the central lemma (Lemma 1), using the bounds on $D(\Phi_{\mathcal{C},\mathcal{N}}(\rho_k)\|\frac{\mathbb{I}}{2^n})$ we have obtained in the main text, and plugging it all into Theorem 2 which pertains to input state aware protocols.

3. *Zero-noise extrapolation* (ZNE) [8]: Given the circuit $\mathcal{C}$ with noise strength $\lambda$, one runs it at $n$ different *amplified* noise levels $c_j\lambda$ ($c_j > 1$), computes expectation values on each such circuit, and combines the noisy estimates into an estimate for the expectation value without noise by extrapolating the set of obtained expectation values to zero noise. Variations on the theme include considering other noise models and interpolation schemes [9–11].

We comment now on how our results apply to the proposal of Ref. [8]. The extension to other variations of this scheme should be similar. Here, to run a circuit at noise level $c\lambda$ where $c \geq 1$, each gate is run for a time $cT$, where $c = 1$ corresponds to the noise rate in the original circuit. Ref. [8] further prescribes that the final estimate should be computed as

$$\hat{E} = \sum_{j=1}^{n} \gamma_j \hat{E}(c_j\lambda) \quad , \tag{7}$$

where $\sum_{j=0}^{n} \gamma_j = 1$ and $\sum_{j=0}^{n} \gamma_j c_j^k = 0$ for $k = 1, \ldots, n$. If the evolution of the original circuit is

$$\frac{\partial}{\partial t}\rho = -i[K(t), \rho] + \lambda \mathcal{L}(\rho) \quad , \tag{8}$$

where $K(t)$ is the coherent circuit and $\mathcal{L}$ represents the noise we want to mitigate, then stretching all gate times by $c$ has the effect of replacing $\lambda \leftarrow c\lambda$ for $c > 1$ in the above equation. In the quantum circuit picture, if we assume that each circuit is acted upon by local depolarizing noise with effective depolarizing parameter $p$ after each gate, this has the effect of replacing the parameter $p \leftarrow p^c$. Let us assume that $c$ is the minimum noise level at which the circuit is run. Indeed, $c$ cannot be arbitrarily small, as in this case the circuit is noiseless and error mitigation would not be required. Then, letting $P_k$ denote the probability distribution over $\hat{E}$ when the input to the circuits we have constructed is $|k\rangle\langle k|$, we apply a similar line of reasoning as in Eqs. (4) to (6) to obtain a relative entropy bound. In fact, our choice to compute these bounds with $c$ taking its minimum value weakens them, as the higher noise levels used in practice result in even faster convergence to the identity.

4. *Probabilistic error cancellation* [8, 9, 12]: In this protocol, one rewrites a noiseless circuit $\mathcal{U}$ in terms of a linear combination of noisy circuits $\mathcal{F}_i$: $\mathcal{U} = \sum_i a_i \mathcal{F}_i$ where $a_i \in \mathbb{R}$ could be negative. Similarly, any desired expectation value of interest (of the observable $M$, say) can be decomposed as the convex combination

$$\mathrm{Tr}(M\mathcal{U}(\cdot)) = \sum_i a_i \mathrm{Tr}(M\mathcal{F}_i(\cdot)) = \gamma \sum_i p_i \mathrm{sign}(a_i) \mathrm{Tr}(M\mathcal{F}_i(\cdot)) \quad . \tag{9}$$

Then, with probability

$$p_i = \frac{|a_i|}{\sum_i |a_i|} \tag{10}$$

and $\gamma = \sum_i |a_i|$, $\mathrm{Tr}(M\mathcal{F}_i(\cdot))$ is estimated by running the noisy circuit $\mathcal{F}_i$ $g$ times. This is repeated $r$ times (each time sampling $\mathcal{F}_i$ anew from $p$), the results are averaged and the estimate is multiplied by $\gamma$.

Our analysis applies here too, because we can model the computation of the desired quantity (the right-hand-side of Eq. (9)) with a quantum channel $\mathcal{M}$, as follows: Consider the state

$$\omega_k^{\otimes r} = \bigotimes_{j=1}^{r} \left( \sum_i p_i |i\rangle\langle i|_{A_j} \otimes (\mathcal{F}_i(|k\rangle\langle k|)^{\otimes g})_{B_j} \right) \quad , \tag{11}$$

and the channel $\mathcal{M}$ which consists in 1) Applying $\bigotimes_j P_{A_j}$ on this state, where $P_{A_j}$ is a projector on the classical system $A_j$; 2) On each of the $r$ systems, measuring $\mathcal{M}^{\otimes g}$ on $B_j$ and averaging the result; 3) Classically averaging the results across the $r$ systems and multiplying by $\gamma$.

Then let the output of this channel be denoted by $P_k$, for $k \in \{0,1\}^n$, while let $P_N$ denote $\mathcal{M}$ acting on the state

$$\sigma^{\otimes r} = \bigotimes_{j=1}^{r} \left( \sum_i p_i |i\rangle\langle i|_{A_j} \otimes \left( \frac{\mathbb{I}}{2^n} \right)^{\otimes g}_{B_j} \right) . \qquad (12)$$

Then we follow similar steps as before to obtain

$$\begin{aligned} D(P_k \| P_N) &= D\left( \mathcal{M}(\omega_k^{\otimes r}) \| \mathcal{M}\left(\sigma^{\otimes r}\right) \right) \leq r D\left(\omega_k \| \sigma\right) \\ &\leq r \sum_i p_i D\left( |i\rangle\langle i| \otimes (\mathcal{F}_i(|k\rangle\langle k|))^{\otimes g} \Big\| |i\rangle\langle i| \otimes \left(\frac{\mathbb{I}}{2^n}\right)^{\otimes g} \right) \\ &= rg \sum_i p_i D\left( \mathcal{F}_i(|k\rangle\langle k|) \Big\| \frac{\mathbb{I}}{2^n} \right). \end{aligned} \qquad (13)$$

It then suffices to upper-bound the expression (13). But by applying Markov's inequality to the probability distribution over the circuit ensemble used in our construction, Eq. (59) shows that there exist a set of circuits $\{\mathcal{F}_i\}$ on $n$ qubits and depth $d$ for which $D\left(\mathcal{F}_i(|k\rangle\langle k|)\|\frac{\mathbb{I}}{2^n}\right)$ is at most a constant factor times exponentially small in $n, d$, which, via the central lemma, yields the conclusion that the total number of circuit runs $rg$ can go up to exponential in $n, d$.

### A. Related work

We review what was known about theoretical limits of error mitigation prior to our work. Prior art typically falls into two camps: 'go' results, which propose new error mitigation algorithms and (often numerically) study their guarantees, and 'no-go' theorems, which demarcate the resource limitations of these algorithms. Before mentioning the more quantitative works on limitations of quantum error mitigation, it is worth noting the experimental motivation for such results: there is a body of heuristic evidence of instances of quantum error mitigation performing poorly in certain cases. For example, Ref. [13] discusses limitations of the extrapolation method to the noise-free case being unable to approach the mean of the error-free circuit. Ref. [14] has also shown that for *probabilistic error cancellation* assuming a sparse noise model, the sampling overhead scales exponentially in both the depth and the width of the circuit.

Concerning sample complexity bounds, Ref. [15] has presented a general framework. In terms of theoretical results, Refs. [15–17] also study the worst-case sample complexity of weak error mitigation under depolarizing noise from information-theoretic considerations, and they show that error mitigation incurs an $\exp(d)$ blowup in sample complexity. However, as explained in the main text, given existing circuit depths, their results are exponentially weaker than ours. They do not mention strong error mitigation. While we have made a case that our findings impose strong limitations also on input state-*aware* protocols, their results do not apply to such protocols. For these papers implicitly assume knowledge of the noiseless circuit—and with knowledge of the input state on top of that, an error mitigation algorithm could produce expectation value estimates without taking even a single sample, simply by simulating the circuit on the input state, even if it takes an exponential amount of computation.

Ref. [17] consider two desiderata for the error mitigation algorithm: a metric, similar to our Eq. (43), that requires the algorithm to be probably-approximately correct on *each* observable (note that this metric is marginally weaker than Eq. (43), which requires the algorithm to succeed with $1 - \delta$ probability at predicting expectation values for *all* observables) and the maximum bias/standard deviation of the estimators. While they also study the relative entropy decay to the maximally-mixed state, their ultimate sample complexity lower bound decays exponentially only in the *depth $D$* of the circuit—and polynomially in the number of qubits $n$. By contrast, our sample complexity lower bounds (Theorem 3) decay *exponentially* in both $n$ and $D$—indicating that their bounds are exponentially loose. Moreover, we make an incursion into proving our results for non-unital noise models.

Refs. [16, 17] studies weak error mitigation from the perspective of Fisher information. More precisely, the authors show that the quantum Fisher information associated with layered noisy circuits decays exponentially with circuit depth, which thus also lower bounds the variance of any unbiased estimator of expectation values. However, we have our reservations about using estimator variance as a metric of resources required, as it gives a very loose bound on the sample complexity to attain a certain estimation error. Take, for instance, the random variable which is $e^n$ with probability $e^{-n/2}$, $-e^n$ with probability $e^{-n/2}$ and $0$ otherwise. This zero-mean random variable has exponential variance, but with a sub-exponential number of samples, the empirical mean is overwhelmingly likely to be exactly the true mean. Of course, such pathological examples use the fact that the underlying random variables take exponentially large values, but this is also the case the case for some error mitigation protocols.

## II. CONCEPTUAL BACKGROUND AND CONTRIBUTIONS

We first show a basic version of our theorem bounding the sample complexity of error mitigation in a circuit subject to depolarizing noise. This theorem already has all the features of our main results but is weaker, in the sense that we will only be able to show an exponential-in-$D$ dependence of the sample complexity, while our main result improves this to an exponential in both $n, D$.

### A. Appetizer: mitigating depolarizing noise requires exponential-in-$D$ samples

We are finally ready to present our 'appetizer' theorem. Let us first establish the following notation:

**Definition 1** (Circuit with depolarizing noise). *For any noiseless circuit $\mathcal{C}$ of depth $D$ acting on $n$ qubits, define the noisy circuit with depolarizing noise parameter $p$, $\Phi_{\mathcal{C},p}$, as*

$$\Phi_{\mathcal{C},p} := \mathcal{C}^{(D)} \circ \mathcal{D}_p^{\otimes n} \cdots \circ \mathcal{C}^{(2)} \circ \mathcal{D}_p^{\otimes n} \circ \mathcal{C}^{(1)}, \tag{14}$$

*where $\mathcal{C}^{(i)}$ represents the $i^{\text{th}}$ layer of the circuit $\mathcal{C}$, and $\mathcal{D}_p$ is a local, single-qubit depolarizing channel.*

We will now present the key reduction we will use to bound the sample complexity of error mitigation. This reduction is visually represented in Fig. 2. To do so, let us step away from the error mitigation setting for a moment, and first consider the following problem of *state discrimination in the presence of noise*:

**Problem 1** (State discrimination in the presence of noise). *Fix a set $\{\rho_0, \rho_1, \ldots, \rho_N\}$ of $N+1$ quantum states and a circuit $\mathcal{C}$ acted upon by noise $\mathcal{N}$. Suppose that a distinguisher has knowledge of $\mathcal{C}$ and $\mathcal{N}$, and is given access to copies of the state $\Phi_{\mathcal{C},\mathcal{N}}(\rho_i)$, with the index $i \in \{0, 1, \ldots, N\}$ unknown. What is the fewest number of copies of $\Phi_{\mathcal{C},\mathcal{N}}(\rho_i)$ needed in order to identify $i$ with high probability, in terms of the noise strength?*

The following variation of Fano's lemma can be used to address this question:

**Lemma 2** (Corollary 2.6 of Ref. [1]). *Let $P_0, \ldots, P_N$ be probability measures on some state space $X$ such that*

$$\frac{1}{N+1} \sum_{k=0}^{N} D(P_k \| P_0) \leq \alpha \log(N) \tag{15}$$

*for $0 < \alpha < 1$. Then the minimum average probability of error over tests $\psi : X \to \{0, 1, \ldots, N\}$ that distinguish the probability distributions $P_0, \ldots, P_N$ which we define as*

$$\bar{p}_{e,N} := \inf_\psi \frac{1}{N+1} \sum_{j=0}^{N} P_j(\psi \neq j) \tag{16}$$

*satisfies*

$$\bar{p}_{e,N} \geq \frac{\log(N+1) - \log(2)}{\log(N)} - \alpha. \tag{17}$$

Let us explain what this lemma implies about the sample complexity of state discrimination in the presence of noise (Problem 1). Intuitively, dialing up the noise on the circuit $\mathcal{C}$ makes the various output states $\Phi_{\mathcal{C},\mathcal{N}}(\rho_i)$ less distinguishable from each other—and so, more copies of $\Phi_{\mathcal{C},\mathcal{N}}(\rho_i)$ must be taken to determine the true identity of the input state. This Lemma formalizes this intuition. For let $P_0, \ldots, P_N$ be distributions over the output of the $m$-copy distinguisher in Problem 1 when the true state is $\rho_0 \ldots \rho_N$, respectively. One can then show that, for $\mathcal{N}$ being depolarizing noise of parameter $p$ on generic depth-$D$ circuits, the average distance of $P_1, \ldots, P_N$ to $P_0$ (the left-hand-side of (15)) scales as $\alpha \log(N) = p^{2D} \log(N)$. To get a constant probability of error in (17), $\alpha$ must be constant. In order for $\alpha$ to be constant, one needs to choose $m$ scaling as $p^{-2D}$.

We can go one step further and use Lemma 2 to bound the performance of error mitigation. The proof proceeds according to the following two basic steps:

1. Embed an error mitigation protocol into a solver for a state discrimination problem. This means, we identify a state discrimination problem such that an error mitigation algorithm can be transformed, with no additional copies required, into a distinguisher that solves the state discrimination problem.

2. Compute Fano's lower bound for the number of copies needed for the state discrimination problem.

More formally, we now state and prove the following theorem:

**Theorem 1** (Appetizer: Error mitigation under depolarizing noise requires exponential-in-depth samples)**.** *Let $\mathcal{A}$ be an input state-agnostic weak error mitigation algorithm (see Definition 5) that takes as input $m$ noisy copies $\Phi_{\mathcal{C},p}(\rho)$. Suppose that for any circuit $\mathcal{C}$ acted upon by depolarizing noise of parameter $p$, and any set of observables $\mathcal{M} = \{O_i\}_{i \in [n]}$, $\mathcal{A}$ is able to produce estimates $\{\hat{o}_i\}_{i \in [n]}$ such that with probability at least $1 - \delta$,*

$$|\hat{o}_i - \mathrm{Tr}\left[O_i \mathcal{C}(\rho)\right]| \leq \epsilon \qquad \text{for all } i \in [n]. \tag{18}$$

*Then there exists an observable set $\mathcal{M}$ and a set of input states $\rho$ such that, as long as $\epsilon < 1/2$, the number of noisy copies, $m$, needed by $\mathcal{A}$, is lower-bounded as*

$$m = \Omega(p^{-2D}) \tag{19}$$

*for any circuit $\mathcal{C}$ on $n$ qubits of depth $D$.*

After we prove this theorem, we remark on extensions of our proof to EM protocols with other types of input state access, including those protocols that are input state-aware.

*Proof of Theorem 1.* Let $\mathcal{A}$ be an input state-agnostic weak error mitigation algorithm and consider the observable set $\mathcal{M} = \{\mathcal{C}(Z_j)\}_{j=1}^n$ where $Z_j$ is the Pauli $Z$ acting on qubit $j$. For ease of notation, we will denote by $\hat{z}_j$ the estimate of $\mathrm{Tr}[\mathcal{C}(Z_j)\mathcal{C}(\rho_i)]$. We are going to identify a set $S$ of states (which is also the set of states referenced in the Theorem), such that to solve Problem 1 when the unknown state is chosen from $S$, it suffices to run $\mathcal{A}$ on copies of the unknown state and then do simple post-processing.

*a. Step 1: A noisy state identification problem that can be solved by weak error mitigation.* Consider choosing the unknown state $\rho_i$, $i \in [2^n + 1]$, in Problem 1 from the following $N + 1 = 2^n + 1$ options:

- $\rho_N = \mathbb{I}/2^n$, the maximally-mixed state. If this option is chosen, say we are in **Case 1**.

- $\{\rho_x := |x\rangle\langle x|\}_{x \in \{0,1\}^n}$. If one of these options is chosen, say we are in **Case 2**.

(We emphasize that the state discrimination algorithm must output $i$, the label of the state, rather than simply decide whether it is in Case 1 or 2.) We will now show that Problem 1 can be solved with access to an error mitigation algorithm $\mathcal{A}$. First suppose hypothetically that the error mitigation algorithm were *perfect*; i.e., when given $m$ copies of $\Phi_{\mathcal{C},\mathcal{N}}(\rho_i)$, $\mathcal{A}$ outputs estimates

$$\hat{z}_j = \mathrm{Tr}[\mathcal{C}^\dagger(Z_j)\mathcal{C}(\rho_i)] = \mathrm{Tr}[Z_j \rho_i] \qquad \forall j \in [n] \tag{20}$$

that are *exact*. Later we will show that the analysis below goes through if we assume only approximate estimates. But consider that, in Case 1, where $\rho_i$ is $\rho_N = \mathbb{I}/2^n$, the true expected value of the observables is

$$\mathrm{Tr}[Z_j \mathbb{I}/2^n] = 0 \qquad \forall j \in [n] \tag{21}$$

i.e., a perfect error mitigation algorithm should estimate $\hat{z}_j = 0$ for every observable.

On the other hand, in Case 2, when $\rho_i = \rho_x$ for some $x = (x_1, x_2, \ldots, x_n) \in \{0,1\}^n$, the true expected value of the observables is

$$\mathrm{Tr}[|x\rangle\langle x|Z_j] = 2x_j - 1 \qquad \forall j \in [n] \ , \tag{22}$$

i.e., a perfect error mitigation algorithm would output the estimate $\hat{z}_j = 1$ if $x_j = 1$ and $\hat{z}_j = -1$ if $x_j = 0$. Hence, if the error mitigation algorithm were perfect, one could perfectly distinguish Case 1 and Case 2; furthermore in the latter case, concatenating the set $\{(1 + \hat{z}_j)/2\}_{j=1}^n$ results in an $n$-bit string corresponding to $x = (x_1, x_2, \ldots, x_n)$, the label of the actual input state. A perfect error mitigation algorithm thus solves the state identification problem perfectly.

Now we argue that state identification can still be solved by an error mitigation algorithm that returns only *approximate* estimates. By assumption, with probability $1 - \delta$, the estimates are accurate up to additive error $\epsilon$. As long as $\epsilon < 1/2$, we find

$$|\hat{z}_j - \mathrm{Tr}[\rho_i Z_j]| \leq 1/2 \qquad \forall j \in [n] \ , \tag{23}$$

(c.f. Eq. (20)). The *true* values, of course, remain the same: Eq. (21) and Eq. (22) still hold in Case 1 and Case 2 respectively. In particular, let $\hat{z} = (\hat{z}_1, \hat{z}_2, \ldots, \hat{z}_n)$ be the estimates produced by $\mathcal{A}$ when the input state is $\rho_N = \mathbb{I}/2^n$ (Case 1) and let

$\hat{z}_x = (\hat{z}_{1,x}, \hat{z}_{2,x}, \ldots, \hat{z}_{n,x})$ be the estimates produced by $\mathcal{A}$ when the input state is $\rho_x = |x\rangle\langle x|$ (Case 2). Then Eq. (23) says that with probability $1 - \delta$,

$$-\frac{1}{2} < \hat{z}_i < \frac{1}{2} \;\forall i \quad \text{(Case 1)} \quad \text{and} \quad \left\{ \begin{array}{ll} \frac{1}{2} < \hat{z}_{j,x} < \frac{3}{2} & \text{if } x_j = 0, \\ -\frac{3}{2} < \hat{z}_{j,x} < -\frac{1}{2} & \text{if } x_j = 1. \end{array} \text{(Case 2)} \right\}. \tag{24}$$

Now, let EstInput be the algorithm that takes the output of $\mathcal{A}$, which is the set of estimates $\{\hat{z}_j\}_{j \in [n]}$, and does the following:

1. If there exists $j \in [n]$ such that $\hat{z}_j \in (-1/2, 1/2)$, it outputs that $\hat{x} = N$, i.e., the unknown state is $\rho_N$.

2. Else, it computes the $n$-bit string $\hat{x} = (\hat{x}_1, \hat{x}_2, \ldots, \hat{x}_n) \in \{0,1\}^n$ as

$$\hat{x}_j = \left\{ \begin{array}{ll} 0 & \text{if } \hat{z}_j \geq 0 \\ 1 & \text{if } \hat{z}_j < 0 \end{array} \right\}, \tag{25}$$

and outputs that the unknown state is $\rho_{\hat{x}}$.

It is easy to see from Eq. (24) that this procedure correctly identifies the unknown state, and hence solves Problem 1 with success probability $1 - \delta$.

*b. Step 2: Computing Fano's lower bound for state identification.* The string $\hat{x}$ output by the above procedure is probabilistic. In particular, let $P_i(\hat{x})$ for $i \in [2^n + 1]$ be the probability that running EstInput $\circ \mathcal{A}$ on the unknown state $\rho_i$ will output $\hat{x}$. Since the last step of our procedure is to sample from $P_i$ and then report the output, this last step is a single-sample test to distinguish probability distributions. By Fano's Lemma (Lemma 2), any such test must fail with probability no less than

$$\frac{\log(N+1) - \log(2)}{\log(N)} - \alpha \xrightarrow{N \text{ large}} 1 - \alpha \tag{26}$$

where $\alpha \log(N)$ is defined to be an *upper bound* to the quantity

$$\frac{1}{N+1} \sum_{k=0}^{N} D(P_k \| P_N) \quad . \tag{27}$$

We will show later (in Proposition 1) that we can bound the average relative entropy between $P_k$ and $P_N$ as follows:

$$\frac{1}{N+1} \sum_{k=0}^{N} D(P_k \| P_N) \leq p^{2D} m \log(N), \tag{28}$$

so it suffices to set $\alpha = p^{2D} m$. Therefore, in order for the test to have a constant failure probability $\delta$, it must take at least $m = p^{-2D}(1-\delta)$-many samples. $\square$

Eq. (27) is a critical quantity in Fano's Lemma and our argument, so let us interpret it: take the maximally-mixed state $\mathbb{I}/2^n$ to be the reference state and the probability distribution output by our noisy state discrimination algorithm on $\mathbb{I}/2^n$ to be the reference probability distribution $P_N$. Eq. (27) examines the divergence between the output of the algorithm initialized with one of the states $\rho_x$, and that when initialized on the reference state. Intuitively, if this quantity is high, it indicates that the states are "far apart", and hence the distinguishing algorithm could potentially have a lower probability of error.

We remark that in our construction we have used the error mitigation protocol on the observables $C(Z_i)$. These observables will in general be supported on almost all qubits and we can efficiently compute their expectation value for all quantum circuits. So the reader might argue that these expectation values are unlikely to be the target of error mitigation protocols and they are usually designed with local observables in mind. However, note that the only property of the observables we have used in the proof was that they allow us to distinguish different input states from the maximally mixed state. Thus, for any set of observables whose expectation values differ for outputs of the noiseless circuit and the maximally mixed state, the same proof goes through. This is because the noiseless expectation values can be used to envision a statistical test to distinguish the output states.

Indeed, bounding the distinguishability of the states in our set (which becomes the left-hand-side of Eq. (27) after one application of data-processing) is the technical crux of our argument. We present a basic bound in Proposition 1 and we will gradually refine this bound in the next few sections to get progressively more severe dependencies of the sample complexity of error mitigation on various parameters of the problem. In this Proposition, we will suppress the notation for the classical inputs $\mathcal{C}, p, \mathcal{M}$ to the error mitigation algorithm $\mathcal{A}$ (see Def. 4) and use $\mathcal{A}(\Phi_{\mathcal{C},p}(\rho)^{\otimes m})$ to denote $\mathcal{A}$ with all the above classical inputs, and a quantum input which is $m$ copies of the quantum state $\Phi_{\mathcal{C},p}(\rho)$.

**Proposition 1** (Simple bound on average divergence). *Let $P_k$, $k \in \{0,1\}^n$ be the probability distributions over the outputs of the algorithm for noisy state discrimination described in the proof of Theorem 1, run on the input state $\rho_k$. That is, $P_k$ for $k \in \{0,1\}^n$ is the probability distribution output by*

$$\mathsf{EstInput} \circ \mathcal{A}(\Phi_{\mathcal{C},p}(\rho_k)^{\otimes m}) \quad , \tag{29}$$

*and similarly let $P_N$ where $N = 2^n$ be the probability distribution over the output of*

$$\mathsf{EstInput} \circ \mathcal{A}(\Phi_{\mathcal{C},p}(\rho_N)^{\otimes m}) \quad . \tag{30}$$

*Then it holds that*

$$\frac{1}{N+1} \sum_{k=0}^{N} D(P_k \| P_N) \leq p^{2D} m \log(N). \tag{31}$$

*Proof.* Firstly, note that since $\rho_N \coloneqq (\mathbb{I}/2^n)^{\otimes m}$, we can explicitly write down the output of the noisy circuit on this state as $\Phi_{\mathcal{C},p}(\rho_N) = (\mathbb{I}/2^n)^{\otimes m}$, as the maximally-mixed state is the fixed point of depolarizing noise. Since $P_k = \mathsf{EstInput} \circ \mathcal{A}(\Phi_{\mathcal{C},p}(\rho_k)^{\otimes m})$ and $P_N = \mathsf{EstInput} \circ \mathcal{A}((\mathbb{I}/2^n)^{\otimes m})$, by the data-processing inequality we have

$$D(P_k \| P_N) \leq D(\Phi_{\mathcal{C},p}(\rho_k)^{\otimes m} \| (\mathbb{I}/2^n)^{\otimes m}), \tag{32}$$

and so it suffices to bound the quantity $D(\Phi_{\mathcal{C},p}(\rho_k)^{\otimes m} \| (\mathbb{I}/2^n)^{\otimes m})$. The rest of this proof shows that for $k \in \{0,1\}^n$:

$$D(\Phi_{\mathcal{C},p}(\rho_k)^{\otimes m} \| (\mathbb{I}/2^n)^{\otimes m}) \leq p^{2D} nm \tag{33}$$

by the unitary invariance of the relative entropy, the strong data-processing inequality for depolarizing noise and the additivity of the relative entropy. To see this, observe that

$$\begin{aligned} D(\Phi_{\mathcal{C},p}(\rho_k)^{\otimes m} \| (\mathbb{I}/2^n)^{\otimes m}) &= m D(\Phi_{\mathcal{C},p}(\rho_k) \| \mathbb{I}/2^n) \tag{34} \\ &\leq m\, p^{2D} D(\rho_k \| \mathbb{I}/2^n) \\ &= m\, p^{2D} D(|k\rangle\langle k| \| \mathbb{I}/2^n) \\ &= p^{2D} mn, \tag{35} \end{aligned}$$

where the first equality follows from the additivity of relative entropy, the first inequality is strong data-processing for $D$ layers of depolarizing noise with parameter $p$ [18],

$$D(\Phi_{\mathcal{C},p}(\rho_k) \| \mathbb{I}/2^n) \leq p^{2D} D(\rho_k \| \mathbb{I}/2^n), \tag{36}$$

and the third equality follows from direct calculation. Lastly, we clearly have $D(P_N \| P_N) = 0 < p^{2D} nm$. Therefore,

$$\frac{1}{N+1} \sum_{k=0}^{N} D(P_k \| P_N) < p^{2D} nm \tag{37}$$

as desired. □

Though it follows immediately, we emphasize that inequality (32) still holds in the case that the error mitigation algorithm is not *entirely* classical, but performs some (usually small amount of) quantum post-processing on the noisy states – as is the case with virtual distillation [2], where a single layer of two-qubit gates is applied on the noisy states. Thus, all our conclusions still hold in this setting.

Theorem 1 states that a general input state-*agnostic* error mitigation protocol should have a worst-case sample complexity scaling exponentially in $D$. So far, our proof method does *not* encompass input state-*aware* error mitigation protocols, which are those that, in addition to knowing the noiseless circuit, also know its input state. Such knowledge would be an unauthorized 'back door' for the error mitigation algorithm used inside the solver for noisy state discrimination.

Can we remedy this lack? We offer two affirmative answers. The first follows immediately from close inspection of our proof method: our results extend to even error mitigation protocols working with a guarantee that the input state is chosen from a known set of size $N+1$—a situation somewhere between input-state agnosticism and full input-state awareness. Secondly, we will next show that a modification of our argument allows us to make a statement about the complexity of fully input-state *aware* error mitigation algorithms.

## B. The input-state *aware* case

The input-state *aware* setting refers to settings where, in doing the translation from the real-life setting to the mathematical framework we have defined in Problem 1, the error mitigation protocol does not return the correct results if the input state to the noisy circuit is chosen arbitrarily. This could be because the error mitigation algorithm is designed with a particular, fixed circuit input state 'hard-coded' into the computation that is done on the noisy output states of the circuit. Clifford data regression in I is one example of such a protocol.

Without further assumptions, we cannot expect lower bounds on the sample complexity of the general error mitigation task in the input-state aware setting, because there always exists a trivial error mitigation algorithm that does not access the quantum device—hence uses zero samples—and instead simulates the quantum circuit to obtain the desired properties of the noiseless state. This would not be possible in the setting of the previous subsection, because there the error mitigation algorithm is restricted to not use information about the input state to the circuit. While we previously found sample complexity bounds on *any* input state-agnostic error mitigation algorithm outputting expectation value estimates, we now change our perspective: we ask how many samples would be needed by an error mitigation algorithm whose outputs *are meaningfully different* from one that obtains comparable results through purely classical means. In other words, if an error mitigation algorithm is successful but it invokes the quantum device an insufficient number of times, then there exists an equivalent classical algorithm. A tool we will use is the following bound on two-hypothesis testing due to Le Cam:

**Lemma 3** (Two-hypothesis testing [1], Theorem 2.2iii)). *For two probability distributions $P_0, P_1$ satisfying*

$$\|P_0 - P_1\|_{TV} \leq \alpha \quad, \tag{38}$$

*the minimum probability of error of any test that draws a single sample and distinguishes $P_0$ and $P_1$ is*

$$P_{error} \geq \frac{1-\alpha}{2} \quad. \tag{39}$$

Our results on the convergence of random quantum circuits to the maximally mixed state can easily be used to obtain bounds in the total variation distance by an application of Pinker's inequality and the data processing inequality. In the following theorem, we prove that our construction can be used to make a statement about the sample complexity of input state aware error mitigation under depolarizing noise, as long as one demands that such error mitigation makes meaningful use of the quantum inputs. We then remark on how our technique can be extended even to the case of non-unital noise.

**Theorem 2** (Input state-awareness does not break the curse of exponentiality). *Let $\mathcal{A}$ be an input state-*aware* weak error mitigation algorithm that takes:*

- Quantum *input: At most $m$ copies of the noisy output state $\Phi_{\mathcal{C},\mathcal{N}}(\rho)$.*

- Classical *input, $CL(\mathcal{C}, \mathcal{N}, \rho)$: Classical descriptions of $\mathcal{C}$, the noise channel $\mathcal{N}$ acting on $\mathcal{C}$ and* the input state to $\mathcal{C}$, which we call $\rho$,

*and outputs a distribution over possible outcomes. Then either:*

- **Case I:** *For some $c > 1$, $mc^{-\tilde{\mathcal{O}}(nD)} = \Omega(1)$, i.e., even for input-state aware error mitigation algorithms, exponential in $n, D$ samples are needed; or*

- **Case II:** *no algorithm exists can distinguish between the outputs of $\mathcal{A}(\Phi_{\mathcal{C},\mathcal{N}}(\rho)^{\otimes m}, CL(\mathcal{C}, \mathcal{N}, \rho))$ and $\mathcal{A}((\mathbb{I}/2)^{\otimes m}, CL(\mathcal{C}, \mathcal{N}, \rho))$ with an error probability significantly bounded away from $2/3$.*

*Proof.* Suppose we are not in Case I. This implies that $mc^{-\tilde{\mathcal{O}}(nD)} = \Theta(1)$. In Sections III and IV we proved there exists a circuit $\mathcal{C}^*$ such that for any algorithm $\mathcal{A}$ run on either $m$ copies of the state $\Phi_{\mathcal{C}^*,\mathcal{N}}(\rho)$ or $m$ copies of the maximally-mixed state,

$$D(\mathcal{A}(\Phi_{\mathcal{C}^*,\mathcal{N}}(\rho)^{\otimes m}, CL(\mathcal{C}^*, \mathcal{N}, \rho))\|\mathcal{A}(\mathbb{I}/2^{n\otimes m}, CL(\mathcal{C}^*, \mathcal{N}, \rho)))$$
$$\leq mD(\Phi_{\mathcal{C}^*,\mathcal{N}}(\rho)\|\mathbb{I}/2^n) \leq mc^{-\tilde{\mathcal{O}}(nD)} = \Theta(1) \tag{40}$$

therefore, Lemma 3 implies that the minimum probability of error distinguishing the output of $\mathcal{A}(\Phi_{\mathcal{C}^*,\mathcal{N}}(\rho)^{\otimes m}, CL(\mathcal{C}, \mathcal{N}, \rho))$ from the output of $\mathcal{A}((\mathbb{I}/2^n)^{\otimes m}, CL(\mathcal{C}, \mathcal{N}, \rho))$ is $1/2 - o(1)$, proving the claim. □

The statement above implies that if the error mitigation was successful, then we could have supplied the algorithm with copies of the maximally mixed state as an input and observed similar outcomes.

We can prove a similar statement in the case of non-unital noise. However, unitary circuits affected by product non-unital noise are qualitatively different: whereas in the case of unital circuits the outputs will converge to the maximally mixed state

asymptotically, non-unital random circuits do not converge to a single state. For instance, as we will show later in the setting of our toy example of noise being applied after $2-$design, what happens instead is that after we apply $D$ layers of $2-$designs followed by noise, yielding the evolution $T_D$, we have

$$\mathbb{E}\left[\|T_D(\rho_i - \rho_j)\|_1\right] \leq 2^{\frac{n+1}{2}} q_n^{\frac{D}{2}}, \tag{41}$$

where $q_n = \mathcal{O}(c_1^n)$ for $c_1 < 1$ is an exponentially small in $n$ parameter that depends on the noise channel and $\rho_i, \rho_j$ orthogonal inputs. That is, any two inputs will be mapped to essentially the same state $\sigma_D$ after $D$ applications of the circuit. However, it is in general not possible to compute this state efficiently to simulate it classically and unless the noise is unital the sequence $\sigma_1, \sigma_2, \ldots$ does not necessarily converge. Thus, in the argument we have used above, the maximally mixed state would have to be replaced by this state $\sigma_D$. Thus, the same conclusions would hold, but the output of the error mitigation algorithm would be indistinguishable from the output given as input copies of a fixed state $\sigma_D$ that is independent of the input $\rho$. We cannot exclude the possibility that sampling from $\sigma_D$ gives some extra computational power, but as in general there is no way to compute $\sigma_D$, it is unclear how to profit from it. In summary, the output of the error mitigation algorithm given the samples from the noisy device would be indistinguishable from that of the same error mitigation algorithm given as input copies of a state $\sigma_D$ that is independent of the input.

As an interlude, let us review the conceptual contribution of this section and look ahead to the next. Theorem 1 leveraged a relation between a noisy state distinguishability problem and a related error mitigation problem to show that the worst-case sample complexity of error mitigation needs to depend exponentially on the depth of the circuit whose errors are being mitigated. The proof of this theorem depended crucially on upper-bounding the quantity

$$\frac{1}{N+1} \sum_{k=0}^{N} D(\Phi_{\mathcal{C},p}(\rho_k) \| \Phi_{\mathcal{C},p}(\mathbb{I}/2^n)) \quad , \tag{42}$$

which is a sort of average distance between the output of the noisy circuit (whose errors we would like to mitigate) and its fixed point, which is the maximally-mixed state. This quantity controls the sample complexity of noisy state discrimination via Fano's Lemma, and in the chain of inequalities (34) to (35) of the proof of Proposition 1, we upper-bounded it by $\alpha \log N$ where we computed $\alpha$ in terms of $m$ and $D$. Plugging into Fano's Lemma then completed the argument and yielded the lower bound on sample complexity $m$.

In the next section, we will refine the upper-bound on (42), so that $\alpha$ has an exponential dependence not only on $D$, but also on $n$. Plugging this into Fano's Lemma yields, in turn, a worst-case sample complexity $m$ that depends exponentially on *both* of these quantities. We will achieve this improvement by carefully constructing the input states $\mathcal{C}^{-1}(|x\rangle\langle x|)$. In fact, the key will be a perspicacious choice of $\mathcal{C}$, which leads to a worst-case-bound *better* than the bound on the right-hand-side of Eq. (36), which has been proven for general $\mathcal{C}$.

### III. RAPIDLY MIXING CIRCUITS ARE HARD TO ERROR MITIGATE

In this section, we show that there exists a 'rapidly mixing' quantum circuit $\mathcal{C}^*$. In so doing we prove a stronger version of our results, which we state in Theorem 4. This circuit is rapidly mixing in the sense that, its output states

$$\Phi_{\mathcal{C}^*,p}(\mathbb{I}/2^n \otimes |0\rangle\langle 0|^{\otimes s}), \Phi_{\mathcal{C}^*,p}(\rho_k \otimes |0\rangle\langle 0|^{\otimes s}), \quad k \in \{0,1\}^n, \tag{43}$$

for the input states $\{\rho_k\}$ and the maximally mixed state (concatenated with auxiliary qubits), satisfies

$$\frac{1}{N+1} \sum_{k=0}^{N} D(\Phi_{\mathcal{C}^*,p}(\rho_k \otimes |0\rangle\langle 0|^{\otimes s}) \| \Phi_{\mathcal{C}^*,p}(\mathbb{I}/2^n \otimes |0\rangle\langle 0|^{\otimes s})) < c^{\Omega(nD)}, \tag{44}$$

for some constant $c < 1$ (c.f. Eq. (42)). Here the notation $\Phi_{\mathcal{C}^*,p}$ refers to a noisy version of the circuit $\mathcal{C}$.

Let us compare this to what we obtained in the previous section. In Eq. (37), for a generic $\mathcal{C}$, we could only prove that the average distinguishability could be bounded by $p^{\Omega(D)}nm$, which does not have an $n$ dependence in the exponent. Thus when $D = O(\log n)$, the old bound is exponentially weaker in $n$. After we prove the improved bound for our particular circuit $\mathcal{C}^*$, we will conclude by again plugging this bound into the information-theoretic outer loop of our argument, which uses the sample complexity of noisy state distinguishability to lower bound that of error mitigation.

### A. Technical background

In fact, we do not construct $\mathcal{C}^*$ explicitly. Instead, we let $\mathcal{C}^*$ be a concatenation of blocks, i.e.,

$$\mathcal{C}^* = \widetilde{\mathcal{C}}_\ell \circ \cdots \circ \widetilde{\mathcal{C}}_2 \circ \widetilde{\mathcal{C}}_1, \tag{45}$$

where every *block* is a circuit $\widetilde{\mathcal{C}}_j$ sampled independently from an *ensemble* of circuits (described in Lemma 4). We will in fact show that the noisy version of $\mathcal{C}^*$, when run on the input states $\mathbb{I}/2^n$ and $\rho_k$, $k \in \{0,1\}^n$, on expectation displays the scaling we claim in Eq. (44). We can then use the probabilistic method to conclude that there *exists* one particular circuit that achieves the same bound. Every block in $\mathcal{C}^*$ is constructed via the following prescription by Cleve, Leung, Liu and Wang [19]:

**Lemma 4** (Exact unitary 2-designs from Clifford gates [19]). *With circuits $\widetilde{\mathcal{C}}$ consisting of $O(n \log^2 n \log \log n)$ single- and two-qubit Clifford gates and $\tilde{\mathcal{O}}(n)$ auxiliary qubits, in $O(\log^2(n))$ depth, it is possible to implement an* exact *unitary 2-design on $n$ qubits. The auxiliary qubits all start in the state $|0\rangle\langle 0|$ and are returned to this state at the end of the circuit.*

In other words, the above is a prescription for constructing Clifford circuits $\widetilde{\mathcal{C}}$ that form unitary 2 designs on a subset of the qubits. Each such circuit has the property that

$$\widetilde{\mathcal{C}}(\rho \otimes |0\rangle\langle 0|^{\otimes s}) = \mathcal{C}(\rho) \otimes |0\rangle\langle 0|^{\otimes s}, \tag{46}$$

where $s$ is the number of auxiliary qubits and $\mathcal{C}$ is a unitary sampled from an exact unitary 2-design.

### B. Dealing with noise within the circuits

In this subsection, we wish to consider a noisy version of the circuit $\mathcal{C}^*$ defined in Eq. (45). This means that in between every layer of single- and two-qubit Clifford gates comprising each circuit $\widetilde{\mathcal{C}}_j$, there is a layer of depolarizing noise on every qubit. In particular, the noisy version of the block $\widetilde{\mathcal{C}}_j$ is of the form

$$\Phi_{\widetilde{\mathcal{C}}_j, \vec{p}_j} = \mathcal{D}_{p_{D'}}^{\otimes(n+s)} \circ \widetilde{\mathcal{C}}_j^{(D')} \circ \cdots \circ \mathcal{D}_{p_2}^{\otimes(n+s)} \circ \widetilde{\mathcal{C}}_j^{(2)} \circ \mathcal{D}_{p_1}^{\otimes(n+s)} \circ \widetilde{\mathcal{C}}_j^{(1)}, \tag{47}$$

where $\vec{p}_j = (p_1, p_2, \ldots, p_{D'})$ is a vector of single-qubit depolarizing noise parameters and $D'$ is the depth of the construction of Ref. [19], as specified in Lemma 4 above. The superscript indexes the layer of the circuit and the subscript indexes the block. The noisy version of the overall circuit $\mathcal{C}^*$, which is the subject of our analysis, is then

$$\Phi_{\mathcal{C}^*, \vec{p}} := \Phi_{\widetilde{\mathcal{C}}_\ell, \vec{p}_\ell} \circ \cdots \circ \Phi_{\widetilde{\mathcal{C}}_2, \vec{p}_2} \circ \Phi_{\widetilde{\mathcal{C}}_1, \vec{p}_1}. \tag{48}$$

We now show that we can simplify the form of this noisy circuit by making use of an important property of Clifford gates, which can even be taken to be their defining property, which is that they map Pauli operators to other Pauli, up to a sign. We state this useful property more formally in Lemma 5, which leads to the simplified form of the circuit presented in Lemma 6 and illustrated in Figure 1.

**Lemma 5** (Pushing a Clifford past a Pauli channel). *Let $\mathcal{P}$ be a Pauli channel and let $\mathcal{C}$ be a Clifford unitary. Then there exists a Pauli channel $\mathcal{P}'$ such that*

$$\mathcal{C} \circ \mathcal{P} = \mathcal{P}' \circ \mathcal{C}. \tag{49}$$

*Proof.* Let $\mathcal{P}' = \mathcal{C} \circ \mathcal{P} \circ \mathcal{C}^\dagger$. It is clear that Eq. (49) holds. Furthermore, we see that letting $\mathcal{P} = \sum_{P \in \{I,X,Y,Z\}^{\otimes n}} q_P P(\cdot) P$, then

$$\mathcal{P}'(\rho) = \sum_{P \in \{I,X,Y,Z\}^{\otimes n}} q_P C P C^\dagger (\cdot) C^\dagger P C.$$

As $C$ is a Clifford unitary, $CPC^\dagger$ will also be a Pauli string, up to a sign. Thus, we conclude that the Kraus operators of $\mathcal{P}'$ are also Paulis and it is a Pauli channel. □

Invoking this property (and noting that depolarizing channels are Pauli channels), we may rewrite $\Phi_{\mathcal{C}^*, \vec{p}}$ in a way that is more convenient to us: namely we rewrite the circuit so that part of the noise has been pushed to the end of the circuit. Note that we retain the noise on the system.

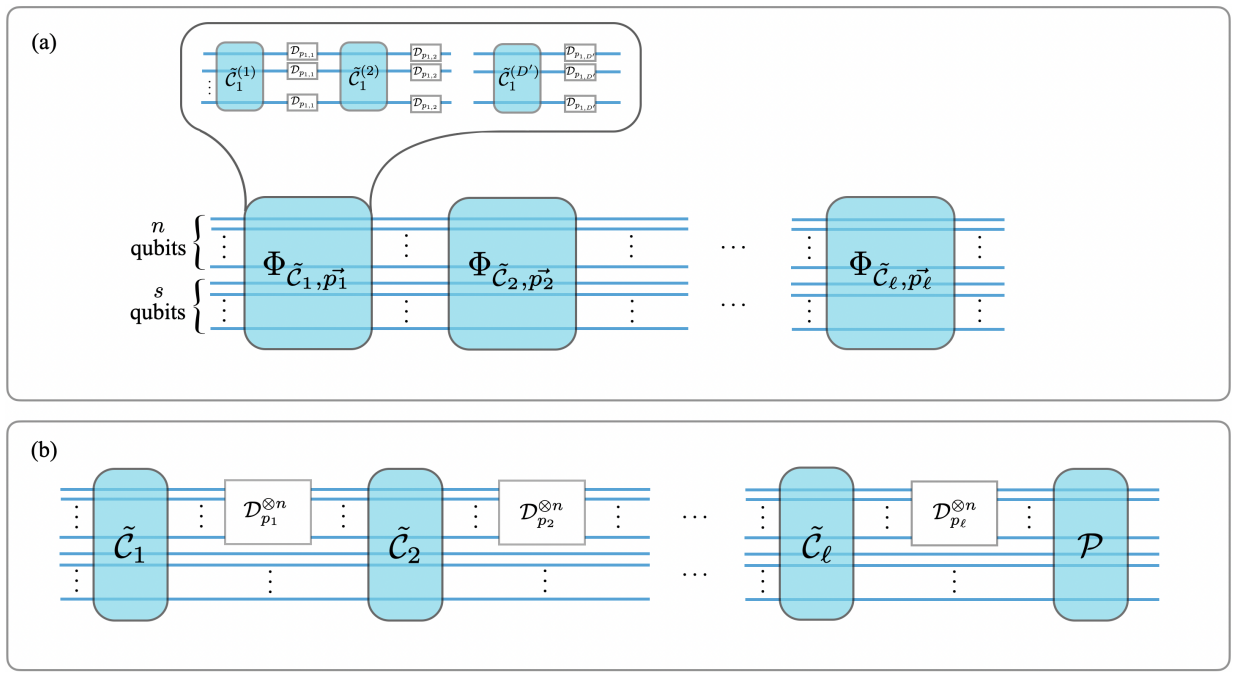

FIG. 1. A sketch of the constructions used in this subsection. Figure a) depicts the noisy circuit that results from our prescription of iterating the construction in Ref. [19] $\ell$ times (refer to Eq. (47) and (48)) and using our model of local depolarizing noise. Figure b) depicts how we can re-write the circuit in Figure a) by pushing the depolarizing noise on the auxiliary qubits to the end of the circuit, as stated in Lemma 6.

**Lemma 6** (Pushing Paulis to the end of the circuit). *Consider the noisy circuit on $n+s$ qubits defined in Eq. (48). There exists a Pauli channel $\mathcal{P}_{n+s}$, acting on all $n+s$ qubits, such that the circuit*

$$\mathcal{P}_{n+s} \circ (\mathcal{D}_{p_\ell}^{\otimes n} \otimes \mathrm{id}_s) \circ \widetilde{\mathcal{C}}_\ell \circ \cdots \circ (\mathcal{D}_{p_2}^{\otimes n} \otimes \mathrm{id}_s) \circ \widetilde{\mathcal{C}}_2 \circ (\mathcal{D}_{p_1}^{\otimes n} \otimes \mathrm{id}_s) \circ \widetilde{\mathcal{C}}_1 \tag{50}$$

*is equivalent to $\Phi_{\mathcal{C}^*,\vec{p}}$. In other words, all of the noise acting on the $s$ auxiliary qubits and within the circuit for each design can be pushed to the end, and the intermediate 2-designs $\widetilde{\mathcal{C}}_i$ can be treated as noiseless.*

*Proof.* First, note that all gates of the circuit are Clifford and all the noisy channels are assumed to be product of $1-$ qubit Pauli noise channels. As observed in Lemma 5, given a Pauli channel followed by a Clifford gate, we can always obtain the same transformation by first applying the gate followed by a (typically different) Pauli channel. Furthermore, Pauli channels all commute with each other. To see this, note that Pauli strings are the eigenvectors of all Pauli channels. As they share a common eigenbasis, they all commute.

Thus, we can take all the individual depolarizing channels either acting on the auxiliary systems or within the blocks that generate each design and commute them to end of the circuit, leaving only the depolarizing channels acting *between* the designs (see Fig. 1b). This gives the desired representation. □

We briefly remark on how the above calculations will be used in our proof. We will again use the central lemma, Lemma 1, but now with the noisy circuits we have specially constructed: $\mathcal{C}' = \Phi_{\mathcal{C}^*,\vec{p}}$. That is, consider a new instance of the noisy state distinguishability problem (Problem 1), in which the task is to identify the label $k$ given a noisy state from the set

$$\Phi_{\mathcal{C}^*,\vec{p}}(|k\rangle\langle k| \otimes |0\rangle\langle 0|^{\otimes s}), \quad k \in \{0,1\}^n, \quad \text{and } \Phi_{\mathcal{C}^*,\vec{p}}(\mathbb{I}_n/2^n \otimes |0\rangle\langle 0|^{\otimes s}). \tag{51}$$

Let us denote $\rho_k := |k\rangle\langle k| \otimes |0\rangle\langle 0|^{\otimes s}$. We then want to consider the quantity

$$D(\Phi_{\mathcal{C}^*,\vec{p}}(\rho_k) \| \Phi_{\mathcal{C}^*,\vec{p}}(\mathbb{I}_n/2^n \otimes |0\rangle\langle 0|^{\otimes s})). \tag{52}$$

Lemma 6 will simplify the task of bounding this quantity. Recall the result in (50). In particular, let

$$\widetilde{\mathcal{N}} := (\mathcal{D}_{p_\ell}^{\otimes n} \otimes \mathrm{id}_s) \circ \widetilde{\mathcal{C}}_\ell \circ \cdots \circ (\mathcal{D}_{p_2}^{\otimes n} \otimes \mathrm{id}_s) \circ \widetilde{\mathcal{C}}_2 \circ (\mathcal{D}_{p_1}^{\otimes n} \otimes \mathrm{id}_s) \circ \widetilde{\mathcal{C}}_1, \tag{53}$$

so that $\Phi_{\mathcal{C}^*,\vec{p}} = \mathcal{P}_{n+s} \circ \widetilde{\mathcal{N}}$. This re-writing is useful because recursive application of (46) implies that the state $\mathbb{I}_n/2^n \otimes |0\rangle\langle 0|^{\otimes s}$

is a fixed-point of $\widetilde{\mathcal{N}}$, a fact that we use in Eq. (55) below. From the data-processing inequality, we obtain

$$D(\Phi_{\mathcal{C}^*,\vec{p}}(\rho_k)\|\Phi_{\mathcal{C}^*,\vec{p}}(\mathbb{I}_n/2^n \otimes |0\rangle\langle 0|^{\otimes s})) \tag{54}$$

$$= D((\mathcal{P}_{n+s} \circ \widetilde{\mathcal{N}})(\rho_k)\|(\mathcal{P}_{n+s} \circ \widetilde{\mathcal{N}})(\mathbb{I}_n/2^n \otimes |0\rangle\langle 0|^{\otimes s})) \tag{55}$$

$$\leq D(\widetilde{\mathcal{N}}(\rho_k)\|\mathbb{I}_n/2^n \otimes |0\rangle\langle 0|^{\otimes s})$$

$$\leq D\left(\left(\bigcirc_{j=1}^{\ell}\mathcal{D}_{p_j}^{\otimes n} \circ \mathcal{C}_j\right)(|k\rangle\langle k|) \otimes |0\rangle\langle 0|^{\otimes s}\|\widetilde{\mathcal{N}}(\mathbb{I}_n/2^n \otimes |0\rangle\langle 0|^{\otimes s})\right)$$

$$= D\left(\left(\bigcirc_{j=1}^{\ell}\mathcal{D}_{p_j}^{\otimes n} \circ \mathcal{C}_j\right)(|k\rangle\langle k|)\|\mathbb{I}_n/2^n\right),$$

where we note that the channel $\bigcirc_{j=1}^{\ell}\mathcal{D}_{p_j}^{\otimes n} \circ \mathcal{C}_j$ acts only on the $n$ system qubits and not on the auxiliary ones. We have thus reduced our task to upper-bounding the quantity

$$D\left(\left(\bigcirc_{j=1}^{\ell}\mathcal{D}_{p_j}^{\otimes n} \circ \mathcal{C}_j\right)(|k\rangle\langle k|)\|\mathbb{I}_n/2^n\right) \tag{56}$$

where the circuits $\mathcal{C}_j$ are noiseless and independently drawn from a unitary 2-design. Since the uniform ensemble of Clifford unitaries forms a unitary 2-design, we assume from now on that the $\mathcal{C}_j$ are Clifford unitary channels and proceed to bound Eq. (56) in the remainder of this section.

### C. Bound for multiple layers of depolarizing noise

Let us recall that our goal is to bound the *average distinguishability* quantity

$$\frac{1}{N+1}\sum_{k=0}^{N} D(P_k\|P_N). \tag{57}$$

It turns out that we can obtain a $k$-independent bound on $D(P_k\|P_N)$; to do so, we invoke the following consequence of Lemma 1, which says that we can bound the expected relative entropy for each $k$ by calculating the expected purity

$$\mathbb{E}_{C\sim\mathcal{E}}\left[D\left(\left(\bigcirc_{j=1}^{\ell}\mathcal{D}_{p_j}^{\otimes n} \circ \mathcal{C}_j\right)(|k\rangle\langle k|)\left\|\frac{\mathbb{I}}{2^n}\right.\right)\right] \leq n + \mathbb{E}_{C\sim\mathcal{E}}\left[\log \mathrm{Tr}\left[\left(\bigcirc_{j=1}^{\ell}\mathcal{D}_{p_j}^{\otimes n} \circ \mathcal{C}_j\right)(|k\rangle\langle k|))^2\right]\right] \leq \tag{58}$$

$$n + \log \mathbb{E}_{C\sim\mathcal{E}}\,\mathrm{Tr}\left[\left(\bigcirc_{j=1}^{\ell}\mathcal{D}_{p_j}^{\otimes n} \circ \mathcal{C}_j\right)(|k\rangle\langle k|))^2\right],$$

where $k \in \{0,1\}^n$ and the last step follows from Jensen's inequality. Let us now compute that purity. In Section IV we will compute the expected purity of the outcome of $l$ applications of a channel with a $2-$design for arbitrary channels. In particular, in Prop. 4 we show that

$$\mathbb{E}_{C\sim\mathcal{E}}\,\mathrm{Tr}\left[\left(\bigcirc_{j=1}^{\ell}\mathcal{D}_{p_j}^{\otimes n} \circ \mathcal{C}_j\right)(|k\rangle\langle k|))^2\right] = \left(q_n^{\ell}(1-2^{-n}) + 2^{-n}\right), \tag{59}$$

where for a qubit depolarizing channel with depolarizing probability $p$ we have

$$q_n = \frac{(3p^2+1)^n - 1}{4^n - 1}. \tag{60}$$

Inserting this bound into Eq. (58), we conclude that

$$\mathbb{E}_{C\sim\mathcal{E}}\left[D\left(\left(\bigcirc_{j=1}^{\ell}\mathcal{D}_{p_j}^{\otimes n} \circ \mathcal{C}_j\right)(|k\rangle\langle k|)\left\|\frac{\mathbb{I}}{2^n}\right.\right)\right] \tag{61}$$

$$\leq n + \log(q_n^{\ell}(1-2^{-n}) + 2^{-n}) = \log(2^n q_n^{\ell} + 1) \leq 2^n q_n^{\ell}.$$

In particular, if

$$\ell > \frac{1}{2 - \log_2(4 - 3p(2-p))}, \tag{62}$$

then
$$2^n q_n^\ell \leq c^n \tag{63}$$

for a constant $c < 1$.

### D. Putting it all together

We finally feed the results of the previous subsections into the information-theoretic outer loop already used in Section II, and obtain a tighter bound on sample complexity. For conciseness, we state Theorem 3 for depolarizing noise that acts with the same strength after every circuit layer; however, as demonstrated in Section III C, the theorem holds with minor modifications for the case when the noise parameter differs for every layer.

**Theorem 3** (Error mitigation with depolarizing noise requires exponential-in-$n$, $D$ samples). *Let $\mathcal{A}$ be an input state-agnostic weak error mitigation algorithm (see Definition 5) that takes as input $m$ noisy copies of the output state of a circuit $\mathcal{C}$ acted upon by local depolarizing noise of parameter $p$. Suppose that for any circuit $\mathcal{C}$ acting on $n$ qubits of depth $D$, where*

$$D = \Omega\left(\frac{\log^2(n)}{\log\left(\frac{4}{1+3p^2}\right)}\right), \tag{64}$$

*and the set of observables $\mathcal{M} = \{Z_i\}_{i=1}^n$, $\mathcal{A}(\mathcal{C}, \mathcal{M}, p)$ is able to produce estimates $\{\hat{z}_i\}_{i\in[n]}$ such that with probability at least $1-\delta$,*

$$|\hat{z}_i - \mathrm{Tr}\left[Z_i \mathcal{C}(\rho)\right]| \leq \epsilon \quad \text{for all } i \in [n]. \tag{65}$$

*Then as long as $\epsilon < 1/2$, there exists a constant $c \geq 1$, such that the number of noisy copies, $m$, needed by $\mathcal{A}(\mathcal{C}, \mathcal{M}, p)$, is lower-bounded as*

$$m = c^{-\tilde{\mathcal{O}}(nD)}, \tag{66}$$

*where the notation $\tilde{\mathcal{O}}(.)$ hides polylogarithmic factors in $n$.*

*Proof.* Let us first note that we must re-scale $n$ by a constant factor because it now refers to the size of both system and auxiliary qubits (the constant factor is because only $s = \tilde{\mathcal{O}}(n)$ auxiliary qubits are required, according to Ref. [19]), while our bounds from the previous subsections were in terms of only the number of system qubits. From Eq. (61) and the linearity of expectation we conclude that

$$\mathbb{E}_\mathcal{C}\left[\frac{1}{N+1}\sum_{k=0}^N D\left(\mathcal{C}(|k\rangle\langle k| \otimes |0\rangle\langle 0|) \,\Big\|\, \sum_E p(E)\frac{\mathbb{I}}{2^n} \otimes \sigma_E\right)\right] \leq c^{\Omega(n\ell)} . \tag{67}$$

Here, the expectation is over the distribution induced by constructing a circuit $\mathcal{C}$ by sequentially performing the randomized circuit construction of Ref. [19],

$$\ell \geq \frac{1}{\log\left(\frac{4}{1+3p^2}\right)} \tag{68}$$

many times. Here, we note that—as stated in the previous subsection—the auxiliary states $\sigma_E$ are a function of $\mathcal{C}$. By the probabilistic method, there thus exists at least one circuit $\tilde{\mathcal{C}}$ of this form, with corresponding states $\{\sigma_E\}_E$, such that the average relative entropy is upper-bounded by the same quantity: the circuit's output states

$$\tilde{\rho}_k := \tilde{\mathcal{C}}(|k\rangle\langle k| \otimes |0\rangle\langle 0|) \quad \text{for } k = 0, \ldots, N-1, \tag{69}$$

$$\tilde{\rho}_N := \tilde{\mathcal{C}}\left(\sum_E p(E)\frac{\mathbb{I}}{2^n} \otimes \sigma_E\right) \tag{70}$$

for $N = 2^n$, satisfy

$$\frac{1}{N+1} \sum_{k=0}^{N} D(\rho_k || \rho_N) \leq c^{\mathcal{O}(n\ell)} = c^{\mathcal{O}(nD/\log^2(n))}. \tag{71}$$

Wrapping up, we plug this into the central lemma (Lemma 1): we lower bound the sample complexity of mitigating the depolarizing noise in the circuit $\mathcal{C}$ by the sample complexity of noisy state discrimination of the states $\{\tilde{\rho}_0, \ldots \tilde{\rho}_N\}$. This proceeds in much the same way as we proved Theorem 1, except that we do not even require all the estimates output by weak error mitigation for the state discrimination problem: we may throw away the information in the estimates $\hat{z}_j$ for $j > n$ (i.e., the Paulis on the auxiliary qubits). □

### E. Strengthening our results

In this last subsection, we strengthen our results so that they hold for circuits that are restricted to use local gates only. We also improve the depth dependence and introduce a dependence on the light-cone size.

#### 1. Geometrically local circuits

Our circuit construction works by iterating the construction of Ref. [19] $D$ times. The reader may observe, however, that the above construction requires all-to-all connectivity. Bearing in mind that error mitigation is usually applied to near-term quantum circuits for which only some geometrically local architecture is available, we now discuss how our results transfer to error mitigation only on such circuits. We observe that, for architectures in which it is possible to apply gates to nearest neighbors on a $d$ dimensional lattice, it is possible to implement depth $D$ circuit with all-to-all connectivity in depth $\mathcal{O}(Dn^{\frac{1}{d}})$ [20, 21] by performing additional layers of SWAP gates. Note that SWAP gates are Clifford, thus combining these SWAP networks with the exact 2-designs we have used before again give rise to Cliffords. Also note that depth $\mathcal{O}(n^{1/d})$ is necessary to ensure that the lightcone of all qubits corresponds to the whole system.

Thus, unsurprisingly, with limited connectivity, the exponential in $n$ cost of error mitigation implied by our results only kicks in at depths at which the light-cone of all qubits is the entire system. And, for $d$ dimensional architectures, this is $\mathcal{O}(n^{1/d})$. Translating the circuit used for Thm. 5 to a $d$-dimensional architecture then gives that our results kick-in at $\mathcal{O}(n^{\frac{1}{d}} \log^2(n))$ depth.

**Corollary 1** (Error-mitigating circuits composed of 2-local gates in $d$ dimensions with depolarizing noise). *There exists a constant $c \geq 1$, such that for $n$-qubit, depth-$D$ circuits composed of 2-local gates in $d$ dimensions, at least*

$$m \geq c^{-\tilde{\mathcal{O}}(nD)} \tag{72}$$

*copies of the noisy circuit output state are needed to mitigate depolarizing noise of parameter $p$, once*

$$D \geq \tilde{\mathcal{O}}\left(\frac{n^{1/d}}{\log(\frac{4}{1+3p^2})}\right). \tag{73}$$

#### 2. Smaller depths and light-cone size dependence

Furthermore, we can easily use the construction above to also obtain a large relative entropy contraction at smaller depths. Indeed, suppose we partition the qubits into $s$ disjoint subsets $S_1, \ldots, S_s$ such that each subset has maximal size $L = \lceil n/s \rceil$. After that, we independently apply the construction above *on each* subset. Then at depth

$$D \geq O\left(\frac{\log^2(L)}{\log\left(\frac{4}{1+3p^2}\right)}\right) \tag{74}$$

on each block the relative entropy will satisfy

$$\frac{1}{N+1} \sum_{k=0}^{N} D(\rho_k || \rho_N) \leq O(c^{-LD}). \tag{75}$$

As the state will be product across different subsets, the total relative entropy will be at most

$$\mathcal{O}\left(sc^{-LD}\right). \tag{76}$$

Furthermore, another property this construction enjoys is that it only involves quantum circuits with light-cones of size at most $\lceil n/s \rceil$. Indeed, as each of the unitaries only acts on subset of qubits of this size, the light-cone of each qubit is of that size. Such light-cone size considerations allow us to reformulate our theorem in terms of a customizable parameter $s$. Choosing $s$ wisely leads to our final theorem with the advertised depth dependence, as follows:

**Theorem 4** (Error mitigation with depolarizing noise with small depth and light-cones). *Let $\mathcal{A}$ be an input state-agnostic weak error mitigation algorithm (see Definition 5) that takes as input $m$ noisy copies of the output state of a circuit $\mathcal{C}$ acted upon by local depolarizing noise of parameter $p$. Suppose that for an arbitrary circuit $\mathcal{C}$ acting on $n$ qubits of depth $D$, where*

$$D \geq O\left(\frac{\log^2(n/s)}{\log\left(\frac{4}{1+3p^2}\right)}\right) \tag{77}$$

*for some $0 < s < n$, and the set of observables $\mathcal{M} = \{Z_i\}_{i=1}^n$, $\mathcal{A}(\mathcal{C}, \mathcal{M}, p)$ is able to produce estimates $\{\hat{z}_i\}_{i \in [n]}$ such that with probability at least $1 - \delta$,*

$$|\hat{z}_i - \text{Tr}[Z_i \mathcal{C}(\rho)]| \leq \epsilon \qquad \text{for all } i \in [n]. \tag{78}$$

*Then as long as $\epsilon < 1/2$, there is a constant $c \geq 1$, such that the number of noisy copies, $m$, needed by $\mathcal{A}$, is lower-bounded as*

$$m = \frac{1}{s} c^{-\tilde{\mathcal{O}}(Dn/s)} \quad , \tag{79}$$

*where the notation $\tilde{\mathcal{O}}()$ hides polylogarithmic factors in $n$. In particular, by picking $s \sim n/\log^2(n)$, we see that there are circuits of depth $\text{poly}(\log \log(n))$ that require a super-polynomial number of samples to error-mitigate.*

The theorem above shows that as long as light-cone sizes are slightly above $\log(n)$ (namely $\log^2(n)$) and depths are slightly above constant ($\log^2(\log(n))$, error mitigation is no longer efficient in general. However, note that it is possible to compute the expectation value of observables with light-cones of size $\mathcal{O}(\log(n))$ efficiently classically. This proves that in general error mitigation can only provide a better exponent in the scaling of an algorithm to estimate expectation values when compared with classical algorithms.

Furthermore, also this construction can be applied to the case of geometrically local circuits on a $d$-dimensional lattice. Note that the above construction consistis of a parallel circuits on $n/s$ qubits. Thus, we can implement circuits with all-to-all connectivity *only* on the $n/s$ qubits with a depth overhead of $\mathcal{O}((n/s)^{\frac{1}{d}})$. In particular, for the choise of $s = \frac{n}{\log(n)^2}$, this gives gives overheads in depth of $\mathcal{O}((\log(n))^{\frac{2}{d}})$ for when superpolynomially-many samples are required for error mitigation in $d$-dimensional circuits.

## IV. BEYOND UNITAL NOISE

So far all the results we have discussed only apply to local random circuits with unital noise. It is thus natural to ask how the picture changes when the noise is non-unital. This is because many physically relevant noise models, such as amplitude damping, are non-unital. However, going beyond the non-unital model brings with it several technical complications. First, as we show, in general the noisy circuit will not converge anymore to a fixed point. Second, the entropy of a state does not necessarily increase under a non-unital map. As our previous results were based on entropic considerations, our arguments do not apply anymore. Furthermore, if the outputs of the noisy circuit are not full-rank, as could be the case for noise models like amplitude damping, the relative entropy between the outputs given different inputs might even be infinitely large. This motivates us to look at a different distance measure as well. Indeed, we resort in this section to the standard technique of controlling the trace distance through the Hilbert-Schmidt distance.

In light of these complications, we now discuss a toy model for the convergence of noisy circuits where it is still possible to obtain explicit bounds on their convergence. More precisely, we consider the model in which we have a circuit composed of (global) 2-designs followed by a noisy channel, and then we iterate this construction.

We show below that if the noisy channel is a qubit product channel, in expectation the trace distance between outputs are exponentially small in both the depth and the number of qubits. Although global 2-designs require at least polylogarithmic depth to implement and we only consider the case in which the noise acts after we implement it, we believe that our results illustrate

that in principle our conclusions in the previous sections should carry over to the non-unital setting for deep enough unitary circuits. We start by computing moments of noisy random circuits that will be crucial for our conclusions later on.

### A. Purity and overlap change after one noisy gate

Our goal is to understand the expected Hilbert-Schmidt distance between two initially-orthogonal pure states that are impacted by a noisy random circuit, in which the noise is potentially non-unital. Specifically, we let the noisy random circuit to be of the form

$$T_D := \bigcirc_{\ell=1}^D \mathcal{N} \circ \mathcal{U}_\ell, \tag{80}$$

where $\mathcal{N}$ is an arbitrary quantum channel and the unitary quantum channels $\mathcal{U}_\ell(\cdot) = U_\ell(\cdot)U_\ell^\dagger$ are such that the unitaries $U_\ell$ are Haar-random. (We will see below that it suffices that the $U_\ell$ come from a unitary 2-design.) The idea is that $\mathcal{N}$ represents the noise after we apply the gate. Then, we are interested in the output states

$$\rho_{i,D} := T_D(\rho_i), \quad i \in \{1,2\}, \tag{81}$$

where $\rho_1$ and $\rho_2$ are arbitrary states of a $d$-dimensional system. In particular, we are interested in how the quantity

$$\mathbb{E}_{U_1,\ldots,U_D}[\|\rho_{1,D} - \rho_{2,D}\|_2^2] = \mathbb{E}_U[\text{Tr}(\rho_{1,D}^2) + \text{Tr}(\rho_{2,D}^2) - \text{Tr}(\rho_{1,D}\rho_{2,D}) - \text{Tr}(\rho_{2,D}\rho_{1,D})] \tag{82}$$

decays with respect to the depth $D$. Note that by controlling overlaps and purities it is then straightforward to control the Hilbert-Schmidt distance.

As a first step, let us study the above question for $D = 1$. We consider the states $\rho_{i,1} = \mathcal{N} \circ \mathcal{U}(\rho_i)$, $i \in \{1,2\}$, and we would like to compute the expected value of the overlap $\text{Tr}[\rho_{1,1}\rho_{2,1}]$. We can express this overlap in a more convenient way, namely,

$$\text{Tr}[\rho_{1,1}\rho_{2,1}] = \text{Tr}[(\mathcal{N} \circ \mathcal{U})(\rho_1)(\mathcal{N} \circ \mathcal{U})(\rho_2)] \tag{83}$$
$$= \text{Tr}[F((\mathcal{N} \circ \mathcal{U})(\rho_1) \otimes (\mathcal{N} \circ \mathcal{U})(\rho_2))],$$

where $F$ is the swap operator. Now, suppose that $\mathcal{N}$ has the Kraus decomposition

$$\mathcal{N}(\cdot) = \sum_{\ell=1}^r K_\ell(\cdot)K_\ell^\dagger. \tag{84}$$

Using this, we get that the expected overlap of the output states is

$$\mathbb{E}_U[\text{Tr}[\rho_{1,1}\rho_{2,1}]] = \sum_{\ell,\ell'=1}^r \mathbb{E}_U\left[\text{Tr}\left[F\left((K_\ell \otimes K_{\ell'})(U \otimes U)(\rho_1 \otimes \rho_2)(U^\dagger \otimes U^\dagger)(K_\ell^\dagger \otimes K_{\ell'}^\dagger)\right)\right]\right] \tag{85}$$

$$= \alpha \sum_{\ell,\ell'=1}^r \text{Tr}\left[F((K_\ell \otimes K_{\ell'})(K_\ell^\dagger \otimes K_{\ell'}^\dagger))\right] + \beta \sum_{\ell,\ell'=1}^r \text{Tr}\left[F(K_\ell \otimes K_{\ell'})F(K_\ell^\dagger \otimes K_{\ell'}^\dagger)\right], \tag{86}$$

where for the last line we have used the general fact (see, e.g., Ref. [22, Example 7.25]) that

$$\mathbb{E}_U[(U \otimes U)X(U^\dagger \otimes U^\dagger)] = \underbrace{\left(\frac{\text{Tr}[X]}{d^2-1} - \frac{\text{Tr}[FX]}{d(d^2-1)}\right)}_{\alpha}\mathbb{I} \otimes \mathbb{I} + \underbrace{\left(\frac{\text{Tr}[FX]}{d^2-1} - \frac{\text{Tr}[X]}{d(d^2-1)}\right)}_{\beta}F. \tag{87}$$

We now observe that, because only two copies of the unitary $U$ appear in the expectation value, we would get the same quantity if we took the expectation value with respect to a unitary 2-design instead of the Haar-random unitaries.

Now, the involved quantities in (86) take the values

$$\text{Tr}\left[F((K_\ell \otimes K_{\ell'})(K_\ell^\dagger \otimes K_{\ell'}^\dagger))\right] = \text{Tr}[K_\ell K_\ell^\dagger K_{\ell'} K_{\ell'}^\dagger], \tag{88}$$

$$\text{Tr}\left[F(K_\ell \otimes K_{\ell'})F(K_\ell^\dagger \otimes K_{\ell'}^\dagger)\right] = \text{Tr}[K_{\ell'} K_\ell^\dagger \otimes K_\ell K_{\ell'}^\dagger] = |\text{Tr}[K_\ell K_{\ell'}^\dagger]|^2, \tag{89}$$

for all $\ell, \ell' \in \{1, 2, \ldots, r\}$ and

$$\alpha = \frac{1}{d^2 - 1} - \frac{\text{Tr}[\rho_1 \rho_2]}{d(d^2 - 1)}, \tag{90}$$

$$\beta = \frac{\text{Tr}[\rho_1 \rho_2]}{d^2 - 1} - \frac{1}{d(d^2 - 1)}. \tag{91}$$

We can rephrase these quantities in terms of global properties of the channel as

$$\sum_{\ell=1}^{r} K_\ell K_\ell^\dagger = \mathcal{N}(\mathbb{I}), \tag{92}$$

$$\sum_{\ell, \ell'=1}^{r} |\text{Tr}[K_\ell K_{\ell'}^\dagger]|^2 = \text{Tr}[(\Gamma^\mathcal{N})^2], \tag{93}$$

where

$$\Gamma^\mathcal{N} := (\text{id} \otimes \mathcal{N})(|\Gamma\rangle\langle\Gamma|), \quad |\Gamma\rangle := \sum_{j=0}^{d-1} |j, j\rangle, \tag{94}$$

is the Choi representation of $\mathcal{N}$. Therefore, the expected overlap of the output states is

$$\mathbb{E}_U[\text{Tr}[\rho_{1,1} \rho_{2,1}]] = \left(\frac{1}{d^2 - 1} - \frac{\text{Tr}[\rho_1 \rho_2]}{d(d^2 - 1)}\right) \text{Tr}[\mathcal{N}(\mathbb{I})^2] + \left(\frac{\text{Tr}[\rho_1 \rho_2]}{d^2 - 1} - \frac{1}{d(d^2 - 1)}\right) \text{Tr}[(\Gamma^\mathcal{N})^2]. \tag{95}$$

Note that both $\text{Tr}[(\Gamma^\mathcal{N})^2]$ and $\text{Tr}[\mathcal{N}(\mathbb{I})^2]$ scale with the underlying dimension of the space. In order to more easily grasp the scaling of the formula above, it is fruitful to look at the rescaled quantities

$$\eta = \text{Tr}[\mathcal{N}(\mathbb{I}/d)^2], \quad \nu = \text{Tr}[(\Gamma^\mathcal{N})^2/d^2]. \tag{96}$$

Note that $\eta$ is just the purity of the output on the maximally mixed state and, thus, takes values in $[d^{-1}, 1]$, whereas $\nu$ is just the purity of the output on the maximally entangled state if we act with the channel on one half of the system (i.e., the Choi state) and takes values in $[d^{-2}, 1]$. In terms of these newly defined quantities, we obtain

$$\mathbb{E}_U[\text{Tr}[\rho_{1,1} \rho_{2,1}]] = \left(\frac{d^2}{d^2 - 1} - \frac{d \, \text{Tr}[\rho_1 \rho_2]}{(d^2 - 1)}\right) \eta + \left(\frac{d^2 \, \text{Tr}[\rho_1 \rho_2]}{d^2 - 1} - \frac{d}{(d^2 - 1)}\right) \nu. \tag{97}$$

**Remark 5.** *It is well-known that twirling any quantum channel with a unitary 2-design will lead to a global depolarizing noise in expectation, i.e., $\mathbb{E}_{U \sim \text{Haar}}[\mathcal{U}^\dagger \circ \mathcal{N} \circ \mathcal{U}] = \mathcal{D}_p^{(d)}$. Note, however, that the problem at hand is slightly different, as we are interested in an expectation value that contains two copies of the channel. In that case, the resulting action cannot be effectively of a global depolarizing channel. To illustrate this point, assume that the noise $\mathcal{N}$ is actually a unitary. In that case, the purity or overlap between two different initial states should not change after the action of $\mathcal{U}^\dagger \circ \mathcal{N} \circ \mathcal{U}$, which would not be the case if the resulting channel were that of depolarizing noise. A simple instantiation of this is when the two initial states are actually identical: $\rho_1 = \rho_2 = \rho$; then plug in $\mathcal{N}(\mathbb{I}) = \mathbb{I}$ and $\text{Tr}[(\Gamma^\mathcal{N})^2] = d^2$ to our earlier calculation, obtaining*

$$\mathbb{E}_U[\text{Tr}[\rho_{1,1} \rho_{2,1}]] = \alpha d + \beta d^2 \tag{98}$$

$$= d\left(\frac{1}{d^2 - 1} - \frac{\text{Tr}[\rho^2]}{d(d^2 - 1)}\right) + d^2\left(\frac{\text{Tr}[\rho^2]}{d^2 - 1} - \frac{1}{d(d^2 - 1)}\right)$$

$$= \text{Tr}[\rho^2],$$

*as expected.*

## B. Convergence of global random circuits

From Eq. (97), we obtain a recursive formula for the overlap and/or purity of the output states $\rho_{i,D}$, $i \in \{1,2\}$, as defined in (81). We have

$$\mathbb{E}_{U_1,\ldots,U_D}[\text{Tr}[\rho_{1,D}\rho_{2,D}]] = \left(\frac{d^2}{d^2-1} - \frac{d\,\text{Tr}[\rho_{1,D-1}\rho_{2,D-1}]}{(d^2-1)}\right)\eta + \left(\frac{d^2\,\text{Tr}[\rho_{1,D-1}\rho_{2,D-1}]}{d^2-1} - \frac{d}{(d^2-1)}\right)\nu$$

$$= \frac{d(d\nu-\eta)}{d^2-1}\mathbb{E}_{U_1,\ldots,U_{D-1}}[\text{Tr}[\rho_{1,D-1}\rho_{2,D-1}]] + \frac{d(d\eta-\nu)}{d^2-1}. \tag{99}$$

The structure of the recursion becomes more apparent when discussed in abstract terms. Letting $A := \frac{d(d\nu-\eta)}{d^2-1}$, $B := \frac{d(d\eta-\nu)}{d^2-1}$, and $f(D) := \mathbb{E}_{U_1,\ldots,U_D}[\text{Tr}[\rho_{1,D}\rho_{2,D}]]$, we have that

$$\begin{aligned} f(D) &= Af(D-1) + B \\ &= A(Af(D-2) + B) + B \\ &= A^2 f(D-2) + AB + B \\ &= A^2(Af(D-3) + B) + (A+1)B \\ &= A^3 f(D-3) + (A^2 + A + 1)B \\ &= A^D f(0) + (A^{D-1} + \cdots + A + 1)B. \end{aligned} \tag{100}$$

The right term is nothing but a geometric series, which, noting that $0 \le A < 1$, yields

$$f(D) = A^D f(0) + \frac{1-A^D}{1-A}B. \tag{101}$$

Applied to our setting, we obtain

$$\mathbb{E}_{U_1,\ldots,U_D}[\text{Tr}[\rho_{1,D}\rho_{2,D}]] = \left(\frac{d(d\nu-\eta)}{d^2-1}\right)^D \text{Tr}[\rho_1\rho_2] + \left[1 - \left(\frac{d(d\nu-\eta)}{d^2-1}\right)^D\right]\frac{d(d\eta-\nu)}{d^2-1-d(d\nu-\eta)}. \tag{102}$$

The exponential behavior in $D$ of this quantity and the applicability of the geometric series is determined by the term $d(d\nu-\eta)/(d^2-1)$. The fact that this term is always positive and strictly smaller than one if the channel is not unitary, indicating an exponential decay in $D$, is a consequence of the following proposition.

**Proposition 2.** *Let $\mathcal{N} : \mathsf{M}_d \to \mathsf{M}_d$ be a quantum channel. Define $\eta$ and $\nu$ as in Eq.* (96). *Then*

$$d\nu - \eta > 0. \tag{103}$$

*Furthermore,*

$$\frac{d^2\nu - d\eta}{d^2-1} \le 1 \tag{104}$$

*with equality if and only if $\mathcal{N}$ is unitary.*

*Proof.* First, note that for any linear superoperator $\mathcal{L}$ we have $\text{Tr}[|\Gamma\rangle\langle\Gamma|(\text{id} \otimes \mathcal{L})(|\Gamma\rangle\langle\Gamma|)] = \text{Tr}[\mathcal{L}]$, where by $\text{Tr}[\mathcal{L}]$ we mean the trace of $\mathcal{L}$ as a linear map. Thus, we see that $\nu = d^{-2}\,\text{Tr}[\mathcal{N}^\dagger\mathcal{N}]$ is just the average of the square of the singular values of quantum channel. Now, note that $d^{-\frac{1}{2}}\mathbb{I}$ is a normalized vector with respect to the Hilbert-Schmidt norm. Furthermore, let $X_i$, $1 \le i \le d^2 - 1$ be an orthonormal, self-adjoint basis for the space of traceless matrices. Then, using the fact that

$$|\Gamma\rangle\langle\Gamma| = \frac{\mathbb{I}}{\sqrt{d}} \otimes \frac{\mathbb{I}}{\sqrt{d}} + \sum_{i=1}^{d^2-1} \overline{X_i} \otimes X_i, \tag{105}$$

we obtain

$$\nu = d^{-2}\operatorname{Tr}[\mathcal{N}^\dagger \mathcal{N}] = d^{-2}\operatorname{Tr}[\mathcal{N}(d^{-\frac{1}{2}}I)^2] + d^{-2}\sum_{i=1}^{d-1}\operatorname{Tr}[\mathcal{N}(X_i)^2] = d^{-1}\eta + d^{-2}\sum_{i=1}^{d^2-1}\operatorname{Tr}[\mathcal{N}(X_i)^2]. \tag{106}$$

As the second term in the sum above is clearly positive, we obtain the claim in Eq. (103). To obtain the claim in Eq. (104), note that the expression is clearly monotone increasing in $\nu$ and decreasing in $\eta$. The maximum value for $\nu$ is 1 because it is a purity, and this value is attained on a unitary channel (pure Choi states are unitaries). Thus, the quantity in Eq. (104) is strictly smaller than 1 for non-unitary channels, as then $\nu < 1$. Furthermore, for unitary channels we have also have $\eta = d^{-1}$, and, thus, the quantity in (104) is 1 for unitary channels. □

In particular, it follows that unless the channel is unitary, the expected trace norm between outputs will be exponentially small. To conclude that, we will first study the convergence in the Hilbert-Schmidt norm.

**Proposition 3.** *Let $T_D = \bigcirc_{\ell=1}^D \mathcal{N} \circ \mathcal{U}_\ell$ be a random quantum channel where $\mathcal{N} : \mathsf{M}_d \to \mathsf{M}_d$ is a fixed quantum channel and $\mathcal{U}_\ell$ are drawn independently from a unitary 2-design on a $d$-dimensional space. Furthermore, let $\eta, \nu$ be defined as in Eq. (96) and*

$$q = \frac{d(d\nu - \eta)}{d^2 - 1}. \tag{107}$$

*Then, for any orthogonal pure states $\rho_1, \rho_2$ we have*

$$\mathbb{E}_{U_1,\ldots,U_D}\left[\|T_D(\rho_1 - \rho_2)\|_2^2\right] = 2q^D \tag{108}$$

*and*

$$\mathbb{E}_{U_1,\ldots,U_D}[\|T_D(\rho_1 - \rho_2)\|_1] \leq \sqrt{2dq^D}. \tag{109}$$

*Proof.* First, observe that

$$\mathbb{E}_{U_1,\ldots,U_D}\left[\|T_D(\rho_1 - \rho_2)\|_2^2\right] = \mathbb{E}_{U_1,\ldots,U_D}[\operatorname{Tr}[\rho_{1,D}^2]] + \mathbb{E}_{U_1,\ldots,U_D}[\operatorname{Tr}[\rho_{2,D}^2]] - 2\mathbb{E}_{U_1,\ldots,U_D}[\operatorname{Tr}[\rho_{1,D}\rho_{2,D}]], \tag{110}$$

where $\rho_{i,D} = T_D(\rho_i), i \in \{1,2\}$. Inserting the expression in Eq. (102) for each term, we get

$$\mathbb{E}_{U_1,\ldots,U_D}\left[\|T_\ell(\rho_1 - \rho_2)\|_2^2\right] = q^D\left(\operatorname{Tr}[\rho_1^2] + \operatorname{Tr}[\rho_2^2] - 2\operatorname{Tr}[\rho_1\rho_2]\right), \tag{111}$$

as the rest of the terms cancel. Using that $\operatorname{Tr}[\rho_1\rho_2] = 0$ and $\operatorname{Tr}[\rho_1^2] = \operatorname{Tr}[\rho_2^2] = 1$ we arrive at Eq. (108). Next, from Jensen's inequality, it follows that

$$\mathbb{E}\left[\|T_D(\rho_1 - \rho_2)\|_1\right]^2 \leq \mathbb{E}\left[\|T_D(\rho_1 - \rho_2)\|_1^2\right]. \tag{112}$$

Furthermore, we have the standard inequality $\|X\|_1 \leq \sqrt{d}\|X\|_2$ for all Hermitian $X$. Combining this observation with the claim in Eq. (112) yields Eq. (109). □

Let us investigate the case of $d = 2^n$ ($n$ qubits) and $\mathcal{N} = \bigotimes_{i=1}^n \mathcal{N}_i$ (i.e., $\mathcal{N}$ is a product channel) in a bit more detail. Furthermore, assume for simplicity that all $\mathcal{N}_i$ are the same. As the purity is multiplicative, we obtain $\eta_\mathcal{N} = \eta_{\mathcal{N}_i}^n, \nu_\mathcal{N} = \nu_{\mathcal{N}_i}^n$ and

$$q_n = \frac{2^{2n}\nu_{\mathcal{N}_i}^n - 2^n\eta_{\mathcal{N}_i}^n}{2^{2n} - 1}. \tag{113}$$

Thus, we see that the 2-norm decays exponentially in the number of qubits, even after one layer. In particular, we obtain the following corollary.

**Corollary 2.** *In the same setting of Proposition 3, assume further that $\mathcal{N}$ is an $n$-qubit product channel and define $q_n$ as in Eq. (113). Then we have*

$$\mathbb{E}_{U_1,\ldots,U_D}[\|T_D(\rho_1 - \rho_2)\|_1] \leq 2^{\frac{n+1}{2}} q_n^{\frac{D}{2}}. \tag{114}$$

*Proof.* The statement follows from Proposition 3 and the discussion above. □

We can combine the statement above with sub-additivity of the trace distance, the triangle inequality and a hybrid argument to conclude that for $m$ copies of the outputs of the noisy circuit we have

$$\mathbb{E}_{U_1,\ldots,U_D}\left[\|T_D^{\otimes m}(\rho_1^{\otimes m} - \rho_2^{\otimes m})\|_1\right] \leq m 2^{\frac{n+1}{2}} \left(\frac{2^{2n}\nu_{\mathcal{N}_i}^n - 2^n \eta_{\mathcal{N}_i}^n}{2^{2n} - 1}\right)^{\frac{D}{2}}. \tag{115}$$

As explained before, successful error mitigation can be used to distinguish orthogonal input states. However, by the Holevo-Helstrom theorem [23, 24], the trace distance bounds the probability of successfully distinguishing between two states. That is, we find the following theorem.

**Theorem 6** (Holevo-Helstrom [23, 24]). *The success probability to discriminate two mixed states represented by $\rho_1$ and $\rho_2$ is at most*

$$\frac{1}{2} + \frac{1}{2}\left(\frac{1}{2}\|\rho_1 - \rho_2\|_1\right).$$

Thus, we conclude the following.

**Theorem 7** (Error mitigation with non-unital noise). *Any input agnostic error mitigation procedure that works successfully for outputs of $T_D$ would require a number of copies $m$ that grows as*

$$m \geq 2^{-\frac{n+1}{2}} \left(\frac{2^{2n}\nu_{\mathcal{N}_i}^n - 2^n \eta_{\mathcal{N}_i}^n}{2^{2n} - 1}\right)^{-\frac{D}{2}}. \tag{116}$$

As discussed before, $\nu_{\mathcal{N}_i} < 1$ unless we have a unitary channel. Thus $m$ will be $\Omega(c^{-nD})$ for some some $c < 1$ at a number of iterations $D = \mathcal{O}(1)$, as

$$\left(\frac{2^{2n}\nu_{\mathcal{N}_i}^n - 2^n \eta_{\mathcal{N}_i}^n}{2^{2n} - 1}\right)^{\frac{D}{2}} \leq 2^{-\frac{n+1}{2}} \tag{117}$$

for $D = \mathcal{O}(1)$.

The input aware case is a bit more subtle. What we conclude from the results above is that all input states will be mapped by $T_D$ to the same state up to an exponential small correction. Thus, the noisy circuit is essentially a replacer channel and any error mitigation algorithm with sub-exponential samples would have the same performance given copies of this one state. However, unlike the case with the maximally mixed state, it is unclear that it is possible to classically simulate that state or that it is not a valuable computational resource. However, it is of course also unclear how a practical error mitigation might profit from access to samples from an (unknown) fixed state.

### C. Bounds on the probability of successful virtual distillation

It is straightforward to use the result introduced above to bound the probability of virtual distillation protocols [2, 3] succeeding under non-unital noise. As discussed before, the probability of such protocols succeeding is bounded by $\mathrm{Tr}[\rho^2]$. If we introduce the quantity

$$r_n = \frac{2^{2n}\eta_{\mathcal{N}_i}^n - 2^n \nu_{\mathcal{N}_i}^n}{2^{2n} - 1}, \tag{118}$$

we can show the following:

**Proposition 4.** *In the same setting of Proposition 3, assume further that $\mathcal{N}$ is an $n$-qubit product channel and define $q_n$ as in Eq. (113) and $r_n$ as in Eq. (118). Then we have for all pure states $\rho$*

$$\mathbb{E}_{U_1,\ldots,U_D}\left[\mathrm{Tr}\left[T_D(\rho)^2\right]\right] = \left(q_n^D + \frac{(1-q_n^D)r_n}{1-q_n}\right). \tag{119}$$

Furthermore, both $r_n$ and $q_n$ are exponentially small in $n$ unless the quantum channel $\mathcal{N}_i$ is unitary or a replacer channel for a pure state, i.e., $\mathcal{N}_i(\sigma) = |\psi\rangle\langle\psi|$ for all $\sigma$ and some $|\psi\rangle$.


}

\end{document}

%% file: commands.tex
\newcommand*{\Tr}{\operatorname{Tr}}

\providecommand{\calA}{\ensuremath{\mathcal{A}}}

\providecommand{\calC}{\ensuremath{\mathcal{C}}}

\providecommand{\calM}{\ensuremath{\mathcal{M}}}
\providecommand{\calN}{\ensuremath{\mathcal{N}}}
\providecommand{\calO}{\ensuremath{\mathcal{O}}}
\providecommand{\calP}{\ensuremath{\mathcal{P}}}

\providecommand{\bbP}{\ensuremath{\mathbb{P}}}

\def\01{\{0,1\}}
\newcommand{\eps}{\epsilon}

\newcommand{\C}{\ensuremath{\mathcal{C}}}

\newcommand{\mN}{\ensuremath{\mathcal{N}}}

\newcommand{\ketbra}[2]{|{#1}\rangle\!\langle{#2}|}

\usepackage{xcolor}
\definecolor{ypurple}{RGB}{139,0,139}

\AtBeginDocument{
\heavyrulewidth=.08em
\lightrulewidth=.05em
\cmidrulewidth=.03em
\belowrulesep=.65ex
\belowbottomsep=0pt
\aboverulesep=.4ex
\abovetopsep=0pt
\cmidrulesep=\doublerulesep
\cmidrulekern=.5em
\defaultaddspace=.5em
}

\newtheorem{theorem}{Theorem}
\newtheorem{problem}{Problem}
\newtheorem{lemma}{Lemma}

\newtheorem{defn}{Definition}

\newcommand{\beq}{\begin{equation}}
\newcommand{\beql}[1]{\begin{equation}\label{#1}}
\newcommand{\eeq}{\end{equation}}
\newcommand{\eeqp}{\;\;\;.\end{equation}}
\newcommand{\eeqc}{\;\;\;,\end{equation}}

\def\01{\{0,1\}}

\newcommand{\Q}{\mathcal{Q}}

\newcommand{\mA}{\mathcal{A}}

%% file: FinalNatPhysArXiv.bbl
\medskip


\begin{thebibliography}{49}%
\makeatletter
\providecommand \@ifxundefined [1]{%
 \@ifx{#1\undefined}
}%
\providecommand \@ifnum [1]{%
 \ifnum #1\expandafter \@firstoftwo
 \else \expandafter \@secondoftwo
 \fi
}%
\providecommand \@ifx [1]{%
 \ifx #1\expandafter \@firstoftwo
 \else \expandafter \@secondoftwo
 \fi
}%
\providecommand \natexlab [1]{#1}%
\providecommand \enquote  [1]{``#1''}%
\providecommand \bibnamefont  [1]{#1}%
\providecommand \bibfnamefont [1]{#1}%
\providecommand \citenamefont [1]{#1}%
\providecommand \href@noop [0]{\@secondoftwo}%
\providecommand \href [0]{\begingroup \@sanitize@url \@href}%
\providecommand \@href[1]{\@@startlink{#1}\@@href}%
\providecommand \@@href[1]{\endgroup#1\@@endlink}%
\providecommand \@sanitize@url [0]{\catcode `\\12\catcode `\$12\catcode
  `\&12\catcode `\#12\catcode `\^12\catcode `\_12\catcode `\%12\relax}%
\providecommand \@@startlink[1]{}%
\providecommand \@@endlink[0]{}%
\providecommand \url  [0]{\begingroup\@sanitize@url \@url }%
\providecommand \@url [1]{\endgroup\@href {#1}{\urlprefix }}%
\providecommand \urlprefix  [0]{URL }%
\providecommand \Eprint [0]{\href }%
\providecommand \doibase [0]{https://doi.org/}%
\providecommand \selectlanguage [0]{\@gobble}%
\providecommand \bibinfo  [0]{\@secondoftwo}%
\providecommand \bibfield  [0]{\@secondoftwo}%
\providecommand \translation [1]{[#1]}%
\providecommand \BibitemOpen [0]{}%
\providecommand \bibitemStop [0]{}%
\providecommand \bibitemNoStop [0]{.\EOS\space}%
\providecommand \EOS [0]{\spacefactor3000\relax}%
\providecommand \BibitemShut  [1]{\csname bibitem#1\endcsname}%
\let\auto@bib@innerbib\@empty
%</preamble>
\bibitem [{\citenamefont {Feynman}(1986)}]{Feynman-1986}%
  \BibitemOpen
  \bibfield  {author} {\bibinfo {author} {\bibfnamefont {R.~P.}\ \bibnamefont
  {Feynman}},\ }\bibfield  {title} {\bibinfo {title} {Quantum mechanical
  computers},\ }\href {https://doi.org/10.1007/BF01886518} {\bibfield
  {journal} {\bibinfo  {journal} {Found. Phys.}\ }\textbf {\bibinfo {volume}
  {16}},\ \bibinfo {pages} {507} (\bibinfo {year} {1986})}\BibitemShut
  {NoStop}%
\bibitem [{\citenamefont {Shor}(1994)}]{Shor-1994}%
  \BibitemOpen
  \bibfield  {author} {\bibinfo {author} {\bibfnamefont {P.~W.}\ \bibnamefont
  {Shor}},\ }\bibfield  {title} {\bibinfo {title} {Algorithms for quantum
  computation: discrete logarithms and factoring},\ }in\ \href
  {https://doi.org/10.1109/sfcs.1994.365700} {\emph {\bibinfo {booktitle}
  {{Proc. 35th Ann. Symp. Found. Comp. Sc.}}}}\ (\bibinfo {year} {1994})\ pp.\
  \bibinfo {pages} {124--134}\BibitemShut {NoStop}%
\bibitem [{\citenamefont {Shor}(1995)}]{PhysRevA.52.R2493}%
  \BibitemOpen
  \bibfield  {author} {\bibinfo {author} {\bibfnamefont {P.~W.}\ \bibnamefont
  {Shor}},\ }\bibfield  {title} {\bibinfo {title} {Scheme for reducing
  decoherence in quantum computer memory},\ }\href
  {https://doi.org/10.1103/PhysRevA.52.R2493} {\bibfield  {journal} {\bibinfo
  {journal} {Phys. Rev. A}\ }\textbf {\bibinfo {volume} {52}},\ \bibinfo
  {pages} {R2493} (\bibinfo {year} {1995})}\BibitemShut {NoStop}%
\bibitem [{\citenamefont {Gottesman}(2009)}]{QEC2}%
  \BibitemOpen
  \bibfield  {author} {\bibinfo {author} {\bibfnamefont {D.}~\bibnamefont
  {Gottesman}},\ }\bibfield  {title} {\bibinfo {title} {An introduction to
  quantum error correction and fault-tolerant quantum computation},\ }\href
  {https://doi.org/10.48550/arXiv.0904.2557} {\bibfield  {journal} {\bibinfo
  {journal} {arXiv:0904.2557}\ } (\bibinfo {year} {2009})}\BibitemShut
  {NoStop}%
\bibitem [{\citenamefont {Campbell}\ \emph {et~al.}(2017)\citenamefont
  {Campbell}, \citenamefont {Terhal},\ and\ \citenamefont {Vuillot}}]{Roads}%
  \BibitemOpen
  \bibfield  {author} {\bibinfo {author} {\bibfnamefont {E.~T.}\ \bibnamefont
  {Campbell}}, \bibinfo {author} {\bibfnamefont {B.~M.}\ \bibnamefont
  {Terhal}},\ and\ \bibinfo {author} {\bibfnamefont {C.}~\bibnamefont
  {Vuillot}},\ }\bibfield  {title} {\bibinfo {title} {Roads towards
  fault-tolerant universal quantum computation},\ }\href
  {https://doi.org/10.1038/nature23460} {\bibfield  {journal} {\bibinfo
  {journal} {Nature}\ }\textbf {\bibinfo {volume} {549}},\ \bibinfo {pages}
  {172} (\bibinfo {year} {2017})}\BibitemShut {NoStop}%
\bibitem [{\citenamefont {Li}\ and\ \citenamefont
  {Benjamin}(2017)}]{PhysRevX.7.021050}%
  \BibitemOpen
  \bibfield  {author} {\bibinfo {author} {\bibfnamefont {Y.}~\bibnamefont
  {Li}}\ and\ \bibinfo {author} {\bibfnamefont {S.~C.}\ \bibnamefont
  {Benjamin}},\ }\bibfield  {title} {\bibinfo {title} {Efficient variational
  quantum simulator incorporating active error minimization},\ }\href
  {https://doi.org/10.1103/PhysRevX.7.021050} {\bibfield  {journal} {\bibinfo
  {journal} {Phys. Rev. X}\ }\textbf {\bibinfo {volume} {7}},\ \bibinfo {pages}
  {021050} (\bibinfo {year} {2017})}\BibitemShut {NoStop}%
\bibitem [{\citenamefont {Temme}\ \emph {et~al.}(2017)\citenamefont {Temme},
  \citenamefont {Bravyi},\ and\ \citenamefont
  {Gambetta}}]{PhysRevLett.119.180509}%
  \BibitemOpen
  \bibfield  {author} {\bibinfo {author} {\bibfnamefont {K.}~\bibnamefont
  {Temme}}, \bibinfo {author} {\bibfnamefont {S.}~\bibnamefont {Bravyi}},\ and\
  \bibinfo {author} {\bibfnamefont {J.~M.}\ \bibnamefont {Gambetta}},\
  }\bibfield  {title} {\bibinfo {title} {Error mitigation for short-depth
  quantum circuits},\ }\href {https://doi.org/10.1103/PhysRevLett.119.180509}
  {\bibfield  {journal} {\bibinfo  {journal} {Phys. Rev. Lett.}\ }\textbf
  {\bibinfo {volume} {119}},\ \bibinfo {pages} {180509} (\bibinfo {year}
  {2017})}\BibitemShut {NoStop}%
\bibitem [{\citenamefont {Endo}\ \emph
  {et~al.}(2018{\natexlab{a}})\citenamefont {Endo}, \citenamefont {Benjamin},\
  and\ \citenamefont {Li}}]{PhysRevX.8.031027}%
  \BibitemOpen
  \bibfield  {author} {\bibinfo {author} {\bibfnamefont {S.}~\bibnamefont
  {Endo}}, \bibinfo {author} {\bibfnamefont {S.~C.}\ \bibnamefont {Benjamin}},\
  and\ \bibinfo {author} {\bibfnamefont {Y.}~\bibnamefont {Li}},\ }\bibfield
  {title} {\bibinfo {title} {Practical quantum error mitigation for near-future
  applications},\ }\href {https://doi.org/10.1103/PhysRevX.8.031027} {\bibfield
   {journal} {\bibinfo  {journal} {Phys. Rev. X}\ }\textbf {\bibinfo {volume}
  {8}},\ \bibinfo {pages} {031027} (\bibinfo {year}
  {2018}{\natexlab{a}})}\BibitemShut {NoStop}%
\bibitem [{\citenamefont {Cai}\ \emph {et~al.}(2022)\citenamefont {Cai},
  \citenamefont {Babbush}, \citenamefont {Benjamin}, \citenamefont {Endo},
  \citenamefont {Huggins}, \citenamefont {Li}, \citenamefont {McClean},\ and\
  \citenamefont {O'Brien}}]{MitigationReview}%
  \BibitemOpen
  \bibfield  {author} {\bibinfo {author} {\bibfnamefont {Z.}~\bibnamefont
  {Cai}}, \bibinfo {author} {\bibfnamefont {R.}~\bibnamefont {Babbush}},
  \bibinfo {author} {\bibfnamefont {S.~C.}\ \bibnamefont {Benjamin}}, \bibinfo
  {author} {\bibfnamefont {S.}~\bibnamefont {Endo}}, \bibinfo {author}
  {\bibfnamefont {W.~J.}\ \bibnamefont {Huggins}}, \bibinfo {author}
  {\bibfnamefont {Y.}~\bibnamefont {Li}}, \bibinfo {author} {\bibfnamefont
  {J.~R.}\ \bibnamefont {McClean}},\ and\ \bibinfo {author} {\bibfnamefont
  {T.~E.}\ \bibnamefont {O'Brien}},\ }\bibfield  {title} {\bibinfo {title}
  {{Quantum Error Mitigation}},\ }\href@noop {} {\bibfield  {journal} {\bibinfo
   {journal} {arXiv:2210.00921}\ } (\bibinfo {year} {2022})}\BibitemShut
  {NoStop}%
\bibitem [{\citenamefont {{van den Berg}}\ \emph {et~al.}(2023)\citenamefont
  {{van den Berg}}, \citenamefont {{Minev}}, \citenamefont {{Kandala}},\ and\
  \citenamefont {{Temme}}}]{IBMexponential}%
  \BibitemOpen
  \bibfield  {author} {\bibinfo {author} {\bibfnamefont {E.}~\bibnamefont {{van
  den Berg}}}, \bibinfo {author} {\bibfnamefont {Z.~K.}\ \bibnamefont
  {{Minev}}}, \bibinfo {author} {\bibfnamefont {A.}~\bibnamefont {{Kandala}}},\
  and\ \bibinfo {author} {\bibfnamefont {K.}~\bibnamefont {{Temme}}},\
  }\bibfield  {title} {\bibinfo {title} {{Probabilistic error cancellation with
  sparse Pauli-Lindblad models on noisy quantum processors}},\ }\href
  {https://doi.org/10.1038/s41567-023-02042-2} {\bibfield  {journal} {\bibinfo
  {journal} {Nature Physics}\ }\textbf {\bibinfo {volume} {19}},\ \bibinfo
  {pages} {1116} (\bibinfo {year} {2023})},\ \Eprint
  {https://arxiv.org/abs/2201.09866} {2201.09866} \BibitemShut {NoStop}%
\bibitem [{\citenamefont {McClean}\ \emph {et~al.}(2016)\citenamefont
  {McClean}, \citenamefont {Romero}, \citenamefont {Babbush},\ and\
  \citenamefont {Aspuru-Guzik}}]{McClean_2016}%
  \BibitemOpen
  \bibfield  {author} {\bibinfo {author} {\bibfnamefont {J.~R.}\ \bibnamefont
  {McClean}}, \bibinfo {author} {\bibfnamefont {J.}~\bibnamefont {Romero}},
  \bibinfo {author} {\bibfnamefont {R.}~\bibnamefont {Babbush}},\ and\ \bibinfo
  {author} {\bibfnamefont {A.}~\bibnamefont {Aspuru-Guzik}},\ }\bibfield
  {title} {\bibinfo {title} {The theory of variational hybrid quantum-classical
  algorithms},\ }\href {https://doi.org/10.1088/1367-2630/18/2/023023}
  {\bibfield  {journal} {\bibinfo  {journal} {New J. Phys.}\ }\textbf {\bibinfo
  {volume} {18}},\ \bibinfo {pages} {023023} (\bibinfo {year}
  {2016})}\BibitemShut {NoStop}%
\bibitem [{\citenamefont {Farhi}\ \emph {et~al.}(2014)\citenamefont {Farhi},
  \citenamefont {Goldstone},\ and\ \citenamefont {Gutmann}}]{QAOA}%
  \BibitemOpen
  \bibfield  {author} {\bibinfo {author} {\bibfnamefont {E.}~\bibnamefont
  {Farhi}}, \bibinfo {author} {\bibfnamefont {J.}~\bibnamefont {Goldstone}},\
  and\ \bibinfo {author} {\bibfnamefont {S.}~\bibnamefont {Gutmann}},\
  }\bibfield  {title} {\bibinfo {title} {A quantum approximate optimization
  algorithm},\ }\href {https://doi.org/10.48550/arXiv.1411.4028} {\bibfield
  {journal} {\bibinfo  {journal} {arXiv:1411.4028}\ } (\bibinfo {year}
  {2014})}\BibitemShut {NoStop}%
\bibitem [{\citenamefont {Huggins}\ \emph {et~al.}(2021)\citenamefont
  {Huggins}, \citenamefont {McArdle}, \citenamefont {O'Brien}, \citenamefont
  {Lee}, \citenamefont {Rubin}, \citenamefont {Boixo}, \citenamefont {Whaley},
  \citenamefont {Babbush},\ and\ \citenamefont {McClean}}]{huggins2021virtual}%
  \BibitemOpen
  \bibfield  {author} {\bibinfo {author} {\bibfnamefont {W.~J.}\ \bibnamefont
  {Huggins}}, \bibinfo {author} {\bibfnamefont {S.}~\bibnamefont {McArdle}},
  \bibinfo {author} {\bibfnamefont {T.~E.}\ \bibnamefont {O'Brien}}, \bibinfo
  {author} {\bibfnamefont {J.}~\bibnamefont {Lee}}, \bibinfo {author}
  {\bibfnamefont {N.~C.}\ \bibnamefont {Rubin}}, \bibinfo {author}
  {\bibfnamefont {S.}~\bibnamefont {Boixo}}, \bibinfo {author} {\bibfnamefont
  {K.~B.}\ \bibnamefont {Whaley}}, \bibinfo {author} {\bibfnamefont
  {R.}~\bibnamefont {Babbush}},\ and\ \bibinfo {author} {\bibfnamefont {J.~R.}\
  \bibnamefont {McClean}},\ }\bibfield  {title} {\bibinfo {title} {Virtual
  distillation for quantum error mitigation},\ }\href
  {https://doi.org/10.1103/PhysRevX.11.041036} {\bibfield  {journal} {\bibinfo
  {journal} {Phys. Rev. X}\ }\textbf {\bibinfo {volume} {11}},\ \bibinfo
  {pages} {041036} (\bibinfo {year} {2021})}\BibitemShut {NoStop}%
\bibitem [{\citenamefont {Koczor}(2022)}]{koczor2021exponential}%
  \BibitemOpen
  \bibfield  {author} {\bibinfo {author} {\bibfnamefont {B.}~\bibnamefont
  {Koczor}},\ }\bibfield  {title} {\bibinfo {title} {Exponential error
  suppression for near-term quantum devices},\ }\href
  {https://doi.org/10.1103/PhysRevX.11.031057} {\bibfield  {journal} {\bibinfo
  {journal} {Phys. Rev. X}\ }\textbf {\bibinfo {volume} {11}},\ \bibinfo
  {pages} {031057} (\bibinfo {year} {2022})}\BibitemShut {NoStop}%
\bibitem [{\citenamefont {Czarnik}\ \emph {et~al.}(2022)\citenamefont
  {Czarnik}, \citenamefont {Arrasmith}, \citenamefont {Coles},\ and\
  \citenamefont {Cincio}}]{czarnik2021error}%
  \BibitemOpen
  \bibfield  {author} {\bibinfo {author} {\bibfnamefont {P.}~\bibnamefont
  {Czarnik}}, \bibinfo {author} {\bibfnamefont {A.}~\bibnamefont {Arrasmith}},
  \bibinfo {author} {\bibfnamefont {P.~J.}\ \bibnamefont {Coles}},\ and\
  \bibinfo {author} {\bibfnamefont {L.}~\bibnamefont {Cincio}},\ }\bibfield
  {title} {\bibinfo {title} {{Error mitigation with Clifford quantum-circuit
  data}},\ }\href {https://doi.org/10.22331/q-2021-11-26-592} {\bibfield
  {journal} {\bibinfo  {journal} {Quantum}\ }\textbf {\bibinfo {volume} {5}},\
  \bibinfo {pages} {592} (\bibinfo {year} {2022})}\BibitemShut {NoStop}%
\bibitem [{\citenamefont {Takagi}\ \emph
  {et~al.}(2022{\natexlab{a}})\citenamefont {Takagi}, \citenamefont {Endo},
  \citenamefont {Minagawa},\ and\ \citenamefont {Gu}}]{TEMG21}%
  \BibitemOpen
  \bibfield  {author} {\bibinfo {author} {\bibfnamefont {R.}~\bibnamefont
  {Takagi}}, \bibinfo {author} {\bibfnamefont {S.}~\bibnamefont {Endo}},
  \bibinfo {author} {\bibfnamefont {S.}~\bibnamefont {Minagawa}},\ and\
  \bibinfo {author} {\bibfnamefont {M.}~\bibnamefont {Gu}},\ }\bibfield
  {title} {\bibinfo {title} {Fundamental limits of quantum error mitigation},\
  }\href {https://doi.org/10.1038/s41534-022-00618-z} {\bibfield  {journal}
  {\bibinfo  {journal} {npj Quant. Inf.}\ }\textbf {\bibinfo {volume} {8}},\
  \bibinfo {pages} {114} (\bibinfo {year} {2022}{\natexlab{a}})}\BibitemShut
  {NoStop}%
\bibitem [{\citenamefont {Takagi}\ \emph
  {et~al.}(2022{\natexlab{b}})\citenamefont {Takagi}, \citenamefont {Tajima},\
  and\ \citenamefont {Gu}}]{TTG22}%
  \BibitemOpen
  \bibfield  {author} {\bibinfo {author} {\bibfnamefont {R.}~\bibnamefont
  {Takagi}}, \bibinfo {author} {\bibfnamefont {H.}~\bibnamefont {Tajima}},\
  and\ \bibinfo {author} {\bibfnamefont {M.}~\bibnamefont {Gu}},\ }\href@noop
  {} {\bibinfo {title} {Universal sample lower bounds for quantum error
  mitigation}} (\bibinfo {year} {2022}{\natexlab{b}}),\ \Eprint
  {https://arxiv.org/abs/2208.09178} {arXiv:2208.09178} \BibitemShut {NoStop}%
\bibitem [{\citenamefont {Tsubouchi}\ \emph {et~al.}(2022)\citenamefont
  {Tsubouchi}, \citenamefont {Sagawa},\ and\ \citenamefont
  {Yoshioka}}]{Fisher22}%
  \BibitemOpen
  \bibfield  {author} {\bibinfo {author} {\bibfnamefont {K.}~\bibnamefont
  {Tsubouchi}}, \bibinfo {author} {\bibfnamefont {T.}~\bibnamefont {Sagawa}},\
  and\ \bibinfo {author} {\bibfnamefont {N.}~\bibnamefont {Yoshioka}},\
  }\href@noop {} {\bibinfo {title} {Universal cost bound of quantum error
  mitigation based on quantum estimation theory}} (\bibinfo {year} {2022}),\
  \Eprint {https://arxiv.org/abs/2208.09385} {arXiv:2208.09385} \BibitemShut
  {NoStop}%
\bibitem [{\citenamefont {Deshpande}\ \emph {et~al.}(2022)\citenamefont
  {Deshpande}, \citenamefont {Niroula}, \citenamefont {Shtanko}, \citenamefont
  {Gorshkov}, \citenamefont {Fefferman},\ and\ \citenamefont
  {Gullans}}]{PRXQuantum.3.040329}%
  \BibitemOpen
  \bibfield  {author} {\bibinfo {author} {\bibfnamefont {A.}~\bibnamefont
  {Deshpande}}, \bibinfo {author} {\bibfnamefont {P.}~\bibnamefont {Niroula}},
  \bibinfo {author} {\bibfnamefont {O.}~\bibnamefont {Shtanko}}, \bibinfo
  {author} {\bibfnamefont {A.~V.}\ \bibnamefont {Gorshkov}}, \bibinfo {author}
  {\bibfnamefont {B.}~\bibnamefont {Fefferman}},\ and\ \bibinfo {author}
  {\bibfnamefont {M.~J.}\ \bibnamefont {Gullans}},\ }\bibfield  {title}
  {\bibinfo {title} {Tight bounds on the convergence of noisy random circuits
  to the uniform distribution},\ }\href
  {https://doi.org/10.1103/PRXQuantum.3.040329} {\bibfield  {journal} {\bibinfo
   {journal} {PRX Quantum}\ }\textbf {\bibinfo {volume} {3}},\ \bibinfo {pages}
  {040329} (\bibinfo {year} {2022})}\BibitemShut {NoStop}%
\bibitem [{\citenamefont {M{\"u}ller-Hermes}\ \emph {et~al.}(2016)\citenamefont
  {M{\"u}ller-Hermes}, \citenamefont {Franca},\ and\ \citenamefont
  {Wolf}}]{MullerHermes}%
  \BibitemOpen
  \bibfield  {author} {\bibinfo {author} {\bibfnamefont {A.}~\bibnamefont
  {M{\"u}ller-Hermes}}, \bibinfo {author} {\bibfnamefont {D.~S.}\ \bibnamefont
  {Franca}},\ and\ \bibinfo {author} {\bibfnamefont {M.~M.}\ \bibnamefont
  {Wolf}},\ }\bibfield  {title} {\bibinfo {title} {Relative entropy convergence
  for depolarizing channels},\ }\href {https://doi.org/10.1063/1.4939560}
  {\bibfield  {journal} {\bibinfo  {journal} {J. Math. Phys.}\ }\textbf
  {\bibinfo {volume} {57}},\ \bibinfo {pages} {2} (\bibinfo {year}
  {2016})}\BibitemShut {NoStop}%
\bibitem [{\citenamefont {Stilck~Fran{\c c}a}\ and\ \citenamefont
  {Garc\'{i}a-Patr\'{o}n}(2021)}]{FG20}%
  \BibitemOpen
  \bibfield  {author} {\bibinfo {author} {\bibfnamefont {D.}~\bibnamefont
  {Stilck~Fran{\c c}a}}\ and\ \bibinfo {author} {\bibfnamefont
  {R.}~\bibnamefont {Garc\'{i}a-Patr\'{o}n}},\ }\bibfield  {title} {\bibinfo
  {title} {Limitations of optimization algorithms on noisy quantum devices},\
  }\href {https://doi.org/10.1038/s41567-021-01356-3} {\bibfield  {journal}
  {\bibinfo  {journal} {Nature Phys.}\ }\textbf {\bibinfo {volume} {17}},\
  \bibinfo {pages} {1221} (\bibinfo {year} {2021})}\BibitemShut {NoStop}%
\bibitem [{\citenamefont {Tsybakov}(2009)}]{Tsybakov_2009}%
  \BibitemOpen
  \bibfield  {author} {\bibinfo {author} {\bibfnamefont {A.~B.}\ \bibnamefont
  {Tsybakov}},\ }\href {https://doi.org/10.1007/b13794} {\emph {\bibinfo
  {title} {{Introduction to non-parametric estimation}}}}\ (\bibinfo
  {publisher} {Springer New York},\ \bibinfo {year} {2009})\BibitemShut
  {NoStop}%
\bibitem [{\citenamefont {De~Palma}\ \emph {et~al.}(2022)\citenamefont
  {De~Palma}, \citenamefont {Marvian}, \citenamefont {Rouzé},\ and\
  \citenamefont {Stilck~Franca}}]{Wasserstein_variational_22}%
  \BibitemOpen
  \bibfield  {author} {\bibinfo {author} {\bibfnamefont {G.}~\bibnamefont
  {De~Palma}}, \bibinfo {author} {\bibfnamefont {M.}~\bibnamefont {Marvian}},
  \bibinfo {author} {\bibfnamefont {C.}~\bibnamefont {Rouzé}},\ and\ \bibinfo
  {author} {\bibfnamefont {D.}~\bibnamefont {Stilck~Franca}},\ }\href@noop {}
  {\bibinfo {title} {Limitations of variational quantum algorithms: a quantum
  optimal transport approach}} (\bibinfo {year} {2022}),\ \Eprint
  {https://arxiv.org/abs/2204.03455} {arXiv:2204.03455} \BibitemShut {NoStop}%
\bibitem [{\citenamefont {Wang}\ \emph
  {et~al.}(2021{\natexlab{a}})\citenamefont {Wang}, \citenamefont {Czarnik},
  \citenamefont {Arrasmith}, \citenamefont {Cerezo}, \citenamefont {Cincio},\
  and\ \citenamefont {Coles}}]{LosAlamos_EM}%
  \BibitemOpen
  \bibfield  {author} {\bibinfo {author} {\bibfnamefont {S.}~\bibnamefont
  {Wang}}, \bibinfo {author} {\bibfnamefont {P.}~\bibnamefont {Czarnik}},
  \bibinfo {author} {\bibfnamefont {A.}~\bibnamefont {Arrasmith}}, \bibinfo
  {author} {\bibfnamefont {M.}~\bibnamefont {Cerezo}}, \bibinfo {author}
  {\bibfnamefont {L.}~\bibnamefont {Cincio}},\ and\ \bibinfo {author}
  {\bibfnamefont {P.~J.}\ \bibnamefont {Coles}},\ }\href@noop {} {\bibinfo
  {title} {Can error mitigation improve trainability of noisy variational
  quantum algorithms?}} (\bibinfo {year} {2021}{\natexlab{a}}),\ \Eprint
  {https://arxiv.org/abs/2109.01051} {arXiv:2109.01051} \BibitemShut {NoStop}%
\bibitem [{\citenamefont {Thanasilp}\ \emph {et~al.}(2022)\citenamefont
  {Thanasilp}, \citenamefont {Wang}, \citenamefont {Cerezo},\ and\
  \citenamefont {Holmes}}]{LosAlamos_Kernels}%
  \BibitemOpen
  \bibfield  {author} {\bibinfo {author} {\bibfnamefont {S.}~\bibnamefont
  {Thanasilp}}, \bibinfo {author} {\bibfnamefont {S.}~\bibnamefont {Wang}},
  \bibinfo {author} {\bibfnamefont {M.}~\bibnamefont {Cerezo}},\ and\ \bibinfo
  {author} {\bibfnamefont {Z.}~\bibnamefont {Holmes}},\ }\bibfield  {title}
  {\bibinfo {title} {Exponential concentration and untrainability in quantum
  kernel methods},\ }\href {https://doi.org/10.48550/arXiv.2208.11060}
  {\bibfield  {journal} {\bibinfo  {journal} {arXiv:2208.11060}\ } (\bibinfo
  {year} {2022})}\BibitemShut {NoStop}%
\bibitem [{\citenamefont {Wang}\ \emph
  {et~al.}(2021{\natexlab{b}})\citenamefont {Wang}, \citenamefont {Fontana},
  \citenamefont {Cerezo}, \citenamefont {Sharma}, \citenamefont {Sone},
  \citenamefont {Cincio},\ and\ \citenamefont {Coles}}]{Wang2021}%
  \BibitemOpen
  \bibfield  {author} {\bibinfo {author} {\bibfnamefont {S.}~\bibnamefont
  {Wang}}, \bibinfo {author} {\bibfnamefont {E.}~\bibnamefont {Fontana}},
  \bibinfo {author} {\bibfnamefont {M.}~\bibnamefont {Cerezo}}, \bibinfo
  {author} {\bibfnamefont {K.}~\bibnamefont {Sharma}}, \bibinfo {author}
  {\bibfnamefont {A.}~\bibnamefont {Sone}}, \bibinfo {author} {\bibfnamefont
  {L.}~\bibnamefont {Cincio}},\ and\ \bibinfo {author} {\bibfnamefont {P.~J.}\
  \bibnamefont {Coles}},\ }\bibfield  {title} {\bibinfo {title} {Noise-induced
  barren plateaus in variational quantum algorithms},\ }\href
  {https://doi.org/10.1038/s41467-021-27045-6} {\bibfield  {journal} {\bibinfo
  {journal} {Nature Comm.}\ }\textbf {\bibinfo {volume} {12}},\ \bibinfo
  {pages} {6961} (\bibinfo {year} {2021}{\natexlab{b}})}\BibitemShut {NoStop}%
\bibitem [{\citenamefont {Deshpande}\ \emph {et~al.}(2021)\citenamefont
  {Deshpande}, \citenamefont {Niroula}, \citenamefont {Shtanko}, \citenamefont
  {Gorshkov}, \citenamefont {Fefferman},\ and\ \citenamefont
  {Gullans}}]{Abhinav_21}%
  \BibitemOpen
  \bibfield  {author} {\bibinfo {author} {\bibfnamefont {A.}~\bibnamefont
  {Deshpande}}, \bibinfo {author} {\bibfnamefont {P.}~\bibnamefont {Niroula}},
  \bibinfo {author} {\bibfnamefont {O.}~\bibnamefont {Shtanko}}, \bibinfo
  {author} {\bibfnamefont {A.~V.}\ \bibnamefont {Gorshkov}}, \bibinfo {author}
  {\bibfnamefont {B.}~\bibnamefont {Fefferman}},\ and\ \bibinfo {author}
  {\bibfnamefont {M.~J.}\ \bibnamefont {Gullans}},\ }\href@noop {} {\bibinfo
  {title} {Tight bounds on the convergence of noisy random circuits to the
  uniform distribution}} (\bibinfo {year} {2021}),\ \Eprint
  {https://arxiv.org/abs/2112.00716} {arXiv:2112.00716} \BibitemShut {NoStop}%
\bibitem [{\citenamefont {Cleve}\ \emph {et~al.}(2016)\citenamefont {Cleve},
  \citenamefont {Leung}, \citenamefont {Liu},\ and\ \citenamefont
  {Wang}}]{Cleve16}%
  \BibitemOpen
  \bibfield  {author} {\bibinfo {author} {\bibfnamefont {R.}~\bibnamefont
  {Cleve}}, \bibinfo {author} {\bibfnamefont {D.}~\bibnamefont {Leung}},
  \bibinfo {author} {\bibfnamefont {L.}~\bibnamefont {Liu}},\ and\ \bibinfo
  {author} {\bibfnamefont {C.}~\bibnamefont {Wang}},\ }\bibfield  {title}
  {\bibinfo {title} {Near-linear constructions of exact unitary 2-designs},\
  }\href {https://doi.org/10.26421/QIC16.9-10-1} {\bibfield  {journal}
  {\bibinfo  {journal} {Quant. Inf. Comp.}\ }\textbf {\bibinfo {volume} {16}},\
  \bibinfo {pages} {721–756} (\bibinfo {year} {2016})}\BibitemShut {NoStop}%
\bibitem [{\citenamefont {Reyzin}(2020)}]{Reyzin}%
  \BibitemOpen
  \bibfield  {author} {\bibinfo {author} {\bibfnamefont {L.}~\bibnamefont
  {Reyzin}},\ }\href@noop {} {\bibinfo {title} {Statistical queries and
  statistical algorithms: Foundations and applications}} (\bibinfo {year}
  {2020}),\ \Eprint {https://arxiv.org/abs/2004.00557} {arXiv:2004.00557}
  \BibitemShut {NoStop}%
\bibitem [{\citenamefont {Yatracos}(1985)}]{Yatracos85}%
  \BibitemOpen
  \bibfield  {author} {\bibinfo {author} {\bibfnamefont {Y.~G.}\ \bibnamefont
  {Yatracos}},\ }\bibfield  {title} {\bibinfo {title} {{Rates of convergence of
  minimum distance estimators and Kolmogorov's entropy}},\ }\href
  {https://doi.org/10.1214/aos/1176349553} {\bibfield  {journal} {\bibinfo
  {journal} {Ann. Stat.}\ }\textbf {\bibinfo {volume} {13}},\ \bibinfo {pages}
  {768 } (\bibinfo {year} {1985})}\BibitemShut {NoStop}%
\bibitem [{\citenamefont {Blum}\ \emph {et~al.}(1994)\citenamefont {Blum},
  \citenamefont {Furst}, \citenamefont {Jackson}, \citenamefont {Kearns},
  \citenamefont {Mansour},\ and\ \citenamefont {Rudich}}]{BFJKMR94}%
  \BibitemOpen
  \bibfield  {author} {\bibinfo {author} {\bibfnamefont {A.}~\bibnamefont
  {Blum}}, \bibinfo {author} {\bibfnamefont {M.}~\bibnamefont {Furst}},
  \bibinfo {author} {\bibfnamefont {J.}~\bibnamefont {Jackson}}, \bibinfo
  {author} {\bibfnamefont {M.}~\bibnamefont {Kearns}}, \bibinfo {author}
  {\bibfnamefont {Y.}~\bibnamefont {Mansour}},\ and\ \bibinfo {author}
  {\bibfnamefont {S.}~\bibnamefont {Rudich}},\ }\bibfield  {title} {\bibinfo
  {title} {{Weakly learning DNF and characterizing statistical query learning
  using Fourier analysis}},\ }in\ \href {https://doi.org/10.1145/195058.195147}
  {\emph {\bibinfo {booktitle} {Proceedings of the Twenty-Sixth Annual ACM
  Symposium on Theory of Computing}}},\ \bibinfo {series and number} {STOC
  '94}\ (\bibinfo  {publisher} {Association for Computing Machinery},\ \bibinfo
  {address} {New York, NY, USA},\ \bibinfo {year} {1994})\ p.\ \bibinfo {pages}
  {253–262}\BibitemShut {NoStop}%
\bibitem [{\citenamefont {Fran\ifmmode~\mbox{\c{c}}\else \c{c}\fi{}a}\ \emph
  {et~al.}(2021)\citenamefont {Fran\ifmmode~\mbox{\c{c}}\else \c{c}\fi{}a},
  \citenamefont {Strelchuk},\ and\ \citenamefont {Studzi\ifmmode~\acute{n}\else
  \'{n}\fi{}ski}}]{PhysRevLett.126.210502}%
  \BibitemOpen
  \bibfield  {author} {\bibinfo {author} {\bibfnamefont {D.~S.}\ \bibnamefont
  {Fran\ifmmode~\mbox{\c{c}}\else \c{c}\fi{}a}}, \bibinfo {author}
  {\bibfnamefont {S.}~\bibnamefont {Strelchuk}},\ and\ \bibinfo {author}
  {\bibfnamefont {M.}~\bibnamefont {Studzi\ifmmode~\acute{n}\else
  \'{n}\fi{}ski}},\ }\bibfield  {title} {\bibinfo {title} {{Efficient classical
  simulation and benchmarking of quantum processes in the Weyl basis}},\ }\href
  {https://doi.org/10.1103/PhysRevLett.126.210502} {\bibfield  {journal}
  {\bibinfo  {journal} {Phys. Rev. Lett.}\ }\textbf {\bibinfo {volume} {126}},\
  \bibinfo {pages} {210502} (\bibinfo {year} {2021})}\BibitemShut {NoStop}%
\bibitem [{\citenamefont {Rall}\ \emph {et~al.}(2019)\citenamefont {Rall},
  \citenamefont {Liang}, \citenamefont {Cook},\ and\ \citenamefont
  {Kretschmer}}]{PhysRevA.99.062337}%
  \BibitemOpen
  \bibfield  {author} {\bibinfo {author} {\bibfnamefont {P.}~\bibnamefont
  {Rall}}, \bibinfo {author} {\bibfnamefont {D.}~\bibnamefont {Liang}},
  \bibinfo {author} {\bibfnamefont {J.}~\bibnamefont {Cook}},\ and\ \bibinfo
  {author} {\bibfnamefont {W.}~\bibnamefont {Kretschmer}},\ }\bibfield  {title}
  {\bibinfo {title} {{Simulation of qubit quantum circuits via Pauli
  propagation}},\ }\href {https://doi.org/10.1103/PhysRevA.99.062337}
  {\bibfield  {journal} {\bibinfo  {journal} {Phys. Rev. A}\ }\textbf {\bibinfo
  {volume} {99}},\ \bibinfo {pages} {062337} (\bibinfo {year}
  {2019})}\BibitemShut {NoStop}%
\bibitem [{\citenamefont {Bravyi}\ \emph {et~al.}(2020)\citenamefont {Bravyi},
  \citenamefont {Kliesch}, \citenamefont {Koenig},\ and\ \citenamefont
  {Tang}}]{PhysRevLett.125.260505}%
  \BibitemOpen
  \bibfield  {author} {\bibinfo {author} {\bibfnamefont {S.}~\bibnamefont
  {Bravyi}}, \bibinfo {author} {\bibfnamefont {A.}~\bibnamefont {Kliesch}},
  \bibinfo {author} {\bibfnamefont {R.}~\bibnamefont {Koenig}},\ and\ \bibinfo
  {author} {\bibfnamefont {E.}~\bibnamefont {Tang}},\ }\bibfield  {title}
  {\bibinfo {title} {Obstacles to variational quantum optimization from
  symmetry protection},\ }\href
  {https://doi.org/10.1103/PhysRevLett.125.260505} {\bibfield  {journal}
  {\bibinfo  {journal} {Phys. Rev. Lett.}\ }\textbf {\bibinfo {volume} {125}},\
  \bibinfo {pages} {260505} (\bibinfo {year} {2020})}\BibitemShut {NoStop}%
\bibitem [{\citenamefont {Farhi}\ \emph {et~al.}(2020)\citenamefont {Farhi},
  \citenamefont {Gamarnik},\ and\ \citenamefont {Gutmann}}]{whole_graph}%
  \BibitemOpen
  \bibfield  {author} {\bibinfo {author} {\bibfnamefont {E.}~\bibnamefont
  {Farhi}}, \bibinfo {author} {\bibfnamefont {D.}~\bibnamefont {Gamarnik}},\
  and\ \bibinfo {author} {\bibfnamefont {S.}~\bibnamefont {Gutmann}},\
  }\bibfield  {title} {\bibinfo {title} {The quantum approximate optimization
  algorithm needs to see the whole graph: Worst case examples},\ }\href
  {https://doi.org/10.48550/arXiv.2005.08747} {\bibfield  {journal} {\bibinfo
  {journal} {arXiv:2005.08747}\ } (\bibinfo {year} {2020})}\BibitemShut
  {NoStop}%
\bibitem [{\citenamefont {Eldar}\ and\ \citenamefont {Harrow}(2017)}]{8104078}%
  \BibitemOpen
  \bibfield  {author} {\bibinfo {author} {\bibfnamefont {L.}~\bibnamefont
  {Eldar}}\ and\ \bibinfo {author} {\bibfnamefont {A.~W.}\ \bibnamefont
  {Harrow}},\ }\bibfield  {title} {\bibinfo {title} {{Local Hamiltonians whose
  ground states are hard to approximate}},\ }in\ \href
  {https://doi.org/10.1109/FOCS.2017.46} {\emph {\bibinfo {booktitle} {2017
  IEEE 58th Annual Symposium on Foundations of Computer Science (FOCS)}}}\
  (\bibinfo {year} {2017})\ pp.\ \bibinfo {pages} {427--438}\BibitemShut
  {NoStop}%
\bibitem [{\citenamefont {Anshu}\ \emph {et~al.}(2023)\citenamefont {Anshu},
  \citenamefont {Breuckmann},\ and\ \citenamefont {Nirkhe}}]{nlts_conjecture}%
  \BibitemOpen
  \bibfield  {author} {\bibinfo {author} {\bibfnamefont {A.}~\bibnamefont
  {Anshu}}, \bibinfo {author} {\bibfnamefont {N.~P.}\ \bibnamefont
  {Breuckmann}},\ and\ \bibinfo {author} {\bibfnamefont {C.}~\bibnamefont
  {Nirkhe}},\ }\bibfield  {title} {\bibinfo {title} {{NLTS Hamiltonians from
  Good Quantum Codes}},\ }in\ \href {https://doi.org/10.1145/3564246.3585114}
  {\emph {\bibinfo {booktitle} {{Proceedings of the 55th Annual ACM Symposium
  on Theory of Computing}}}},\ \bibinfo {series and number} {{STOC 2023}}\
  (\bibinfo  {publisher} {Association for Computing Machinery},\ \bibinfo
  {address} {New York, NY, USA},\ \bibinfo {year} {2023})\ pp.\ \bibinfo
  {pages} {1090--1096},\ \Eprint {https://arxiv.org/abs/2206.13228}
  {2206.13228} \BibitemShut {NoStop}%
\bibitem [{\citenamefont {Gonz\'alez-Garc{\i}a}\ \emph
  {et~al.}(2022)\citenamefont {Gonz\'alez-Garc{\i}a}, \citenamefont {Trivedi},\
  and\ \citenamefont {Cirac}}]{CiracNoise}%
  \BibitemOpen
  \bibfield  {author} {\bibinfo {author} {\bibfnamefont {G.}~\bibnamefont
  {Gonz\'alez-Garc{\i}a}}, \bibinfo {author} {\bibfnamefont {R.}~\bibnamefont
  {Trivedi}},\ and\ \bibinfo {author} {\bibfnamefont {J.~I.}\ \bibnamefont
  {Cirac}},\ }\bibfield  {title} {\bibinfo {title} {{Error propagation in NISQ
  devices for solving classical optimization problems}},\ }\href
  {https://doi.org/10.1103/PRXQuantum.3.040326} {\bibfield  {journal} {\bibinfo
   {journal} {PRX Quantum}\ }\textbf {\bibinfo {volume} {3}},\ \bibinfo {pages}
  {040326} (\bibinfo {year} {2022})},\ \Eprint
  {https://arxiv.org/abs/2203.15632} {2203.15632} \BibitemShut {NoStop}%
\bibitem [{\citenamefont {Endo}\ \emph
  {et~al.}(2018{\natexlab{b}})\citenamefont {Endo}, \citenamefont {Benjamin},\
  and\ \citenamefont {Li}}]{endo2018practical}%
  \BibitemOpen
  \bibfield  {author} {\bibinfo {author} {\bibfnamefont {S.}~\bibnamefont
  {Endo}}, \bibinfo {author} {\bibfnamefont {S.~C.}\ \bibnamefont {Benjamin}},\
  and\ \bibinfo {author} {\bibfnamefont {Y.}~\bibnamefont {Li}},\ }\bibfield
  {title} {\bibinfo {title} {Practical quantum error mitigation for near-future
  applications},\ }\href {https://doi.org/10.1103/PhysRevX.8.031027} {\bibfield
   {journal} {\bibinfo  {journal} {Phys. Rev. X}\ }\textbf {\bibinfo {volume}
  {8}},\ \bibinfo {pages} {031027} (\bibinfo {year}
  {2018}{\natexlab{b}})}\BibitemShut {NoStop}%
\bibitem [{\citenamefont {Mosonyi}\ and\ \citenamefont
  {Hiai}(2011)}]{MosonyiHiai}%
  \BibitemOpen
  \bibfield  {author} {\bibinfo {author} {\bibfnamefont {M.}~\bibnamefont
  {Mosonyi}}\ and\ \bibinfo {author} {\bibfnamefont {F.}~\bibnamefont {Hiai}},\
  }\bibfield  {title} {\bibinfo {title} {{On the quantum Rényi relative
  entropies and related capacity formulas}},\ }\href
  {https://doi.org/10.1109/TIT.2011.2110050} {\bibfield  {journal} {\bibinfo
  {journal} {IEEE Trans. Inf. Th.}\ }\textbf {\bibinfo {volume} {57}},\
  \bibinfo {pages} {2474} (\bibinfo {year} {2011})}\BibitemShut {NoStop}%
\bibitem [{\citenamefont {Dankert}\ \emph {et~al.}(2009)\citenamefont
  {Dankert}, \citenamefont {Cleve}, \citenamefont {Emerson},\ and\
  \citenamefont {Livine}}]{DCEL09}%
  \BibitemOpen
  \bibfield  {author} {\bibinfo {author} {\bibfnamefont {C.}~\bibnamefont
  {Dankert}}, \bibinfo {author} {\bibfnamefont {R.}~\bibnamefont {Cleve}},
  \bibinfo {author} {\bibfnamefont {J.}~\bibnamefont {Emerson}},\ and\ \bibinfo
  {author} {\bibfnamefont {E.}~\bibnamefont {Livine}},\ }\bibfield  {title}
  {\bibinfo {title} {Exact and approximate unitary 2-designs and their
  application to fidelity estimation},\ }\href
  {https://doi.org/10.1103/PhysRevA.80.012304} {\bibfield  {journal} {\bibinfo
  {journal} {Phys. Rev. A}\ }\textbf {\bibinfo {volume} {80}},\ \bibinfo
  {pages} {012304} (\bibinfo {year} {2009})}\BibitemShut {NoStop}%
\bibitem [{\citenamefont {Gottesman}(1997)}]{Gottesman_thesis}%
  \BibitemOpen
  \bibfield  {author} {\bibinfo {author} {\bibfnamefont {D.}~\bibnamefont
  {Gottesman}},\ }\bibfield  {title} {\bibinfo {title} {{Stabilizer codes and
  quantum error correction}},\ }\href
  {https://doi.org/10.48550/arXiv.quant-ph/9705052} {\bibfield  {journal}
  {\bibinfo  {journal} {arXiv:quant-ph/9705052}\ } (\bibinfo {year}
  {1997})}\BibitemShut {NoStop}%
\bibitem [{\citenamefont {Watrous}(2018)}]{Wat18_book}%
  \BibitemOpen
  \bibfield  {author} {\bibinfo {author} {\bibfnamefont {J.}~\bibnamefont
  {Watrous}},\ }\href {https://doi.org/10.1017/9781316848142} {\emph {\bibinfo
  {title} {{The theory of quantum information}}}}\ (\bibinfo  {publisher}
  {Cambridge University Press},\ \bibinfo {year} {2018})\BibitemShut {NoStop}%
\bibitem [{\citenamefont {Cerezo}\ \emph {et~al.}(2021)\citenamefont {Cerezo},
  \citenamefont {Arrasmith}, \citenamefont {Babbush}, \citenamefont {Benjamin},
  \citenamefont {Endo}, \citenamefont {Fujii}, \citenamefont {{McClean}},
  \citenamefont {Mitarai}, \citenamefont {Yuan}, \citenamefont {Cincio},\ and\
  \citenamefont {Coles}}]{cerezo2021variational}%
  \BibitemOpen
  \bibfield  {author} {\bibinfo {author} {\bibfnamefont {M.}~\bibnamefont
  {Cerezo}}, \bibinfo {author} {\bibfnamefont {A.}~\bibnamefont {Arrasmith}},
  \bibinfo {author} {\bibfnamefont {R.}~\bibnamefont {Babbush}}, \bibinfo
  {author} {\bibfnamefont {S.~C.}\ \bibnamefont {Benjamin}}, \bibinfo {author}
  {\bibfnamefont {S.}~\bibnamefont {Endo}}, \bibinfo {author} {\bibfnamefont
  {K.}~\bibnamefont {Fujii}}, \bibinfo {author} {\bibfnamefont {J.~R.}\
  \bibnamefont {{McClean}}}, \bibinfo {author} {\bibfnamefont {K.}~\bibnamefont
  {Mitarai}}, \bibinfo {author} {\bibfnamefont {X.}~\bibnamefont {Yuan}},
  \bibinfo {author} {\bibfnamefont {L.}~\bibnamefont {Cincio}},\ and\ \bibinfo
  {author} {\bibfnamefont {P.~J.}\ \bibnamefont {Coles}},\ }\bibfield  {title}
  {\bibinfo {title} {Variational quantum algorithms},\ }\href
  {https://doi.org/10.1038/s42254-021-00348-9} {\bibfield  {journal} {\bibinfo
  {journal} {Nature Rev. Phys.}\ }\textbf {\bibinfo {volume} {3}},\ \bibinfo
  {pages} {625} (\bibinfo {year} {2021})}\BibitemShut {NoStop}%
\bibitem [{\citenamefont {Bharti}\ \emph {et~al.}(2022)\citenamefont {Bharti},
  \citenamefont {Cervera-Lierta}, \citenamefont {Kyaw}, \citenamefont {Haug},
  \citenamefont {Alperin-Lea}, \citenamefont {Anand}, \citenamefont {Degroote},
  \citenamefont {Heimonen}, \citenamefont {Kottmann}, \citenamefont {Menke},
  \citenamefont {Mok}, \citenamefont {Sim}, \citenamefont {Kwek},\ and\
  \citenamefont {Aspuru-Guzik}}]{bharti2022noisy}%
  \BibitemOpen
  \bibfield  {author} {\bibinfo {author} {\bibfnamefont {K.}~\bibnamefont
  {Bharti}}, \bibinfo {author} {\bibfnamefont {A.}~\bibnamefont
  {Cervera-Lierta}}, \bibinfo {author} {\bibfnamefont {T.~H.}\ \bibnamefont
  {Kyaw}}, \bibinfo {author} {\bibfnamefont {T.}~\bibnamefont {Haug}}, \bibinfo
  {author} {\bibfnamefont {S.}~\bibnamefont {Alperin-Lea}}, \bibinfo {author}
  {\bibfnamefont {A.}~\bibnamefont {Anand}}, \bibinfo {author} {\bibfnamefont
  {M.}~\bibnamefont {Degroote}}, \bibinfo {author} {\bibfnamefont
  {H.}~\bibnamefont {Heimonen}}, \bibinfo {author} {\bibfnamefont {J.~S.}\
  \bibnamefont {Kottmann}}, \bibinfo {author} {\bibfnamefont {T.}~\bibnamefont
  {Menke}}, \bibinfo {author} {\bibfnamefont {W.-K.}\ \bibnamefont {Mok}},
  \bibinfo {author} {\bibfnamefont {S.}~\bibnamefont {Sim}}, \bibinfo {author}
  {\bibfnamefont {L.-C.}\ \bibnamefont {Kwek}},\ and\ \bibinfo {author}
  {\bibfnamefont {A.}~\bibnamefont {Aspuru-Guzik}},\ }\bibfield  {title}
  {\bibinfo {title} {Noisy intermediate-scale quantum algorithms},\ }\href
  {https://doi.org/10.1103/RevModPhys.94.015004} {\bibfield  {journal}
  {\bibinfo  {journal} {Rev. Mod. Phys.}\ }\textbf {\bibinfo {volume} {94}},\
  \bibinfo {pages} {015004} (\bibinfo {year} {2022})}\BibitemShut {NoStop}%
\bibitem [{\citenamefont {Kearns}(1998)}]{KearnsSQ}%
  \BibitemOpen
  \bibfield  {author} {\bibinfo {author} {\bibfnamefont {M.}~\bibnamefont
  {Kearns}},\ }\bibfield  {title} {\bibinfo {title} {Efficient noise-tolerant
  learning from statistical queries},\ }\href
  {https://doi.org/10.1145/293347.293351} {\bibfield  {journal} {\bibinfo
  {journal} {J. ACM}\ }\textbf {\bibinfo {volume} {45}},\ \bibinfo {pages}
  {983–1006} (\bibinfo {year} {1998})}\BibitemShut {NoStop}%
\bibitem [{\citenamefont {Hinsche}\ \emph {et~al.}(2023)\citenamefont
  {Hinsche}, \citenamefont {Ioannou}, \citenamefont {Nietner}, \citenamefont
  {Haferkamp}, \citenamefont {Quek}, \citenamefont {Hangleiter}, \citenamefont
  {Seifert}, \citenamefont {Eisert},\ and\ \citenamefont
  {Sweke}}]{ouramazingTgatepaper}%
  \BibitemOpen
  \bibfield  {author} {\bibinfo {author} {\bibfnamefont {M.}~\bibnamefont
  {Hinsche}}, \bibinfo {author} {\bibfnamefont {M.}~\bibnamefont {Ioannou}},
  \bibinfo {author} {\bibfnamefont {A.}~\bibnamefont {Nietner}}, \bibinfo
  {author} {\bibfnamefont {J.}~\bibnamefont {Haferkamp}}, \bibinfo {author}
  {\bibfnamefont {Y.}~\bibnamefont {Quek}}, \bibinfo {author} {\bibfnamefont
  {D.}~\bibnamefont {Hangleiter}}, \bibinfo {author} {\bibfnamefont {J.-P.}\
  \bibnamefont {Seifert}}, \bibinfo {author} {\bibfnamefont {J.}~\bibnamefont
  {Eisert}},\ and\ \bibinfo {author} {\bibfnamefont {R.}~\bibnamefont
  {Sweke}},\ }\bibfield  {title} {\bibinfo {title} {{A single $T$-gate makes
  distribution learning hard}},\ }\href
  {https://doi.org/10.1103/PhysRevLett.130.240602} {\bibfield  {journal}
  {\bibinfo  {journal} {Phys. Rev. Lett.}\ }\textbf {\bibinfo {volume} {130}},\
  \bibinfo {pages} {240602} (\bibinfo {year} {2023})}\BibitemShut {NoStop}%
\bibitem [{\citenamefont {Arunachalam}\ \emph {et~al.}(2020)\citenamefont
  {Arunachalam}, \citenamefont {Grilo},\ and\ \citenamefont {Yuen}}]{AGY20}%
  \BibitemOpen
  \bibfield  {author} {\bibinfo {author} {\bibfnamefont {S.}~\bibnamefont
  {Arunachalam}}, \bibinfo {author} {\bibfnamefont {A.~B.}\ \bibnamefont
  {Grilo}},\ and\ \bibinfo {author} {\bibfnamefont {H.}~\bibnamefont {Yuen}},\
  }\bibfield  {title} {\bibinfo {title} {Quantum statistical query learning},\
  }\href {https://doi.org/10.48550/arXiv.2002.08240} {\bibfield  {journal}
  {\bibinfo  {journal} {arXiv:2002.08240}\ } (\bibinfo {year}
  {2020})}\BibitemShut {NoStop}%
\bibitem [{\citenamefont {Fran{\c c}a}\ and\ \citenamefont
  {{Garcia-Patron}}(2022)}]{francaGameQuantumAdvantage2020}%
  \BibitemOpen
  \bibfield  {author} {\bibinfo {author} {\bibfnamefont {D.~S.}\ \bibnamefont
  {Fran{\c c}a}}\ and\ \bibinfo {author} {\bibfnamefont {R.}~\bibnamefont
  {{Garcia-Patron}}},\ }\bibfield  {title} {\bibinfo {title} {A game of quantum
  advantage: Linking verification and simulation},\ }\href
  {https://doi.org/10.48550/arXiv.2011.12173} {\bibfield  {journal} {\bibinfo
  {journal} {Quantum}\ }\textbf {\bibinfo {volume} {6}},\ \bibinfo {pages}
  {753} (\bibinfo {year} {2022})}\BibitemShut {NoStop}%
\end{thebibliography}
